\documentclass[letterpaper,11pt]{article}

\usepackage{graphicx}

\usepackage[margin=1in]{geometry}
\usepackage{amsmath,amsthm,amsfonts,amssymb,bm,bbm}
\usepackage{latexsym}

\usepackage[linesnumbered,ruled,vlined]{algorithm2e}

\usepackage[textsize=scriptsize]{todonotes}
\usepackage[dvipsnames]{xcolor}
\usepackage{mathtools}
\usepackage[hidelinks]{hyperref}
\usepackage{graphicx}
\usepackage{authblk}
\usepackage{array}
\usepackage{multirow}
\usepackage[capitalize,noabbrev]{cleveref}
\usepackage{aliascnt}
\usepackage{lineno}
\usepackage{booktabs}
\usepackage{tabularx}

\usepackage[square,numbers]{natbib}
\usepackage{subcaption}
\usepackage{tikz}
\usetikzlibrary{positioning,calc,arrows.meta,patterns.meta}

\theoremstyle{definition}
\newtheorem{definition}{Definition}[section]
\newtheorem{example}{Example}[section]

\theoremstyle{plain}
\newtheorem{theorem}{Theorem}[section]
\newtheorem{lemma}{Lemma}[section]

\newaliascnt{corollary}{theorem}
\newtheorem{corollary}[corollary]{Corollary}
\aliascntresetthe{corollary}

\crefname{appendix}{Appendix}{Appendices}
\Crefname{appendix}{Appendix}{Appendices}

\newcommand{\ba}{\bm{a}}
\newcommand{\bx}{\bm{x}}
\newcommand{\by}{\bm{y}}
\newcommand{\bR}{\mathbb{R}}
\newcommand{\bZ}{\mathbb{Z}}
\newcommand{\cA}{\mathcal{A}}
\newcommand{\cF}{\mathcal{F}}
\newcommand{\cR}{\mathcal{R}}
\newcommand{\cS}{\mathcal{S}}
\newcommand{\cX}{\mathcal{X}}

\newcommand{\PURE}{\textsc{Pure-Circuit}}
\newcommand{\NOT}{\textsc{Not}}
\newcommand{\AND}{\textsc{And}}
\newcommand{\PURIFY}{\textsc{Purify}}
\newcommand{\BIAS}{\textsc{Bias}}

\newcommand{\val}[1]{\boldsymbol{\mathrm{x}}[#1]}
\newcommand{\valonly}{\boldsymbol{\mathrm{x}}}

\newcommand{\eps}{\ensuremath{\varepsilon}}

\DeclareMathOperator{\Ind}{\mathbbm{1}}
\DeclareMathOperator{\Exp}{\mathbb{E}}
\DeclareMathOperator{\Prob}{\mathrm{Prob}}

\newcommand{\PrimP}{\mathsf{P}}   
\newcommand{\DeP}{\mathsf{De}}    
\newcommand{\BlP}{\mathsf{Bl}}    
\newcommand{\GuP}{\mathsf{Gu}}    
\newcommand{\PrP}{\mathsf{Pr}}    
\newcommand{\BiP}{\mathsf{Bi}}    
\newcommand{\BitB}{\mathsf{B}}    
\newcommand{\AuxB}{\mathsf{A}}    
\newcommand{\AuxP}{\mathsf{Q}}    
\newcommand{\DumP}{\mathsf{D}}    

\newcommand{\itemwithtable}[2]{%
\item[]%
\parbox[t]{\linewidth}{%
  \begin{minipage}[t]{0.55\textwidth}\vspace{0pt}#1\end{minipage}\hfill
  \begin{minipage}[t]{0.4\textwidth}\vspace{0pt}\centering #2\end{minipage}%
}%
}

\title{The Complexity of Multiplayer Colonel Blotto Games with Player-Specific Values}
\author{Martin Bichler, Abheek Ghosh}
\affil{Technical University of Munich\\
  \texttt{\{m.bichler,abheek.ghosh\}@tum.de}
}
\date{}

\begin{document}


\maketitle


\begin{abstract}

We study equilibrium computation in discrete multiplayer Colonel Blotto games with player-specific battlefield values.
In the two-player model with common battlefield values, equilibria can be computed in polynomial time.
We show that this tractability breaks down in the multiplayer model with player-specific values under the standard uniform tie-breaking rule.
In particular, computing a $(c/n)$-approximate Nash equilibrium is PPAD-hard for some constant $c>0$, even when every player has three resources, where $n$ is the number of players.
The main technical step is PPAD-hardness for computing a constant-approximate well-supported Nash equilibrium.
In contrast, under uniform tie-breaking, a pure Nash equilibrium can be computed in polynomial time when every player has one resource.
We also prove PPAD membership for computing $\eps$-approximate Nash equilibria for inverse-exponentially small $\eps$.
Finally, for non-uniform monotone tie-breaking, we show PPAD-hardness even when every player has one resource and all players have identical battlefield values.


\end{abstract}


\thispagestyle{empty}

\clearpage

\setcounter{page}{1}
\resetlinenumber[1]

\section{Introduction}

The Colonel Blotto game is a classical model of competition with budget constraints.
Players have limited resources and distribute them across multiple battlefields.
On each battlefield, the player who allocates the most resources wins,
with ties resolved by a tie-breaking rule.
The payoff of a player is the sum of the values of the battlefields that it wins.
The model originates in Borel's work on games of strategy and the subsequent formulation of Gross and Wagner~\cite{borel1921theorie,borel1938applications,gross1950continuous}.
It has since been used to model electoral competition~\cite{myerson1993incentives,laslier2002distributive,klumpp2019dynamics,boix2021multiplayer}, advertising allocation~\cite{friedman1958game}, security allocation~\cite{schwartz2014heterogeneous,ferdowsi2017colonel,afiouni2026colonel}, research and development competition~\cite{ewerhart2021class,kovenock2021generalizations}, and other forms of strategic resource allocation~\cite{golman2009general}.

The algorithmic literature primarily focuses on the two-player constant-sum discrete model (see, e.g., \cite{ahmadinejad2019duels,behnezhad2023fast,beaglehole2023sampling,kontogiannis2025efficient}).
In the discrete model, each player has an integer budget, and a pure strategy is an allocation of this integer budget across battlefields.
This model has two features that make it interesting from an algorithmic point of view.
On the one hand, the strategy space is large: each player has exponentially many ways to distribute its discrete budget across the battlefields.
On the other hand, the payoff structure is simple.
Payoffs are additive over battlefields, and the payoff rule is local to each battlefield.
The battlefields are related only through the players’ budget constraints.

In the two-player constant-sum case, this structure can be exploited algorithmically.
Although mixed strategies are distributions over an exponentially large set of pure allocations, their payoff-relevant information can be represented compactly, and equilibria can be computed in polynomial time~\cite{ahmadinejad2019duels,behnezhad2023fast}.
These algorithms leverage the constant-sum property to construct an efficiently solvable linear programming formulation.

Many motivating applications of Blotto games are not inherently two-player or constant-sum.
In political competition, several candidates or parties may allocate campaign resources across districts~\cite{myerson1993incentives,laslier2002distributive,boix2021multiplayer}.
In advertising and research and development, several firms may compete across market segments, technologies, or projects~\cite{friedman1958game,kovenock2021generalizations,ewerhart2021class}.
Moreover, the same battlefield may matter differently to different players: a district, market segment, security target, or technology may be more valuable for one competitor than another~\cite{roberson2012non,schwartz2014heterogeneous,kovenock2021generalizations,vu2019approximate,kim2018lottery,liu2024multiplayer}.
This motivates the study of discrete Blotto games with multiple players and player-specific battlefield values.

These generalizations preserve the local and additive payoff structure of Blotto games.
A player's payoff is still the sum of its payoffs from individual battlefields, and the outcome on each battlefield is still determined only by the resources allocated to that battlefield.
However, the game is no longer a two-player constant-sum game, so the known polynomial-time algorithms do not directly apply.
This raises the main question studied in this paper: can equilibria still be computed efficiently in discrete multiplayer Blotto games with player-specific battlefield values?

Our main result is a negative answer to this question.
Even under the standard uniform tie-breaking rule, where tied winners split a battlefield equally, equilibrium computation becomes PPAD-hard once we allow multiple players with player-specific battlefield values.
More precisely, we prove that it is PPAD-hard to compute a constant-approximate well-supported Nash equilibrium (WSNE), even when every player has only three resources.
This also gives us hardness for approximate Nash equilibrium (NE): computing an $(\eps/n)$-approximate NE is PPAD-hard for some constant $\eps > 0$, where $n$ is the number of players.
Our hardness reduction is from the PPAD-complete \PURE{} problem~\cite{deligkas2024pure}.
In a natural restricted-action variant, where players may allocate resources only to battlefields for which they have positive value, the same construction gives PPAD-hardness for constant-approximate NE.

The three-resource restriction is worth emphasizing.
With a constant number of resources, each player has only polynomially many pure strategies, so the hardness is not due to any player having an exponentially large set of pure strategies.
At the same time, the result is close to a tractable case.
If every player has a single resource, then each pure strategy is simply to choose one battlefield.
Under uniform tie-breaking, this case reduces to a singleton congestion game with player-specific payoff functions, and hence a pure NE can be computed in polynomial time~\cite{milchtaich1996congestion,ackermann2009pure}.
This leaves the two-resource case as a natural open problem.

The main challenge in proving the hardness result is the limited interaction among the players in the game.
Two players interact only through their allocations on individual battlefields.
Even within a battlefield, the payoff structure is very simple: a player can only win, tie, or lose.
Further, under uniform tie-breaking, the players interact anonymously.
A player loses a battlefield if another player allocates strictly more resources to it, irrespective of which player does so; and tied winners split the battlefield equally.
Consequently, hardness results for standard game models with richer payoff interactions do not directly transfer to this setting.
In the hardness proof, the construction has to create the required incentives using only the local win/tie/lose comparisons, additive battlefield values, and the players' global budget constraints.

The one-resource tractability result above does not extend to general non-uniform tie-breaking rules.
The literature primarily focuses on uniform tie-breaking, but some applications may break ties asymmetrically~\cite{vu2021colonel}.
A battlefield may favor one player in a tie because of incumbency, priority, brand recognition, local advantage, or other exogenous factors.
This motivates us to consider non-uniform monotone tie-breaking rules, where a player's share cannot increase as more players tie for the highest allocation.
We show that computing a constant-approximate NE under such rules is PPAD-hard even when every player has one resource and all players have identical battlefield values.\footnote{
Notice that the hardness result for non-uniform tie-breaking is for constant-approximate NE, whereas the corresponding result for uniform tie-breaking is for inverse-linear-approximate NE.
}
Thus, with non-uniform tie-breaking, the hardness comes from the tie-breaking rule, rather than from player-specific values, large budgets, or a large action space.

We complement the hardness results by proving PPAD membership.
For uniform tie-breaking, the problems of computing $\eps$-NE and $\eps$-WSNE are in PPAD for inverse-exponential $\eps$.
The proof uses compactly represented mixed strategies~\cite{behnezhad2023fast}, along with a strong separation oracle for feasible strategies and polynomial-size arithmetic circuits for computing utilities and gradients, and applies PPAD-membership results for concave games~\cite{papadimitriou2023computational}.
Together with the hardness results, this gives PPAD-completeness for approximate equilibrium computation in multiplayer Blotto games with player-specific values and uniform tie-breaking.
For non-uniform tie-breaking, the same membership argument applies whenever the tie-breaking rule admits polynomial-size circuits for the relevant expected payoff computations; for arbitrary succinct non-uniform rules, this payoff computation may itself be hard.

The technical overview (\cref{sec:overview}), following the preliminaries (\cref{sec:prelim}), gives the main ideas behind the reductions and the PPAD-membership result.
We then present the formal constructions and proofs in the subsequent sections and conclude with some open problems.

\subsection{Related Work}
\label{sec:related-work}

The classical model originates in Borel's work~\cite{borel1921theorie,borel1938applications} and was further developed by Gross and Wagner~\cite{gross1950continuous}.
Roberson~\cite{roberson2006colonel} analyzed equilibria in the continuous-budget two-player constant-sum model, and Hart~\cite{hart2008discrete} studied discrete Colonel Blotto and General Lotto models.
Subsequent work has studied heterogeneous battlefield values, asymmetric budgets, incomplete information, and related General Lotto models~\cite{weinstein2012two,roberson2012non,thomas2018n,kovenock2021generalizations,adamo2009blotto,kovenock2011blotto,paarporn2019characterizing,paarporn2021general,paarporn2023strategically,paarporn2024incomplete}.
These works focus on equilibrium characterization or existence in analytically tractable regimes.
Our focus, on the other hand, is on the equilibrium computation problem in discrete multiplayer Blotto games.

The algorithmic literature has mainly studied two-player constant-sum Blotto models and learning in Blotto and other structured games.
The authors in~\cite{ahmadinejad2019duels} gave a polynomial-time algorithm for computing equilibria in the discrete two-player constant-sum model, and the authors in~\cite{behnezhad2023fast} gave a faster algorithm using a compact LP formulation.
For continuous-budget variants, Perchet, Rigollet, and Le Gouic~\cite{perchet2024algorithmic} gave algorithmic results using multimarginal couplings in certain constant-sum two-player settings.
Other work studies no-regret learning and online and bandit variants of Blotto games~\cite{vu2018efficient,vu2019combinatorial,vu2020path,beaglehole2023sampling,leon2024online,kontogiannis2025efficient}.
In particular, no-regret methods can give efficient computation of approximate coarse correlated equilibria in structured Blotto-type games~\cite{beaglehole2023sampling,kontogiannis2025efficient}.
In contrast, our work studies approximate Nash equilibrium computation in multiplayer non-constant-sum Blotto games.

Outside the two-player constant-sum setting,
Roberson and Kvasov~\cite{roberson2012non}, Schwartz, Loiseau, and Sastry~\cite{schwartz2014heterogeneous}, and Vu, Loiseau, and Silva~\cite{vu2019approximate} study variants with player-specific objectives.
Boix-Adser{\`a}, Edelman, and Jayanti~\cite{boix2021multiplayer} and Jayanti~\cite{jayanti2021nash} study multiplayer Blotto models, while the authors of~\cite{liu2024multiplayer} study a multiplayer General Lotto model.
These works primarily study equilibrium existence and structure.
We study the complexity of Nash equilibrium computation in the discrete multiplayer model.

Our non-uniform tie-breaking model is related to work on asymmetric rules for deciding battlefield winners.
Vu and Loiseau~\cite{vu2021colonel} study Blotto games with favoritism, where players may have different effectiveness on the same battlefield.
Their focus is on equilibrium structure rather than computational complexity.
In our model, the asymmetry appears only when players tie with the highest allocation and receive different shares of the battlefield.
We show that even this form of asymmetry can make equilibrium computation hard in otherwise restricted settings.

Our model also has connections to congestion games.
With one resource and uniform tie-breaking, a Blotto game reduces to a singleton congestion game with player-specific payoff functions: each player chooses one battlefield, and its payoff from that battlefield depends on the number of players choosing it.
Milchtaich~\cite{milchtaich1996congestion} showed that such games admit pure Nash equilibria, and Ackermann, R{\"o}glin, and V{\"o}cking~\cite{ackermann2009pure} gave polynomial-time algorithms for pure equilibria in player-specific matroid congestion games, which include singleton games.
For larger budgets, Blotto games are related to weighted and splittable congestion models, but neither connection is exact.
A player's budget is analogous to a weight, but unlike in weighted singleton congestion games, the budget can be split across battlefields.
This resembles splittable congestion games, where players divide demand across resources.
However, the payoff rule is different from congestion games: Blotto payoffs are determined by which player allocates the most resources to a battlefield, rather than by a load-based cost or payoff function.

Our membership proof is related to compact representations used in algorithmic work on Blotto games.
\cite{ahmadinejad2019duels} uses marginal representations, and \cite{behnezhad2023fast} uses a flow representation over a layered graph.
In the two-player constant-sum setting, these representations lead to polynomial-time equilibrium computation via linear programming.
In our multiplayer setting, the game is not a two-player constant-sum game, so the compact representation does not give such a linear program.
Instead, we use it to obtain polynomial-size circuits for utilities and gradients, and a strong separation oracle for the strategy space.
We then apply the PPAD-membership framework for concave games~\cite{papadimitriou2023computational}.

We use tools from the complexity theory of Nash equilibria in our reductions.
The class PPAD was introduced by Papadimitriou~\cite{papadimitriou1994complexity}, and computing a Nash equilibrium in a two-player normal-form game is PPAD-complete~\cite{daskalakis2009complexity,chen2009settling}.
Our reductions start from the \PURE{} problem of Deligkas, Fearnley, Hollender, and Melissourgos~\cite{deligkas2024pure}, a PPAD-complete circuit problem.
The structure of the \PURE{} problem enables us to encode it using local interactions between players and battlefields in the Blotto model.

\section{Model and Preliminaries} \label{sec:prelim}
Let $[n] = \{ 1, 2, \ldots, n \}$ for positive integer $n \in \bZ_{>0}$.
Let $\Delta(S) = \{ \bx \in \bR_{\ge 0}^S \mid \sum_{s \in S} x_s = 1 \}$ be the probability simplex of dimension $|S|-1$ for a finite set $S$.
Let $\Ind[\cdot]$ be the indicator function that outputs $1$ if the input to the function is true and $0$ otherwise.

In its full generality, the Colonel Blotto game, or simply the Blotto game, we consider has
$n$ players and $m$ battlefields.
Player $i \in [n]$ has a resource budget of $B_i \in \bZ_{> 0}$ that they distribute across the $m$ battlefields; the resources are discrete.
A pure strategy $\ba_i = (a_{ij})_{j \in [m]} \in \bZ_{\ge 0}^m$ of player $i$ denotes how many resources $i$ allocates to each battlefield; $\ba_i$ satisfies the budget constraint $\sum_{j \in [m]} a_{ij} = B_i$. Let $\ba = (\ba_i)_{i \in [n]}$. Let $\cA_i$ be the set of pure strategies for player $i$ and $\cA = \prod_i \cA_i$.
Player $i$ has a value $v_{ij} \in \bR_{\ge 0}$ for winning battlefield $j$.
A Blotto game has player-specific values if $v_{ij} \neq v_{i'j}$ for some players $i, i' \in [n]$ and battlefield $j \in [m]$.
The player(s) who allocate the most resources to a battlefield win the battlefield and get the corresponding payoff, and everyone else gets zero payoff from the battlefield.
In particular, if $a_{ij} > \max_{i' \neq i} a_{i'j}$, then $i$ gets a payoff of $v_{ij}$ from battlefield $j$ and everyone else gets zero.
In case of a tie, we assume there is a tie-breaking rule for each battlefield, as discussed below.
The total payoff of a player is the sum of payoffs from the individual battlefields.

We assume that every battlefield $j \in [m]$ uses a tie-breaking rule denoted by $\tau_j$.
Let $S_j(\ba) \subseteq [n]$ be the set of players who allocate the most resources on battlefield $j$, i.e., $S_j(\ba) = \{ i \mid a_{ij} = \max_{i'} a_{i'j} \}$.
Then $\tau_j(S_j(\ba), i) \ge 0$ denotes player $i$'s probability for winning battlefield $j$.
Note that $\tau_j(S_j(\ba), i) = 0$ if $i \notin S_j(\ba)$ and $\sum_{i \in S_j(\ba)} \tau_j(S_j(\ba), i) = 1$.
We assume monotonicity of the tie-breaking rule: if $S \subseteq T$, then $\tau_j(S, i) \ge \tau_j(T, i)$ for all $i \in S$, i.e., the winning probability of a tied player does not increase as the set of tied highest allocators grows.
The literature primarily focuses on uniform tie-breaking, i.e., $\tau_j(S_j(\ba), i) = 1/|S_j(\ba)|$ for all $i \in S_j(\ba)$, and we too will primarily focus on this.

Given a pure strategy profile $\ba = (a_{ij})_{i \in [n], j \in [m]} \in \cA$ and the corresponding set of players who allocate the most resources in each battlefield $S_j(\ba) = \{ i \mid a_{ij} = \max_{i'} a_{i'j} \}$, the utility of player $i$ can be written formally as
\[
    u_i(\ba) = \sum_{j \in [m]} v_{ij} \, \tau_j(S_j(\ba), i).
\]
Let $\bx = (\bx_i)_{i \in [n]} \in \cX = \prod_i \Delta(\cA_i)$ be a mixed strategy profile, where $\bx_i \in \cX_i = \Delta(\cA_i)$ is the mixed strategy of player $i$.
We denote the probability assigned to the pure strategy $\ba_i \in \cA_i$ in the mixed strategy $\bx_i$ by $\bx_i(\ba_i)$.
The expected utility of player $i$ in the mixed profile $\bx$ is
\[
    u_i(\bx) = \sum_{\ba \in \cA} u_i(\ba) \prod_{i' \in [n]} \bx_{i'}(\ba_{i'}).
\]

\paragraph{Nash equilibrium and approximations.}
Let $\bx_{-i} = (\bx_1, \ldots, \bx_{i-1}, \bx_{i+1}, \ldots, \bx_n)$. A mixed strategy profile $\bx$ is a (mixed) Nash equilibrium (NE) if no player can gain by unilaterally deviating:
\[
    u_i(\bx) \ge u_i(\bx_i', \bx_{-i}), \qquad\text{for every $i \in [n]$ and $\bx_i' \in \cX_i$.}
\]

For $\eps \ge 0$, we say $\bx$ is an $\eps$-Nash equilibrium ($\eps$-NE) if no player can improve her expected payoff by more than $\eps$ by a unilateral deviation. Formally,
\[
    u_i(\bx) + \eps \ge u_i(\bx_i', \bx_{-i}), \qquad\text{for every $i \in [n]$ and $\bx_i' \in \cX_i$.}
\]
Notice that the $\eps$ approximation in the $\eps$-NE definition above is additive. Hence, to make the definition meaningful, we assume that the utilities are normalized to be between $0$ and $1$, i.e., $u_i(\bx) \in [0, 1]$ for all $\bx \in \cX$.

A stronger approximation notion is that of an $\eps$-well-supported Nash equilibrium ($\eps$-WSNE): here every pure strategy that is played with positive probability must be an $\eps$-best-response. Formally, $\bx$ is an $\eps$-WSNE if
\begin{align*}
    u_i(\ba_i, \bx_{-i}) + \eps \ge u_i(\bx_i', \bx_{-i}), \qquad\text{for every $i \in [n]$, $\ba_i \in \cA_i$ such that $\bx_i(\ba_i) > 0$, and $\bx_i' \in \cX_i$.}
\end{align*}
Again, we assume the utilities are normalized to $[0, 1]$.

Although every $\eps$-WSNE is also an $\eps$-NE, the converse is not true in general.
However, an approximate NE can be converted to an approximate WSNE with a polynomial loss in the approximation parameter.
The conversion removes from each player's support the pure strategies whose payoff is sufficiently below the best-response payoff and then renormalizes the remaining probabilities.
The next lemma states this result formally; the proof is in \cref{app:proof:wsne2ne}.

\begin{lemma}\label{lm:wsne2ne}
Consider an $n$-player finite game with utilities normalized to lie in $[0,1]$.
For every $\delta \in (0,1]$, given a $(\delta^2/8n)$-NE, removing from each player's support every pure strategy whose payoff is more than $\delta/2$ below the best-response payoff, and then renormalizing, gives a $\delta$-WSNE.
\end{lemma}

The conversion in \cref{lm:wsne2ne} can be implemented in polynomial time whenever the starting profile has polynomial support and the payoffs of supported pure strategies and the best-response values can be computed in polynomial time.

Notice that the Blotto game defined earlier is finite---it has finitely many players with each player having finitely many actions---and hence always has an NE by Nash's theorem.

\paragraph{Computational problem.}
In this paper, we study the $\eps$-NE and $\eps$-WSNE computation problems in Blotto games.
The computational problems take as input a discrete Blotto game together with an approximation parameter $\eps$.
The game input contains
the number of players $n$,
the number of battlefields $m$,
the values $(v_{ij})_{i \in [n], j \in [m]}$,
and the budgets $(B_i)_{i \in [n]}$.
For uniform tie-breaking, no additional input is needed.
For non-uniform tie-breaking, the input additionally contains a compact representation of the tie-breaking rule.

The values are nonnegative rational numbers, represented in binary by their numerators and denominators.
The approximation parameter $\eps$ is also represented in binary.
The budgets are positive integers represented in unary.
The reason for using unary rather than binary encoding of budgets is that the natural compact representations of mixed strategies have size polynomial in the budget values $B_i$ but not in $\log(B_i)$.
For example, even the marginal distribution of player $i$ on a single battlefield may require specifying probabilities over $\{0, 1, \ldots, B_i\}$.
Previous papers that study discrete Blotto games also use this convention~\cite{ahmadinejad2019duels,behnezhad2023fast,beaglehole2023sampling,kontogiannis2025efficient}.
The input representation does not affect our PPAD-hardness results.
We use budgets of at most $3$, values that are rational numbers with constant-size numerators and denominators, and approximation parameters that are constant for approximate WSNE and up to inverse-linear in $n$ for approximate NE.
Thus, unary and binary encodings of the budgets, values, and approximation parameter are polynomially equivalent for the hardness proofs.
The input representation becomes crucial in the PPAD-membership results and is further discussed in \cref{sec:membership}.
The representation of non-uniform tie-breaking rules is also discussed there.

The output is a mixed strategy profile of the players.
Naively, each player $i$ has $\binom{m + B_i - 1}{B_i}$ pure actions, so explicitly listing probabilities for all pure actions is too large.
However, as will be discussed in \cref{sec:membership}, equivalent mixed strategies can be compactly represented in size polynomial in the input (assuming unary encoding of budgets).
For the hardness results, since budgets are at most $3$, even the explicit representation of mixed strategies is polynomially large in the input size.

\paragraph{Pure-Circuit.}
PPAD-hardness reductions start from a known PPAD-hard problem, which in our case will be the \PURE{} problem described below, and reduce it to the problems at hand, which in our case are the approximate NE and WSNE computation problems in Blotto games.
We use a variant of \PURE{} with \NOT{}, \AND{}, and \PURIFY{} gates, which is also PPAD-hard. Next, we formally define the \PURE{} problem.

\begin{definition}[\PURE{} \cite{deligkas2024pure}] \label{def:pure}
We are given a circuit that may have cycles.
The circuit has \NOT{}, \AND{}, and \PURIFY{} gates, described below. The input/output variables of each gate take values in the interval $[0,1]$. Each variable must be the output of \emph{exactly} one gate.
We can also view the circuit as a directed graph with the nodes corresponding to the variables and the directed edges corresponding to the dependency of one variable on another variable through a gate (directed edges point from the inputs to the outputs of the gates).
Let $V$ denote the set of nodes (variables) of this graph. For a node $v \in V$, let $\val{v} \in [0,1]$ denote the value of $v$.
Next, we describe the behavior of the three gates:
\begin{itemize}
    \itemwithtable{
        \NOT{}: Takes an input, say $u$, and gives an output, say $v$. If $\val{u} \in \{0, 1\}$, then $\val{v} = 1 - \val{u}$. If $\val{u} \in (0, 1)$, then $\val{v}$ is allowed to take any value in $[0, 1]$.
    }{
        \begin{tabular}{c || c}
            $u$ & $v$ \\
            \hline
            $0$ & $1$ \\
            $1$ & 0 \\
            $(0,1)$ & $[0, 1]$ \\
        \end{tabular}
    }

    \itemwithtable{
        \AND{}: Takes two inputs, say $u$ and $v$, and gives an output, say $w$. If either input value $\val{u}$ or $\val{v}$ is $0$, then the output $\val{w} = 0$. If both $\val{u}$ and $\val{v}$ are $1$, then the output $\val{w} = 1$. Otherwise, $\val{w}$ is allowed to take any value in $[0, 1]$.
    }{
        \begin{tabular}{c | c || c}
            $u$ & $v$ & $w$ \\
            \hline
            $0$ & $[0, 1]$ & 0 \\
            $[0, 1]$ & 0 & 0 \\
            $1$ & $1$ & $1$ \\
            \multicolumn{2}{c || }{otherwise} & $[0, 1]$ \\
        \end{tabular}
    }

    \itemwithtable{
        \PURIFY{}: Takes one input, say $u$, and gives two outputs, say $v$ and $w$. If the input $\val{u} \in \{0, 1\}$, then the outputs $\val{v} = \val{w} = \val{u}$ are the same as the input.
        If $\val{u} \in (0,1)$, then at least one of the output variables must take a value in $\{0,1\}$, while the other output can take any value in $[0,1]$.
    }{
        \begin{tabular}{c || >{\centering\arraybackslash}p{1cm} | >{\centering\arraybackslash}p{1cm}}
            $u$ & $v$ & $w$ \\
            \hline
            $0$ & 0 & 0 \\
            $1$ & $1$ & $1$ \\
            \multirow{2}{*}{$(0, 1)$} & \multicolumn{2}{c}{\multirow{2}{2.8cm}{at least one output in $\{0, 1\}$}} \\
            & \multicolumn{2}{c}{} \\
        \end{tabular}
    }
\end{itemize}
The \PURE{} problem is to find an assignment of all the variables in the range $[0,1]$ such that the variables are consistent with the gates in the circuit. This computational problem has been proven to be PPAD-hard~\cite{deligkas2024pure}.
Further, we assume that every variable is used \emph{exactly} once as an input to a gate (and exactly once as an output of a gate, as already required in the definition of \PURE{}), and that no variable is both an input and an output of the same gate. These assumptions are without loss of generality for PPAD-hardness~\cite[Corollary 2.3]{deligkas2024pure}.
\end{definition}

\section{Technical Overview}
\label{sec:overview}

We now give an overview of the main ideas behind the hardness and membership results.
We first discuss the hardness proof for uniform tie-breaking.
The reduction is from \PURE{} (\cref{def:pure}) and constructs a Blotto game in which every constant-approximate WSNE encodes a solution to the \PURE{} instance.
The hardness result holds even when every player has at most three resources, and the construction can also be modified so that every player has exactly three resources (\cref{thm:hardness-unif-wsne}).
For approximate NE, the same construction gives hardness for approximation parameter $\Theta(1/n)$ with unrestricted action space (\cref{cor:ne}), and for a constant approximation parameter in the restricted-action variant where players may allocate resources only to battlefields for which they have positive value (\cref{cor:ne2}).
On the other hand, when all players have unit budgets, the game reduces to a singleton congestion game with player-specific payoff functions, and a pure NE can be computed in polynomial time~\cite{milchtaich1996congestion,ackermann2009pure}.
We then discuss the non-uniform tie-breaking result, where hardness holds even with unit budgets and identical values (\cref{thm:nonunif}).
Finally, we explain the compact representation used for PPAD membership.
The main point is that explicitly listing a mixed strategy over pure allocations may be too large.
Its payoff-relevant information can instead be represented by a polynomial-size flow, which is sufficient for computing payoffs under uniform tie-breaking.
This allows us to apply a PPAD-membership theorem for concave games (\cref{thm:concave-games-so}, \cite{papadimitriou2023computational}) and obtain membership for approximate NE and WSNE under uniform tie-breaking (\cref{thm:membership-ne,cor:membership-wsne}).
For non-uniform tie-breaking, the same argument applies when the relevant payoff computations have polynomial-size circuits.

\subsection{Hardness}
\label{sec:overview:hardness}

The hardness proof for uniform tie-breaking reduces from \PURE{}.
Given a \PURE{} instance, we construct a Blotto game whose approximate WSNE gives an assignment to all variables of the \PURE{} instance.
The construction has one primary player for each variable, and several auxiliary players for each gate.
The auxiliary players are called detectors ($\DeP$), blockers ($\BlP$), guards ($\GuP$), pressure players ($\PrP$), or bias players ($\BiP$) depending upon their behavior.
The primary players encode the values of the variables, while the auxiliary players implement the \NOT{}, \AND{}, and \PURIFY{} gates.

The proof is first carried out for a slightly modified version of uniform tie-breaking.
If no player allocates any resource to a battlefield, then every player receives zero payoff from that battlefield.
This avoids a small issue caused by empty battlefields.
The payoffs are also left unnormalized during the construction.
Both changes are only for convenience.
At the end of the proof, we return to the standard uniform tie-breaking convention and normalize the payoffs, losing only a constant factor in the approximation parameter.

There are three types of battlefields: bit battlefields, auxiliary battlefields, and reserve battlefields.
The bit battlefields help encode the values of the variables.
The auxiliary battlefields are used inside the gate gadgets.
The reserve battlefields give outside options to players and are used to keep players away from battlefields for which they have zero value.

\paragraph{Primary players.}
For every variable $v$ of the \PURE{} instance, the Blotto game has a primary player $\PrimP^v$ and two bit battlefields $\BitB^v_0$ and $\BitB^v_1$.
The primary player $\PrimP^v$ has three resources and has value $1$ for each of these two bit battlefields.
It has value zero for all auxiliary battlefields and for all bit battlefields that do not belong to $v$; its values for reserve battlefields are described below.
The intended behavior is that $\PrimP^v$ puts one resource on each of $\BitB^v_0$ and $\BitB^v_1$, and then chooses where to put the third resource.
Thus the two intended actions are the allocations $(2,1)$ and $(1,2)$ over $(\BitB^v_0,\BitB^v_1)$.
Let $p^v_s$ be the probability that $\PrimP^v$ puts at least two resources on $\BitB^v_s$.
We will show that $\PrimP^v$ uses only the two intended actions, and hence $p^v_0 + p^v_1 = 1$.
For a small constant $\alpha>0$, we set $\val{v}=0$ when $p^v_1\leq \alpha$, set $\val{v}=1$ when $p^v_1\geq 1-\alpha$, and linearly map the interval $[\alpha,1-\alpha]$ to $[0,1]$.
In particular, if $p^v_0$ is close to $1$, then $\val{v}=0$, and if $p^v_1$ is close to $1$, then $\val{v}=1$.

\paragraph{Auxiliary players.}
Consider a gate, one of its input variables $u$, and a bit value $s \in \{0,1\}$.
The construction has a small group of auxiliary players that checks whether $\PrimP^u$ puts two resources on $\BitB^u_s$ with high probability.
This group is designed to react to the case $\val{u}=s$.
It has two auxiliary battlefields, which we call the first and second auxiliary battlefield of the group.
It also has a detector, a blocker, a guard, and some pressure players.
The \PURIFY{} gate also has bias players.

The detector has one resource and has positive value for $\BitB^u_s$ and for the first auxiliary battlefield of its group.
The blocker has two resources and has positive value only for the two auxiliary battlefields of its group.
The guard has one resource and has positive value only for the second auxiliary battlefield of its group.
Each pressure player has one resource and has positive value for the second auxiliary battlefield of its group and for one bit battlefield of an output variable of the same gate.
Here an output variable means a variable produced by this gate.
We call this bit battlefield the selected output bit battlefield of the pressure player.
Each bias player has one resource and has positive value for one bit battlefield of an output variable of the \PURIFY{} gate.
Thus, the only auxiliary players that can allocate resources to bit battlefields are detectors, pressure players, and bias players, and all of them have one resource.
Blockers and guards have zero value for every bit battlefield.
Also, each bit battlefield is used by only constantly many auxiliary players in the whole construction.

Each player has a few private reserve battlefields with small positive value, and has value zero for the reserve battlefields of other players.
Every value not specified in the construction is zero.
The reserve battlefields are used to prove that no player puts positive probability on an action that allocates a resource to a battlefield for which the player has zero value.

\cref{fig:uniform-not-gate,fig:uniform-and-gate,fig:uniform-purify-gate} show the structure of the three gates in representative cases.
They are useful to keep in mind while reading the next paragraphs.
Each group has an input bit battlefield, two auxiliary battlefields, and some players that may later move to an output bit battlefield.
The arrows show the movements that happen in the depicted case.
The exact payoff comparisons are proved later; here the figures should be read as a guide to the direction in which the incentives move players.

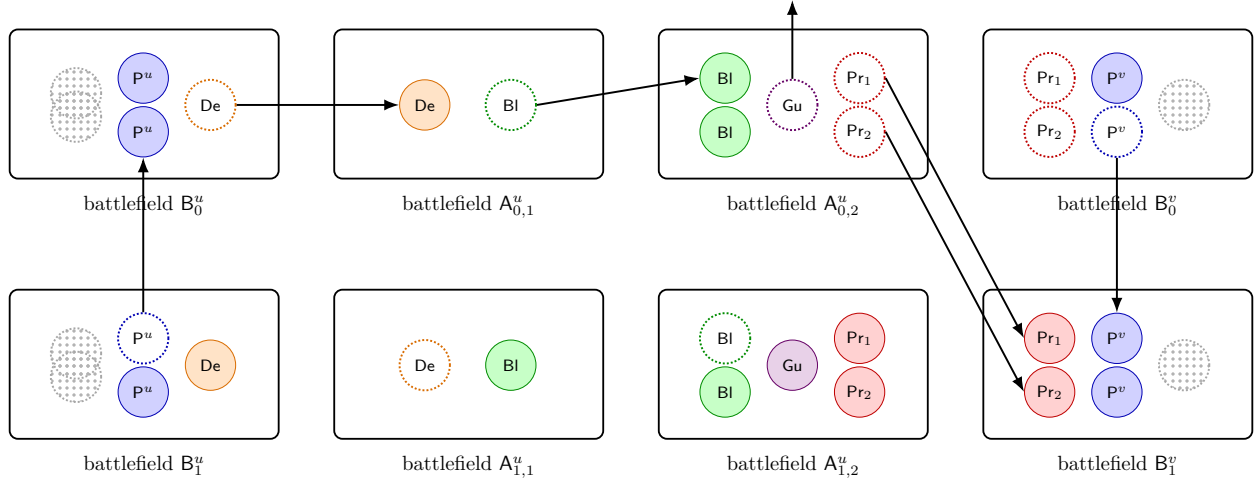
\begin{figure}[t]
\centering
\resizebox{\linewidth}{!}{\begingroup%
\def\ngBoxW{4.6}
\def\ngBoxH{2.55}
\def\ngXSep{5.55}
\def\ngYSep{4.45}
\def\ngColA{0.25}%
\def\ngColB{0.50}%
\def\ngColC{0.75}%
\def\ngColL{0.34}%
\def\ngColR{0.66}%
\def\ngRowU{0.68}%
\def\ngRowM{0.50}%
\def\ngRowD{0.32}%
\def\ngOutsideSep{0.20}%
\newcommand{\ngPt}[3]{($ (#1.south west) + (#2*\ngBoxW cm,#3*\ngBoxH cm) $)}%
\newcommand{\ngOutsidePair}[3]{%
    \node[externalsmall] (#1a) at ($(#2.south west)+(#3*\ngBoxW cm,\ngRowM*\ngBoxH cm)+(0,\ngOutsideSep)$) {};
    \node[externalsmall] (#1b) at ($(#2.south west)+(#3*\ngBoxW cm,\ngRowM*\ngBoxH cm)+(0,-\ngOutsideSep)$) {};
}%
\newcommand{\ngOutsideSingle}[3]{%
    \node[externalsmall] (#1) at ($(#2.south west)+(#3*\ngBoxW cm,\ngRowM*\ngBoxH cm)$) {};
}%
\begin{tikzpicture}[
    >=Latex,
    every node/.style={font=\small},
    battlefield/.style={
        draw,
        rounded corners=4pt,
        line width=0.9pt,
        minimum width=\ngBoxW cm,
        minimum height=\ngBoxH cm
    },
    boxlabel/.style={font=\small},
    move/.style={->, line width=0.95pt},
    tok/.style={
        circle,
        minimum size=8.6mm,
        inner sep=0pt,
        font=\scriptsize
    },
    primary/.style={tok, draw=blue!70!black, fill=blue!18},
    detector/.style={tok, draw=orange!85!black, fill=orange!22},
    blocker/.style={tok, draw=green!55!black, fill=green!22},
    guard/.style={tok, draw=violet!80!black, fill=violet!18},
    pressure/.style={tok, draw=red!75!black, fill=red!18},
    ghost/.style={densely dotted, line width=0.95pt, fill=white},
    externalsmall/.style={
        circle,
        draw=gray!70,
        densely dotted,
        line width=0.95pt,
        fill=gray!8,
        pattern={Dots[distance=0.95mm, radius=0.20mm]},
        pattern color=gray!60,
        minimum size=8.6mm,
        inner sep=0pt
    }
]

\node[battlefield] (Bu0) at (0,0) {};
\node[battlefield] (A01) at (\ngXSep,0) {};
\node[battlefield] (A02) at (2*\ngXSep,0) {};
\node[battlefield] (Bv0) at (3*\ngXSep,0) {};

\node[battlefield] (Bu1) at (0,-\ngYSep) {};
\node[battlefield] (A11) at (\ngXSep,-\ngYSep) {};
\node[battlefield] (A12) at (2*\ngXSep,-\ngYSep) {};
\node[battlefield] (Bv1) at (3*\ngXSep,-\ngYSep) {};

\node[boxlabel, below=4pt of Bu0] {battlefield $\BitB^u_0$};
\node[boxlabel, below=4pt of A01] {battlefield $\AuxB^u_{0,1}$};
\node[boxlabel, below=4pt of A02] {battlefield $\AuxB^u_{0,2}$};
\node[boxlabel, below=4pt of Bv0] {battlefield $\BitB^v_0$};

\node[boxlabel, below=4pt of Bu1] {battlefield $\BitB^u_1$};
\node[boxlabel, below=4pt of A11] {battlefield $\AuxB^u_{1,1}$};
\node[boxlabel, below=4pt of A12] {battlefield $\AuxB^u_{1,2}$};
\node[boxlabel, below=4pt of Bv1] {battlefield $\BitB^v_1$};

\ngOutsidePair{gin0}{Bu0}{\ngColA}
\node[primary]        (pu0u) at \ngPt{Bu0}{\ngColB}{\ngRowU} {$\PrimP^u$};
\node[primary]        (pu0d) at \ngPt{Bu0}{\ngColB}{\ngRowD} {$\PrimP^u$};
\node[detector,ghost] (d0g)  at \ngPt{Bu0}{\ngColC}{\ngRowM} {$\DeP$};

\ngOutsidePair{gin1}{Bu1}{\ngColA}
\node[primary,ghost]  (pu1u) at \ngPt{Bu1}{\ngColB}{\ngRowU} {$\PrimP^u$};
\node[primary]        (pu1d) at \ngPt{Bu1}{\ngColB}{\ngRowD} {$\PrimP^u$};
\node[detector]       (d1)   at \ngPt{Bu1}{\ngColC}{\ngRowM} {$\DeP$};

\node[detector]       (d0)   at \ngPt{A01}{\ngColL}{\ngRowM} {$\DeP$};
\node[blocker,ghost]  (bl0g) at \ngPt{A01}{\ngColR}{\ngRowM} {$\BlP$};

\node[detector,ghost] (d1g)  at \ngPt{A11}{\ngColL}{\ngRowM} {$\DeP$};
\node[blocker]        (bl1)  at \ngPt{A11}{\ngColR}{\ngRowM} {$\BlP$};

\node[blocker]        (bl0u) at \ngPt{A02}{\ngColA}{\ngRowU} {$\BlP$};
\node[blocker]        (bl0d) at \ngPt{A02}{\ngColA}{\ngRowD} {$\BlP$};
\node[guard,ghost]    (gu0g) at \ngPt{A02}{\ngColB}{\ngRowM} {$\GuP$};
\node[pressure,ghost] (pr0u) at \ngPt{A02}{\ngColC}{\ngRowU} {$\PrP_1$};
\node[pressure,ghost] (pr0d) at \ngPt{A02}{\ngColC}{\ngRowD} {$\PrP_2$};

\node[blocker,ghost]  (bl1u)  at \ngPt{A12}{\ngColA}{\ngRowU} {$\BlP$};
\node[blocker]        (bl1d)  at \ngPt{A12}{\ngColA}{\ngRowD} {$\BlP$};
\node[guard]          (gu1)   at \ngPt{A12}{\ngColB}{\ngRowM} {$\GuP$};
\node[pressure]       (pr1uA) at \ngPt{A12}{\ngColC}{\ngRowU} {$\PrP_1$};
\node[pressure]       (pr1dA) at \ngPt{A12}{\ngColC}{\ngRowD} {$\PrP_2$};

\node[pressure,ghost] (p0u)  at \ngPt{Bv0}{\ngColA}{\ngRowU} {$\PrP_1$};
\node[pressure,ghost] (p0d)  at \ngPt{Bv0}{\ngColA}{\ngRowD} {$\PrP_2$};
\node[primary]        (pv0u) at \ngPt{Bv0}{\ngColB}{\ngRowU} {$\PrimP^v$};
\node[primary,ghost]  (pv0d) at \ngPt{Bv0}{\ngColB}{\ngRowD} {$\PrimP^v$};
\ngOutsideSingle{gout0}{Bv0}{\ngColC}

\node[pressure]       (p1u)  at \ngPt{Bv1}{\ngColA}{\ngRowU} {$\PrP_1$};
\node[pressure]       (p1d)  at \ngPt{Bv1}{\ngColA}{\ngRowD} {$\PrP_2$};
\node[primary]        (pv1u) at \ngPt{Bv1}{\ngColB}{\ngRowU} {$\PrimP^v$};
\node[primary]        (pv1d) at \ngPt{Bv1}{\ngColB}{\ngRowD} {$\PrimP^v$};
\ngOutsideSingle{gout1}{Bv1}{\ngColC}

\draw[move] (pu1u.north) -- (pu0d.south);
\draw[move] (d0g.east)   -- (d0.west);
\draw[move] (bl0g.east)  -- (bl0u.west);
\draw[move] (gu0g.north) -- ++(0,1.35);
\draw[move] (pr0u.east)  -- (p1u.west);
\draw[move] (pr0d.east)  -- (p1d.west);
\draw[move] (pv0d.south) -- (pv1u.north);

\pgfresetboundingbox
\path[use as bounding box]
    ($(Bu0.west)+(-0.03,1.95)$)
    rectangle
    ($(Bv1.east)+(0.03,-2.00)$);

\end{tikzpicture}%
\endgroup%
}
\caption{
Illustrative figure for the \NOT{} gate in the case $\val{u}=0$, which forces $\val{v}=1$.
}
\label{fig:uniform-not-gate}
\end{figure}

\begin{figure}[t]
\centering
\resizebox{\linewidth}{!}{\begingroup%
\def\agBoxW{3.95}
\def\agBoxH{2.15}
\def\agXSep{4.75}
\def\agOutX{14.55}
\def\agRowSep{2.70}
\def\agGroupSep{3.25}
\def\agColA{0.25}%
\def\agColB{0.50}%
\def\agColC{0.75}%
\def\agColL{0.34}%
\def\agColR{0.66}%
\def\agOutPLeft{0.18}%
\def\agOutPRight{0.34}%
\def\agOutPrim{0.60}%
\def\agOutExt{0.84}%
\def\agRowU{0.69}%
\def\agRowMU{0.59}%
\def\agRowM{0.50}%
\def\agRowML{0.41}%
\def\agRowD{0.31}%
\def\agOutsideSep{0.20}%
\newcommand{\agPt}[3]{($ (#1.south west) + (#2*\agBoxW cm,#3*\agBoxH cm) $)}%
\newcommand{\agOutsidePair}[3]{%
    \node[externalsmall] (#1a) at ($(#2.south west)+(#3*\agBoxW cm,0.5*\agBoxH cm)+(0,\agOutsideSep)$) {};
    \node[externalsmall] (#1b) at ($(#2.south west)+(#3*\agBoxW cm,0.5*\agBoxH cm)+(0,-\agOutsideSep)$) {};
}%
\newcommand{\agOutsideSingle}[3]{%
    \node[externalsmall] (#1) at ($(#2.south west)+(#3*\agBoxW cm,0.5*\agBoxH cm)$) {};
}%
\begin{tikzpicture}[
    >=Latex,
    every node/.style={font=\small},
    battlefield/.style={draw, rounded corners=4pt, line width=0.9pt, minimum width=\agBoxW cm, minimum height=\agBoxH cm},
    boxlabel/.style={font=\small, fill=white, inner sep=1pt},
    move/.style={->, line width=0.9pt},
    tok/.style={circle, minimum size=7.6mm, inner sep=0pt, font=\tiny},
    primary/.style={tok, draw=blue!70!black, fill=blue!18},
    detector/.style={tok, draw=orange!85!black, fill=orange!22},
    blocker/.style={tok, draw=green!55!black, fill=green!22},
    guard/.style={tok, draw=violet!80!black, fill=violet!18},
    pressure/.style={tok, draw=red!75!black, fill=red!18},
    ghost/.style={densely dotted, line width=0.9pt, fill=white},
    externalsmall/.style={circle, draw=gray!70, densely dotted, line width=0.9pt, fill=gray!8, pattern={Dots[distance=0.90mm, radius=0.18mm]}, pattern color=gray!60, minimum size=7.6mm, inner sep=0pt}
]

\node[battlefield] (Bu0) at (0,0) {};
\node[battlefield] (Au01) at (\agXSep,0) {};
\node[battlefield] (Au02) at (2*\agXSep,0) {};

\node[battlefield] (Bu1) at (0,-\agRowSep) {};
\node[battlefield] (Au11) at (\agXSep,-\agRowSep) {};
\node[battlefield] (Au12) at (2*\agXSep,-\agRowSep) {};

\node[battlefield] (Bv0) at (0,-\agRowSep-\agGroupSep) {};
\node[battlefield] (Av01) at (\agXSep,-\agRowSep-\agGroupSep) {};
\node[battlefield] (Av02) at (2*\agXSep,-\agRowSep-\agGroupSep) {};

\node[battlefield] (Bv1) at (0,-2*\agRowSep-\agGroupSep) {};
\node[battlefield] (Av11) at (\agXSep,-2*\agRowSep-\agGroupSep) {};
\node[battlefield] (Av12) at (2*\agXSep,-2*\agRowSep-\agGroupSep) {};

\node[battlefield] (Bw0) at (\agOutX,-1.35) {};
\node[battlefield] (Bw1) at (\agOutX,-7.25) {};

\node[boxlabel, below=2pt of Bu0] {battlefield $\BitB^u_0$};
\node[boxlabel, below=2pt of Au01] {battlefield $\AuxB^u_{0,1}$};
\node[boxlabel, below=2pt of Au02] {battlefield $\AuxB^u_{0,2}$};

\node[boxlabel, below=2pt of Bu1] {battlefield $\BitB^u_1$};
\node[boxlabel, below=2pt of Au11] {battlefield $\AuxB^u_{1,1}$};
\node[boxlabel, below=2pt of Au12] {battlefield $\AuxB^u_{1,2}$};

\node[boxlabel, below=2pt of Bv0] {battlefield $\BitB^v_0$};
\node[boxlabel, below=2pt of Av01] {battlefield $\AuxB^v_{0,1}$};
\node[boxlabel, below=2pt of Av02] {battlefield $\AuxB^v_{0,2}$};

\node[boxlabel, below=2pt of Bv1] {battlefield $\BitB^v_1$};
\node[boxlabel, below=2pt of Av11] {battlefield $\AuxB^v_{1,1}$};
\node[boxlabel, below=2pt of Av12] {battlefield $\AuxB^v_{1,2}$};

\node[boxlabel, below=2pt of Bw0] {battlefield $\BitB^w_0$};
\node[boxlabel, below=2pt of Bw1] {battlefield $\BitB^w_1$};

\agOutsidePair{ginu0}{Bu0}{\agColA}
\node[primary]        (pu0u) at \agPt{Bu0}{\agColB}{\agRowU} {$\PrimP^u$};
\node[primary]        (pu0d) at \agPt{Bu0}{\agColB}{\agRowD} {$\PrimP^u$};
\node[detector,ghost] (du0g) at \agPt{Bu0}{\agColC}{\agRowM} {$\DeP$};

\node[detector]       (du0) at \agPt{Au01}{\agColL}{\agRowM} {$\DeP$};
\node[blocker,ghost]  (blu01g) at \agPt{Au01}{\agColR}{\agRowM} {$\BlP$};

\node[blocker]        (blu02u) at \agPt{Au02}{\agColA}{\agRowU} {$\BlP$};
\node[blocker]        (blu02d) at \agPt{Au02}{\agColA}{\agRowD} {$\BlP$};
\node[guard,ghost]    (gu0g) at \agPt{Au02}{\agColB}{\agRowM} {$\GuP$};
\node[pressure,ghost] (pru01g) at \agPt{Au02}{\agColC}{\agRowU} {$\PrP_1$};
\node[pressure,ghost] (pru02g) at \agPt{Au02}{\agColC}{\agRowM} {$\PrP_2$};
\node[pressure,ghost] (pru03g) at \agPt{Au02}{\agColC}{\agRowD} {$\PrP_3$};

\agOutsidePair{ginu1}{Bu1}{\agColA}
\node[primary,ghost]  (pu1u) at \agPt{Bu1}{\agColB}{\agRowU} {$\PrimP^u$};
\node[primary]        (pu1d) at \agPt{Bu1}{\agColB}{\agRowD} {$\PrimP^u$};
\node[detector]       (du1) at \agPt{Bu1}{\agColC}{\agRowM} {$\DeP$};

\node[detector,ghost] (du1g) at \agPt{Au11}{\agColL}{\agRowM} {$\DeP$};
\node[blocker]        (blu11) at \agPt{Au11}{\agColR}{\agRowM} {$\BlP$};

\node[blocker,ghost]  (blu12g) at \agPt{Au12}{\agColA}{\agRowU} {$\BlP$};
\node[blocker]        (blu12) at \agPt{Au12}{\agColA}{\agRowD} {$\BlP$};
\node[guard]          (gu1) at \agPt{Au12}{\agColB}{\agRowM} {$\GuP$};
\node[pressure]       (pru1) at \agPt{Au12}{\agColC}{\agRowM} {$\PrP_1$};

\agOutsidePair{ginv0}{Bv0}{\agColA}
\node[primary]        (pv0u) at \agPt{Bv0}{\agColB}{\agRowU} {$\PrimP^v$};
\node[primary,ghost]  (pv0d) at \agPt{Bv0}{\agColB}{\agRowD} {$\PrimP^v$};
\node[detector]       (dv0) at \agPt{Bv0}{\agColC}{\agRowM} {$\DeP$};

\node[detector,ghost] (dv0g) at \agPt{Av01}{\agColL}{\agRowM} {$\DeP$};
\node[blocker]        (blv01) at \agPt{Av01}{\agColR}{\agRowM} {$\BlP$};

\node[blocker,ghost]  (blv02g) at \agPt{Av02}{\agColA}{\agRowU} {$\BlP$};
\node[blocker]        (blv02) at \agPt{Av02}{\agColA}{\agRowD} {$\BlP$};
\node[guard]          (gv0) at \agPt{Av02}{\agColB}{\agRowM} {$\GuP$};
\node[pressure]       (prv04) at \agPt{Av02}{\agColC}{\agRowU} {$\PrP_1$};
\node[pressure]       (prv05) at \agPt{Av02}{\agColC}{\agRowM} {$\PrP_2$};
\node[pressure]       (prv06) at \agPt{Av02}{\agColC}{\agRowD} {$\PrP_3$};

\agOutsidePair{ginv1}{Bv1}{\agColA}
\node[primary]        (pv1u) at \agPt{Bv1}{\agColB}{\agRowU} {$\PrimP^v$};
\node[primary]        (pv1d) at \agPt{Bv1}{\agColB}{\agRowD} {$\PrimP^v$};
\node[detector,ghost] (dv1g) at \agPt{Bv1}{\agColC}{\agRowM} {$\DeP$};

\node[detector]       (dv1) at \agPt{Av11}{\agColL}{\agRowM} {$\DeP$};
\node[blocker,ghost]  (blv11g) at \agPt{Av11}{\agColR}{\agRowM} {$\BlP$};

\node[blocker]        (blv12u) at \agPt{Av12}{\agColA}{\agRowU} {$\BlP$};
\node[blocker]        (blv12d) at \agPt{Av12}{\agColA}{\agRowD} {$\BlP$};
\node[guard,ghost]    (gv1g) at \agPt{Av12}{\agColB}{\agRowM} {$\GuP$};
\node[pressure,ghost] (prv1g) at \agPt{Av12}{\agColC}{\agRowM} {$\PrP_1$};

\node[pressure]       (pwu1) at \agPt{Bw0}{\agOutPLeft}{\agRowU} {$\PrP^u_1$};
\node[pressure]       (pwu2) at \agPt{Bw0}{\agOutPLeft}{\agRowM} {$\PrP^u_2$};
\node[pressure]       (pwu3) at \agPt{Bw0}{\agOutPLeft}{\agRowD} {$\PrP^u_3$};
\node[pressure,ghost] (pwv4) at \agPt{Bw0}{\agOutPRight}{\agRowU} {$\PrP^v_1$};
\node[pressure,ghost] (pwv5) at \agPt{Bw0}{\agOutPRight}{\agRowM} {$\PrP^v_2$};
\node[pressure,ghost] (pwv6) at \agPt{Bw0}{\agOutPRight}{\agRowD} {$\PrP^v_3$};
\node[primary]        (pw0u) at \agPt{Bw0}{\agOutPrim}{\agRowU} {$\PrimP^w$};
\node[primary]        (pw0d) at \agPt{Bw0}{\agOutPrim}{\agRowD} {$\PrimP^w$};
\agOutsideSingle{gout0}{Bw0}{\agOutExt}

\node[pressure,ghost] (pwu7) at \agPt{Bw1}{\agOutPLeft}{\agRowM} {$\PrP^u_1$};
\node[pressure]       (pwv8) at \agPt{Bw1}{\agOutPRight}{\agRowM} {$\PrP^v_1$};
\node[primary,ghost]  (pw1u) at \agPt{Bw1}{\agOutPrim}{\agRowU} {$\PrimP^w$};
\node[primary]        (pw1d) at \agPt{Bw1}{\agOutPrim}{\agRowD} {$\PrimP^w$};
\agOutsideSingle{gout1}{Bw1}{\agOutExt}

\draw[move] (du0g.east) -- (du0.west);
\draw[move] (blu01g.east) -- (blu02u.west);
\draw[move] (gu0g.north) -- ++(0,1.05);
\draw[move] (pru01g.east) -- (pwu1.west);
\draw[move] (pru02g.east) -- (pwu2.west);
\draw[move] (pru03g.east) -- (pwu3.west);

\draw[move] (dv1g.east) -- (dv1.west);
\draw[move] (blv11g.east) -- (blv12u.west);
\draw[move] (gv1g.south) -- ++(0,-1.05);
\draw[move] (prv1g.east) -- (pwv8.west);

\draw[move] (pu1u.north) -- (pu0d.south);
\draw[move] (pv0d.south) -- (pv1u.north);
\draw[move] (pw1u.north) -- (pw0d.south);

\coordinate (agTop) at ($(gu0g.north)+(0,1.15)$);
\coordinate (agBottom) at ($(gv1g.south)+(0,-1.15)$);

\pgfresetboundingbox
\path[use as bounding box]
    ($(Bu0.west |- agTop)+(-0.03,0.03)$)
    rectangle
    ($(Bw1.east |- agBottom)+(0.03,-0.03)$);

\end{tikzpicture}%
\endgroup%
}
\caption{Illustrative figure for \AND{} gate in the case $\val{u}=0$ and $\val{v}=1$, which forces $\val{w}=0$.}
\label{fig:uniform-and-gate}
\end{figure}
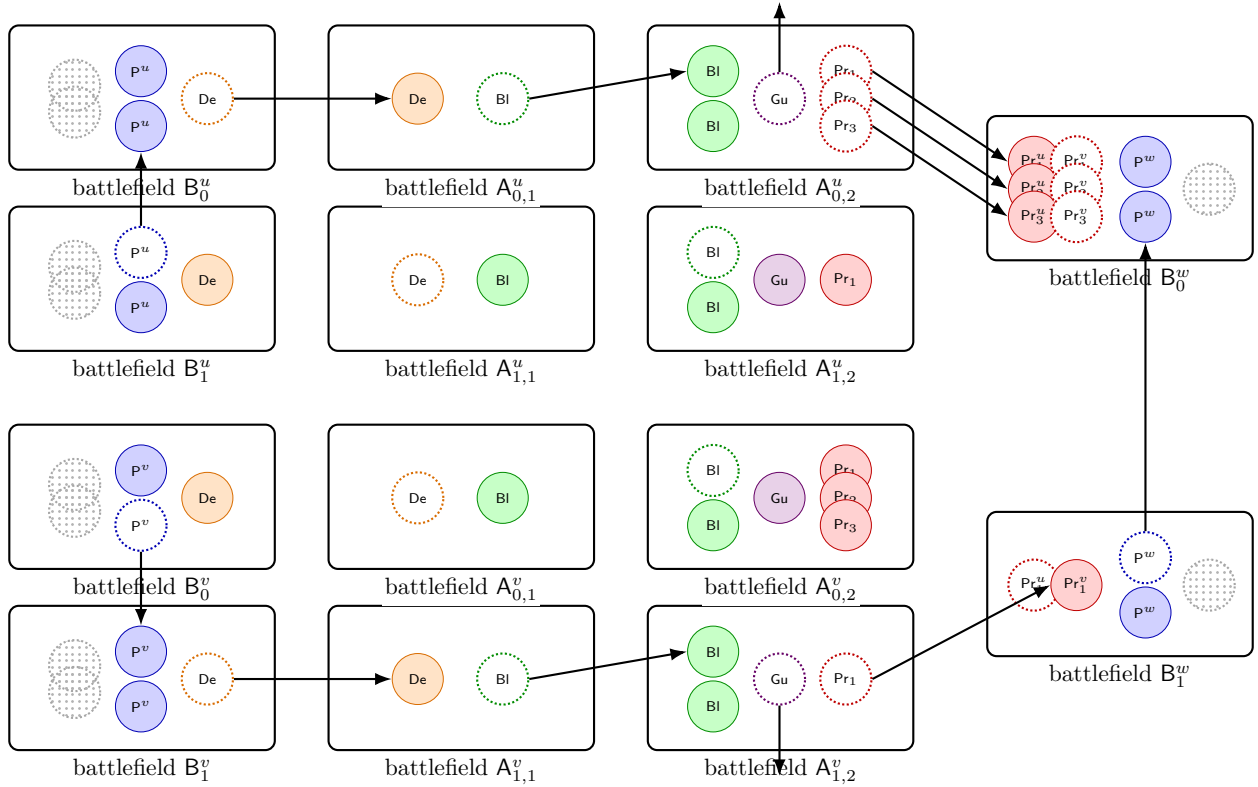

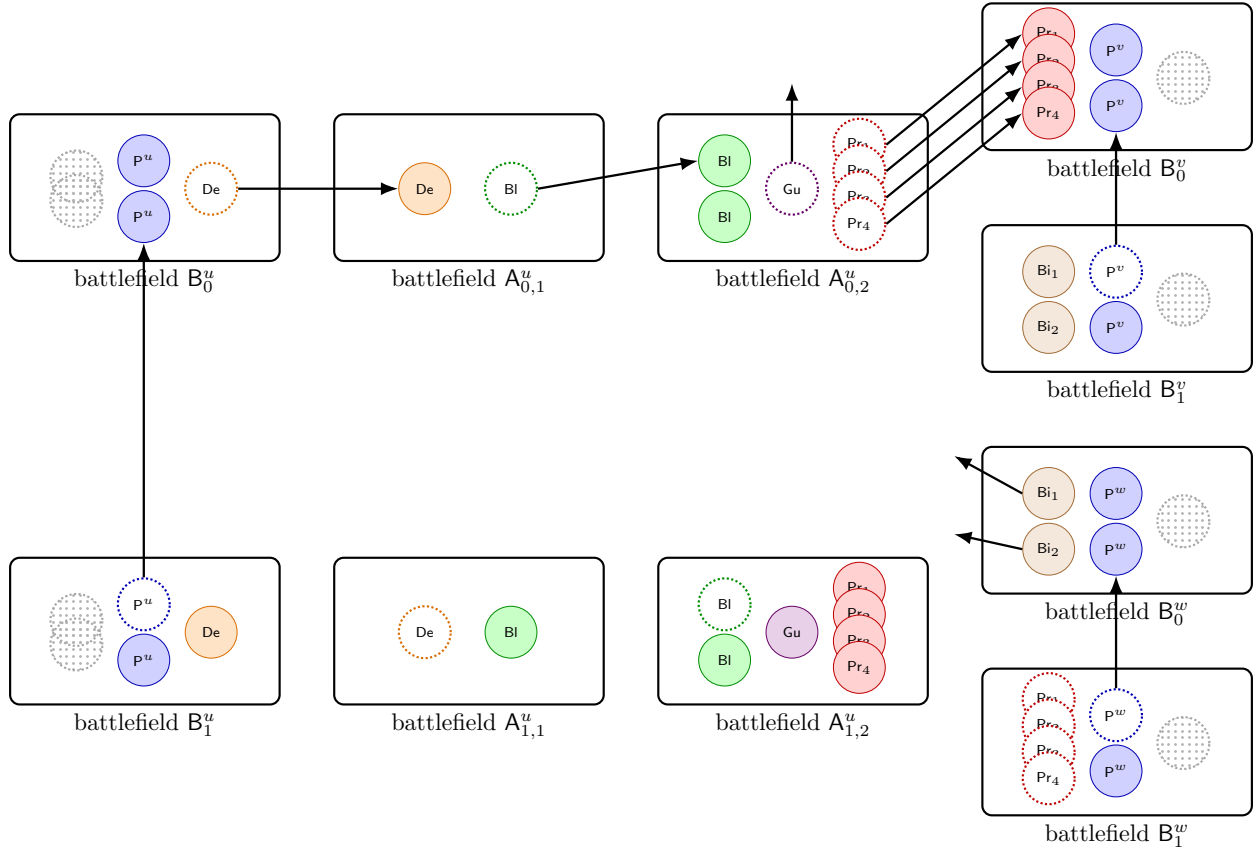
\begin{figure}[t]
\centering
\resizebox{\linewidth}{!}{\begingroup%
\def\pgBoxW{3.95}
\def\pgBoxH{2.15}
\def\pgXSep{4.75}
\def\pgYGap{3.25}
\def\pgHalfGap{1.625}

\def\pgColA{0.25}%
\def\pgColB{0.50}%
\def\pgColC{0.75}%
\def\pgColL{0.34}%
\def\pgColR{0.66}%

\def\pgRowU{0.69}%
\def\pgRowM{0.50}%
\def\pgRowD{0.31}%

\def\pgRowPone{0.80}%
\def\pgRowPtwo{0.62}%
\def\pgRowPthree{0.44}%
\def\pgRowPfour{0.26}%

\def\pgOutsideSep{0.18}%

\newcommand{\pgPt}[3]{($ (#1.south west) + (#2*\pgBoxW cm,#3*\pgBoxH cm) $)}%

\newcommand{\pgOutsidePair}[3]{%
  \node[externalsmall] (#1a) at ($(#2.south west)+(#3*\pgBoxW cm,0.5*\pgBoxH cm)+(0,\pgOutsideSep)$) {};
  \node[externalsmall] (#1b) at ($(#2.south west)+(#3*\pgBoxW cm,0.5*\pgBoxH cm)+(0,-\pgOutsideSep)$) {};
}%

\newcommand{\pgOutsideSingle}[3]{%
  \node[externalsmall] (#1) at ($(#2.south west)+(#3*\pgBoxW cm,0.5*\pgBoxH cm)$) {};
}%

\begin{tikzpicture}[
  >=Latex,
  every node/.style={font=\small},
  battlefield/.style={
      draw,
      rounded corners=4pt,
      line width=0.9pt,
      minimum width=\pgBoxW cm,
      minimum height=\pgBoxH cm
  },
  boxlabel/.style={font=\small, fill=white, inner sep=1pt},
  move/.style={->, line width=0.95pt},
  tok/.style={circle, minimum size=7.6mm, inner sep=0pt, font=\tiny},
  primary/.style={tok, draw=blue!70!black, fill=blue!18},
  detector/.style={tok, draw=orange!85!black, fill=orange!22},
  blocker/.style={tok, draw=green!55!black, fill=green!22},
  guard/.style={tok, draw=violet!80!black, fill=violet!18},
  pressure/.style={tok, draw=red!75!black, fill=red!18},
  bias/.style={tok, draw=brown!85!black, fill=brown!18},
  ghost/.style={densely dotted, line width=0.9pt, fill=white},
  externalsmall/.style={
      circle,
      draw=gray!70,
      densely dotted,
      line width=0.9pt,
      fill=gray!8,
      pattern={Dots[distance=0.90mm, radius=0.18mm]},
      pattern color=gray!60,
      minimum size=7.6mm,
      inner sep=0pt
  }
]

\node[battlefield] (Bv0) at (3*\pgXSep,0) {};
\node[battlefield] (Bv1) at (3*\pgXSep,-\pgYGap) {};
\node[battlefield] (Bw0) at (3*\pgXSep,-2*\pgYGap) {};
\node[battlefield] (Bw1) at (3*\pgXSep,-3*\pgYGap) {};

\node[battlefield] (Bu0) at (0,-\pgHalfGap) {};
\node[battlefield] (A01) at (\pgXSep,-\pgHalfGap) {};
\node[battlefield] (A02) at (2*\pgXSep,-\pgHalfGap) {};

\node[battlefield] (Bu1) at (0,-2*\pgYGap-\pgHalfGap) {};
\node[battlefield] (A11) at (\pgXSep,-2*\pgYGap-\pgHalfGap) {};
\node[battlefield] (A12) at (2*\pgXSep,-2*\pgYGap-\pgHalfGap) {};

\node[boxlabel, below=2pt of Bu0] {battlefield $\BitB^u_0$};
\node[boxlabel, below=2pt of A01] {battlefield $\AuxB^u_{0,1}$};
\node[boxlabel, below=2pt of A02] {battlefield $\AuxB^u_{0,2}$};
\node[boxlabel, below=2pt of Bv0] {battlefield $\BitB^v_0$};
\node[boxlabel, below=2pt of Bv1] {battlefield $\BitB^v_1$};
\node[boxlabel, below=2pt of Bw0] {battlefield $\BitB^w_0$};
\node[boxlabel, below=2pt of Bw1] {battlefield $\BitB^w_1$};
\node[boxlabel, below=2pt of Bu1] {battlefield $\BitB^u_1$};
\node[boxlabel, below=2pt of A11] {battlefield $\AuxB^u_{1,1}$};
\node[boxlabel, below=2pt of A12] {battlefield $\AuxB^u_{1,2}$};

\pgOutsidePair{gin0}{Bu0}{\pgColA}
\node[primary]        (pu0u) at \pgPt{Bu0}{\pgColB}{\pgRowU} {$\PrimP^u$};
\node[primary]        (pu0d) at \pgPt{Bu0}{\pgColB}{\pgRowD} {$\PrimP^u$};
\node[detector,ghost] (du0g) at \pgPt{Bu0}{\pgColC}{\pgRowM} {$\DeP$};

\node[detector]       (du0)  at \pgPt{A01}{\pgColL}{\pgRowM} {$\DeP$};
\node[blocker,ghost]  (bl0g) at \pgPt{A01}{\pgColR}{\pgRowM} {$\BlP$};

\node[blocker]        (bl0u) at \pgPt{A02}{\pgColA}{\pgRowU} {$\BlP$};
\node[blocker]        (bl0d) at \pgPt{A02}{\pgColA}{\pgRowD} {$\BlP$};
\node[guard,ghost]    (gu0g) at \pgPt{A02}{\pgColB}{\pgRowM} {$\GuP$};
\node[pressure,ghost] (pr0a) at \pgPt{A02}{\pgColC}{\pgRowPone}   {$\PrP_1$};
\node[pressure,ghost] (pr0b) at \pgPt{A02}{\pgColC}{\pgRowPtwo}   {$\PrP_2$};
\node[pressure,ghost] (pr0c) at \pgPt{A02}{\pgColC}{\pgRowPthree} {$\PrP_3$};
\node[pressure,ghost] (pr0d) at \pgPt{A02}{\pgColC}{\pgRowPfour}  {$\PrP_4$};

\pgOutsidePair{gin1}{Bu1}{\pgColA}
\node[primary,ghost]  (pu1u) at \pgPt{Bu1}{\pgColB}{\pgRowU} {$\PrimP^u$};
\node[primary]        (pu1d) at \pgPt{Bu1}{\pgColB}{\pgRowD} {$\PrimP^u$};
\node[detector]       (du1)  at \pgPt{Bu1}{\pgColC}{\pgRowM} {$\DeP$};

\node[detector,ghost] (du1g) at \pgPt{A11}{\pgColL}{\pgRowM} {$\DeP$};
\node[blocker]        (bl1)  at \pgPt{A11}{\pgColR}{\pgRowM} {$\BlP$};

\node[blocker,ghost]  (bl1u) at \pgPt{A12}{\pgColA}{\pgRowU} {$\BlP$};
\node[blocker]        (bl1d) at \pgPt{A12}{\pgColA}{\pgRowD} {$\BlP$};
\node[guard]          (gu1)  at \pgPt{A12}{\pgColB}{\pgRowM} {$\GuP$};
\node[pressure]       (pr1a) at \pgPt{A12}{\pgColC}{\pgRowPone}   {$\PrP_1$};
\node[pressure]       (pr1b) at \pgPt{A12}{\pgColC}{\pgRowPtwo}   {$\PrP_2$};
\node[pressure]       (pr1c) at \pgPt{A12}{\pgColC}{\pgRowPthree} {$\PrP_3$};
\node[pressure]       (pr1d) at \pgPt{A12}{\pgColC}{\pgRowPfour}  {$\PrP_4$};

\node[pressure]       (pv0a) at \pgPt{Bv0}{\pgColA}{\pgRowPone}   {$\PrP_1$};
\node[pressure]       (pv0b) at \pgPt{Bv0}{\pgColA}{\pgRowPtwo}   {$\PrP_2$};
\node[pressure]       (pv0c) at \pgPt{Bv0}{\pgColA}{\pgRowPthree} {$\PrP_3$};
\node[pressure]       (pv0d) at \pgPt{Bv0}{\pgColA}{\pgRowPfour}  {$\PrP_4$};
\node[primary]        (qv0u) at \pgPt{Bv0}{\pgColB}{\pgRowU} {$\PrimP^v$};
\node[primary]        (qv0d) at \pgPt{Bv0}{\pgColB}{\pgRowD} {$\PrimP^v$};
\pgOutsideSingle{gv0}{Bv0}{\pgColC}

\node[bias]           (bv1u) at \pgPt{Bv1}{\pgColA}{\pgRowU} {$\BiP_1$};
\node[bias]           (bv1d) at \pgPt{Bv1}{\pgColA}{\pgRowD} {$\BiP_2$};
\node[primary,ghost]  (qv1u) at \pgPt{Bv1}{\pgColB}{\pgRowU} {$\PrimP^v$};
\node[primary]        (qv1d) at \pgPt{Bv1}{\pgColB}{\pgRowD} {$\PrimP^v$};
\pgOutsideSingle{gv1}{Bv1}{\pgColC}

\node[bias]           (bw0u) at \pgPt{Bw0}{\pgColA}{\pgRowU} {$\BiP_1$};
\node[bias]           (bw0d) at \pgPt{Bw0}{\pgColA}{\pgRowD} {$\BiP_2$};
\node[primary]        (qw0u) at \pgPt{Bw0}{\pgColB}{\pgRowU} {$\PrimP^w$};
\node[primary]        (qw0d) at \pgPt{Bw0}{\pgColB}{\pgRowD} {$\PrimP^w$};
\pgOutsideSingle{gw0}{Bw0}{\pgColC}

\node[pressure,ghost] (pw1a) at \pgPt{Bw1}{\pgColA}{\pgRowPone}   {$\PrP_1$};
\node[pressure,ghost] (pw1b) at \pgPt{Bw1}{\pgColA}{\pgRowPtwo}   {$\PrP_2$};
\node[pressure,ghost] (pw1c) at \pgPt{Bw1}{\pgColA}{\pgRowPthree} {$\PrP_3$};
\node[pressure,ghost] (pw1d) at \pgPt{Bw1}{\pgColA}{\pgRowPfour}  {$\PrP_4$};
\node[primary,ghost]  (qw1u) at \pgPt{Bw1}{\pgColB}{\pgRowU} {$\PrimP^w$};
\node[primary]        (qw1d) at \pgPt{Bw1}{\pgColB}{\pgRowD} {$\PrimP^w$};
\pgOutsideSingle{gw1}{Bw1}{\pgColC}

\draw[move] (pu1u.north) -- (pu0d.south);
\draw[move] (du0g.east) -- (du0.west);
\draw[move] (bl0g.east) -- (bl0u.west);
\draw[move] (gu0g.north) -- ++(0,1.15);
\draw[move] (pr0a.east) -- (pv0a.west);
\draw[move] (pr0b.east) -- (pv0b.west);
\draw[move] (pr0c.east) -- (pv0c.west);
\draw[move] (pr0d.east) -- (pv0d.west);
\draw[move] (qv1u.north) -- (qv0d.south);
\draw[move] (qw1u.north) -- (qw0d.south);
\draw[move] (bw0u.west) -- ++(-1.00,0.55);
\draw[move] (bw0d.west) -- ++(-1.00,0.22);

\coordinate (pgTop) at ($(Bv0.north)+(0,0.18)$);
\coordinate (pgBottom) at ($(Bw1.south)+(0,-0.55)$);
\pgfresetboundingbox
\path[use as bounding box]
  ($(Bu0.west |- pgTop)+(-0.03,0.03)$)
  rectangle
  ($(Bv0.east |- pgBottom)+(0.03,-0.03)$);

\end{tikzpicture}%
\endgroup%
}
\caption{Illustrative figure for the \PURIFY{} gate in the case $\val{u}=0$, which forces $\val{v} = \val{w} = 0$.}
\label{fig:uniform-purify-gate}
\end{figure}

\paragraph{Primary player behavior.}

We frequently use the following simple consequences of uniform tie-breaking.
A player who puts one resource on a battlefield gets no payoff from that battlefield if some opponent puts two or more resources there.
If all players who use the battlefield put exactly one resource there, then they share the battlefield uniformly.
A player with two resources can either split them across two battlefields, or put both resources on one battlefield and beat all one-resource players on that battlefield.
The primary players use three resources to combine these effects: one resource can stay on each bit battlefield, while the third resource decides which bit battlefield is won against all one-resource auxiliary players.

The reserve battlefield argument works as follows.
Suppose a player uses a resource on a battlefield for which it has zero value.
Then this resource can be moved to one of the player's private reserve battlefields.
If this deviation were not profitable, then many opponents would have to be placing resources on this player's reserve battlefields.
Summing this over all players who use zero-valued battlefields gives a contradiction, because all such resources must themselves come from players using zero-valued battlefields.
This proves that zero-valued battlefields are not used in any approximate WSNE (\cref{lm:reserve}).

We can now describe the behavior of primary players.
Given the reserve battlefield argument, the primary player $\PrimP^v$ can only use $\BitB^v_0$, $\BitB^v_1$, and its private reserve battlefields.
It will not put all three resources on one bit battlefield.
Indeed, putting two resources there already beats every possible opponent on that battlefield, because every auxiliary player that may use a bit battlefield has only one resource.
Moving the third resource to the other bit battlefield keeps the payoff from the first battlefield and gives additional payoff from the second.
The primary player also will not use its reserve battlefields.
If it uses a reserve battlefield, then one of the two bit battlefields is missing a resource or has only one resource.
Moving the reserve resource to such a bit battlefield gives more payoff than the reserve option.
The constants are chosen so that this improvement is larger than the approximation parameter, even after accounting for the constantly many one-resource auxiliary players that can interact with a bit battlefield.
Thus, in any approximate WSNE, $\PrimP^v$ only plays the two actions $(2,1)$ and $(1,2)$ over $(\BitB^v_0,\BitB^v_1)$ (\cref{lm:primary}).
This gives the desired probabilistic encoding of variables.

\paragraph{Auxiliary player behavior.}
Fix a gate, an input variable $u$ of this gate, and a bit value $s \in \{0,1\}$.
Consider the auxiliary group that checks whether $\PrimP^u$ puts two resources on $\BitB^u_s$ with high probability.
If $\PrimP^u$ puts two resources on $\BitB^u_s$ with high probability, then the detector does not get enough payoff from $\BitB^u_s$, and moves to the first auxiliary battlefield of its group.
In the figures, this is the first arrow leaving the input bit battlefield.
We call this detector active.
If $\PrimP^u$ puts only one resource on $\BitB^u_s$ with high enough probability, then the detector gets enough payoff from $\BitB^u_s$, and stays there.
We call this detector inactive.
This is the first step in converting the value of an input variable into incentives for auxiliary players.

The blocker and guard propagate the detector's behavior to the pressure players.
If the detector is inactive, then the first auxiliary battlefield is free.
The blocker puts one resource on each of the two auxiliary battlefields of the group, the guard stays on the second auxiliary battlefield, and the pressure players also stay on the second auxiliary battlefield.
Thus an inactive detector creates no pressure on any output bit battlefield.

If the detector is active, then it occupies the first auxiliary battlefield.
The blocker then puts both resources on the second auxiliary battlefield with high probability.
This makes the second auxiliary battlefield unattractive for the pressure players, so they leave it and choose between their selected output bit battlefield and their reserve battlefields.
This is the movement from the auxiliary battlefield toward the output bit battlefield in the figures.
If the output primary player has not yet put two resources on that selected bit battlefield with high probability, then the pressure players strictly prefer the selected bit battlefield to their reserves (\cref{lm:pressure,lm:pressure2}).
Thus an active detector creates conditional pressure on a selected output bit battlefield.

This pressure is what forces the output value.
Suppose some pressure players move to an output bit battlefield $\BitB^v_s$.
If the primary player $\PrimP^v$ puts only one resource on $\BitB^v_s$, then it has to share that battlefield with the pressure players.
If instead it puts two resources on $\BitB^v_s$, then it wins the battlefield outright.
The constants and the number of pressure players are chosen so that this payoff comparison forces $\PrimP^v$ to put two resources on the selected bit battlefield with high probability.
In this way the gadgets implement gate constraints.

\paragraph{Gates.}
For a \NOT{} gate $v=\NOT(u)$, the auxiliary group that checks the case $\val{u}=s$ sends pressure to $\BitB^v_{1-s}$.
\Cref{fig:uniform-not-gate} shows the case $\val{u}=0$.
The primary player $\PrimP^u$ puts two resources on $\BitB^u_0$ with high probability, so the detector for bit value $0$ is active.
This makes the pressure players leave the second auxiliary battlefield and look at $\BitB^v_1$.
At the same time, $p^u_1$ is close to $0$, so the detector for bit value $1$ is inactive, and the corresponding pressure players stay on their second auxiliary battlefield.
Now suppose for contradiction that $p^v_1$ is not close to $1$.
Then the released pressure players move to $\BitB^v_1$.
Given this pressure, the primary player $\PrimP^v$ strictly prefers the action $(1,2)$ over $(2,1)$ on $(\BitB^v_0,\BitB^v_1)$.
Hence $\PrimP^v$ must play $(1,2)$ with probability $1$, contradicting that $p^v_1$ is not close to $1$.
Therefore $p^v_1$ is close to $1$, and so $\val{v}=1$.
The case $\val{u}=1$ is symmetric and forces $\val{v}=0$.
If $\val{u}\in(0,1)$, then the \PURE{} gate imposes no condition on the output, so no further argument is needed.

For an \AND{} gate $w=\AND(u,v)$, the construction uses the same auxiliary mechanism for both input variables.
For each input variable $\sigma \in \{u,v\}$, the auxiliary group that checks the case $\val{\sigma}=s$ sends pressure to $\BitB^w_s$.
The main difference from the \NOT{} gate is the number of pressure players.
There are more pressure players in the groups for bit value $0$ than in the groups for bit value $1$.
This matches the logic of \AND{}: one input equal to $0$ should already force the output to $0$, while both inputs must be equal to $1$ to force the output to $1$.

\Cref{fig:uniform-and-gate} shows the case $\val{u}=0$ and $\val{v}=1$, which is one of the cases where the output must be $0$.
The detector for bit value $0$ of input $u$ is active and releases pressure players toward $\BitB^w_0$.
The detector for bit value $1$ of input $u$ is inactive, so its pressure players stay on their second auxiliary battlefield.
If $p^w_0$ were not close to $1$, then the released pressure players would move to $\BitB^w_0$.
There are enough of them to make $\PrimP^w$ strictly prefer putting two resources on $\BitB^w_0$, even after accounting for the possible effect of the other input.
Thus $\val{w}=0$.
The case $\val{v}=0$ is symmetric.
If instead $\val{u}=\val{v}=1$, then the groups for bit value $1$ of both inputs release pressure toward $\BitB^w_1$, while the groups for bit value $0$ of both inputs create no pressure on the output.
Together, the released pressure players force $\PrimP^w$ to put two resources on $\BitB^w_1$, and hence $\val{w}=1$.
All other input cases are unrestricted in \PURE{}.

The \PURIFY{} gate has one input and two outputs.
Consider a gate $(v,w)=\PURIFY(u)$.
The gate again uses detector, blocker, guard, and pressure players, but it also uses bias players.
The purpose of the gate is different from \NOT{} and \AND{}.
If the input is in $\{0,1\}$, then both outputs must copy it.
If the input lies in $(0,1)$, then at least one of the two outputs must be in $\{0,1\}$.

The construction is arranged so that the two detectors give four implications.
If the detector for bit value $0$ is active, then pressure players force the first output $v$ to be $0$.
If the detector for bit value $0$ is inactive, then no pressure players move to $\BitB^v_0$, and the bias players force $v$ to be $1$.
Similarly, if the detector for bit value $1$ is active, then pressure players force the second output $w$ to be $1$.
If the detector for bit value $1$ is inactive, then no pressure players move to $\BitB^w_1$, and the bias players force $w$ to be $0$.
The bias players are needed in the cases where the corresponding detector is inactive.

\Cref{fig:uniform-purify-gate} shows the case $\val{u}=0$.
The detector for bit value $0$ is active and the detector for bit value $1$ is inactive.
The active detector sends pressure that forces $\val{v}=0$, while the inactive detector leaves the bias players to force $\val{w}=0$.
Thus both outputs copy the Boolean input.
If $\val{u}=1$, the symmetric argument gives $\val{v}=1$ and $\val{w}=1$.
Finally, suppose $\val{u}\in(0,1)$.
The active and inactive thresholds are chosen so that, using the fact that $p^u_0 + p^u_1 = 1$, at least one detector has a specified behavior.
Therefore, one of the four implications above applies, and at least one of the two outputs lies in $\{0,1\}$, as required.

This proves that every approximate WSNE of the constructed Blotto game gives a valid solution to the \PURE{} instance.
All players in the base construction have at most three resources.
To make every player have exactly three resources, the proof adds private high-value battlefields to players with fewer than three resources.
In any approximate WSNE, these players must put one resource on each such private battlefield, and the remaining resources simulate the original construction.
The proof then returns from the modified no-resource tie-breaking convention to the standard uniform tie-breaking convention.
The two games differ only on empty battlefields, and this changes every payoff by at most a small additive term.
After normalization, we obtain \cref{thm:hardness-unif-wsne}.

\paragraph{Approximate NE and small budgets.}
The proof above uses the well-supported condition in several places.
For example, \cref{lm:reserve} shows that no player uses a zero-valued battlefield by proving that any such action is worse than another action by more than the approximation parameter.
This conclusion is valid in a WSNE, but not in an approximate NE.
In an approximate NE, a player may put small probability on a much worse action as long as the contribution to expected regret is small.
Since every player can in principle allocate resources to every battlefield, these small probabilities can accumulate over all players and change the incentives on a battlefield by an amount proportional to the number of players.
This is why the direct consequence for approximate NE gives hardness only for approximation parameter $\Theta(1/n)$ (\cref{cor:ne}).
If we restrict the action space so that a player may allocate resources only to battlefields for which it has positive value, then each battlefield in the construction can be affected by only constantly many players, so the accumulated error remains constant, and the same construction gives hardness for constant-approximate NE (\cref{cor:ne2}).

The hardness construction uses three resources for the primary players in an essential way: one resource stays on each bit battlefield, and the third resource determines which bit battlefield receives two resources.
The allocation of two resources to a bit battlefield is essential to induce the detectors, and as a consequence, the other auxiliary players, to change their actions.
On the other hand, the allocation of a single resource to the other bit battlefield ensures that the primary player still gets a positive payoff from it, which is essential to ensure that the primary player chooses the suitable action among the two when influenced by the pressure and bias players on the output side.
At the other extreme, when every player has one resource, each pure strategy is simply to choose one battlefield.
The game then reduces to a singleton congestion game with player-specific payoff functions.\footnote{There is a small bookkeeping issue under the standard uniform tie-breaking convention, because players may also receive payoff from battlefields on which all players allocate zero resources. This is handled in \cref{app:singleton-congestion}.}
Such games have polynomial-time computable pure Nash equilibria~\cite{milchtaich1996congestion,ackermann2009pure}.
The case of two resources lies between these two: it does not reduce to singleton congestion games, but the three-resource encoding above is also difficult to implement using only two resources.
We leave this case open.

\paragraph{Non-uniform tie-breaking.}
The hardness overview so far used uniform tie-breaking.
We also prove a hardness result for non-uniform tie-breaking.
The reduction is again from \PURE{}, but it uses a different source of payoff differences.
In the uniform reduction, these differences come from player-specific values, budgets of more than one, and auxiliary battlefields.
In the non-uniform reduction, the values are identical and every player has one resource; the payoff differences come from the tie-breaking rule itself.

For every variable $v$, the primary player $\PrimP^v$ chooses between two bit battlefields $\BitB^v_0$ and $\BitB^v_1$.
The probability that $\PrimP^v$ chooses $\BitB^v_1$ encodes the value of $v$.
Each bit battlefield also has dummy players.
The tie-breaking rule ensures that, in every approximate WSNE, the dummy players play their own battlefield with probability $1$.
This makes unrelated battlefields unattractive for the other players and keeps the influence of each player local.

The gate logic is implemented by auxiliary players and the tie-breaking shares.
For each variable $v$ and bit value $s \in \{0,1\}$, there is an auxiliary player that chooses between the input battlefield $\BitB^v_s$ and an output battlefield determined by the gate.
The tie-breaking rule makes this auxiliary player move to the output battlefield when the encoded input value is $s$.
Once it moves there, the same tie-breaking rule changes the payoff of the output primary player and forces the required output choice.
In this way, the tie-breaking rule replaces the detector, blocker, guard, and pressure-player chain from the uniform reduction.

The proof then shows that these local effects implement the \NOT{}, \AND{}, and \PURIFY{} constraints.
The tie-breaking rules used in the construction are monotone.
Moreover, the construction gives hardness not only for constant-approximate WSNE, but also for constant-approximate NE (\cref{thm:nonunif}).

\subsection{Membership}
\label{sec:overview:membership}

We next explain the PPAD-membership result.
We derive the result for approximate equilibria instead of exact equilibria.
In two-player normal-form games, once the supports are fixed, the equilibrium conditions are linear.
This gives a rational polyhedral structure when the payoffs are rational.
For three or more players, the payoff of a pure action is multilinear in the mixed strategies of the other players.
Thus, after fixing supports, the equilibrium conditions become polynomial equations, and isolated equilibria may be irrational even when all payoffs are rational~\cite{etessami2010complexity}.
The same phenomenon already appears in small multiplayer Blotto games with rational values and uniform tie-breaking (\cref{ex:isolated-irrational}).
Therefore we work with $\eps$-approximate equilibria, where $\eps$ is part of the input and may be inverse exponential in the input size.

We use the input and output conventions described in the preliminaries and in \cref{sec:membership}: budgets are represented in unary, and the output is given in a compact representation.
The compact representation we use is the standard layered-graph flow representation for discrete Blotto games~\cite{behnezhad2023fast}.
For each player $i$, we build a layered acyclic graph $D_i$.
The layers correspond to battlefields, and an $s_i$--$t_i$ path corresponds to a pure allocation of player $i$'s resources.
A mixed strategy can therefore be represented by the unit flow it induces in this graph.
For a given flow, say $f_i$, the edge variables determine the marginal probabilities, say $p_{ijr}(f_i)$, that player $i$ allocates exactly $r$ resources to battlefield $j$.
Conversely, since $D_i$ is acyclic, every unit flow can be decomposed in polynomial time into a polynomial-support distribution over pure allocations~\cite{ahuja1993network}.

The key point is that these marginals are enough to compute payoffs.
Blotto payoffs are additive over battlefields, and the payoff from one battlefield depends only on how many resources each player allocates to that battlefield.
Under uniform tie-breaking, the expected share of a battlefield can be computed efficiently from the marginals.
Fix a player $i$, a battlefield $j$, and an amount $r$.
If player $i$ allocates $r$ resources to battlefield $j$, then she receives positive payoff only when no opponent allocates more than $r$ resources there.
Conditional on this, her share is $1/(t+1)$ if exactly $t$ opponents also allocate exactly $r$ resources.
A dynamic program over the opponents computes this expected share.
This gives polynomial-size arithmetic circuits for the utilities and, by the Baur--Strassen theorem, for their gradients~\cite{baur1983complexity} (\cref{lm:compact-payoff}).

We can then view the game as a compact concave game.
Player $i$'s strategy space is the unit-flow polytope $\cF_i$ in $D_i$.
For fixed opponents' flows, player $i$'s utility is linear in her own flow variables, and hence concave.
The flow polytopes have polynomial-size linear descriptions and strong separation oracles.
Together with the utility and gradient circuits from \cref{lm:compact-payoff}, this lets us apply the PPAD-membership theorem for concave games with strong separation oracles (\cref{thm:concave-games-so}, \cite{papadimitriou2023computational}).
Thus computing an $\eps$-NE under uniform tie-breaking is in PPAD (\cref{thm:membership-ne}).

The WSNE result follows from the NE result using the standard conversion from approximate NE to approximate WSNE.
We first compute a sufficiently accurate approximate NE in the flow representation and decompose each player's flow into a polynomial-support distribution over pure allocations.
Then we remove actions whose payoff is too far below the player's best-response value, as in the proof of \cref{lm:wsne2ne}.
The best-response value can be computed efficiently in the flow representation.
For fixed opponents' flows, player $i$'s payoff is linear over the unit-flow polytope, so a best response is a maximum-weight path in the acyclic graph $D_i$.
This can be found by dynamic programming.
Therefore the support-cleaning step can be implemented in polynomial time, giving PPAD membership for computing an $\eps$-WSNE (\cref{cor:membership-wsne}).

Finally, uniform tie-breaking is used only in the payoff-computation step.
The rest of the membership proof uses the flow representation, concavity in each player's own variables, and separation for the flow polytopes.
Therefore the same argument applies to non-uniform tie-breaking rules when the relevant expected shares, utilities, and gradients have polynomial-size arithmetic circuits.
It does not automatically apply to arbitrary succinct non-uniform tie-breaking rules, because evaluating a player's expected share of a battlefield may itself be computationally hard; this is illustrated in \cref{app:nonuniform-tiebreaking-hardness}.

\section{Hardness for Computing Approximate Equilibrium}
\label{sec:hardness}

We first focus on uniform tie-breaking.
For each battlefield, all players who allocate the maximum resources have an equal chance of winning the battlefield. Formally, if the players in set $S \subseteq [n]$ tie for battlefield $j$, then each $i \in S$ has a $1/|S|$ probability of winning battlefield $j$ and has an expected payoff of $v_{ij}/|S|$ from this battlefield.

We prove PPAD-hardness for computing a constant-approximate WSNE in Blotto games. The corresponding hardness result for approximate NE, with approximation inverse-linear in the number of players, will then follow from \cref{lm:wsne2ne}.

\begin{theorem}\label{thm:hardness-unif-wsne}
It is PPAD-hard to compute a $\delta$-WSNE in Colonel Blotto games with uniform tie-breaking for sufficiently small constant $\delta > 0$. The result holds even if the players in the Colonel Blotto games have exactly (or at most) three resources.
\end{theorem}

\begin{proof}
We do a reduction from the \PURE{} problem (\cref{def:pure}). We are given an instance of \PURE{}. We use the same notation as \cref{def:pure}. Let $V$ denote the set of nodes/variables and $\val{v}$ the value of $v \in V$.
The variables are related through the \NOT{}, \AND{}, and \PURIFY{} gates. Every variable is the input to and the output of exactly one gate.
Throughout the proof, we will use several constants. These are listed in \cref{tab:reduc-params}. Their use will be explained when they are introduced in the proof.

\begin{table}[htbp]
    \centering
    \begin{tabular}{c|c|c}
        parameter & value & description \\
        \hline
        $\eps$ & $10^{-5}$ & unnormalized approximation parameter for $\eps$-WSNE \\
        $\alpha$ & $1/50$ & truncation threshold parameter \\
        $\rho_{\GuP}$ & $10^{-4}$ & reserve battlefield value for guards \\
        $\rho$ & $1/1400$ & reserve battlefield value for all other players \\
        $\eta_{\DeP}$ & $1/20$ & detector's value for the first auxiliary battlefield \\
        $\eta_{\BlP}$ & $11/10$ & blocker's value for the second auxiliary battlefield \\
        $\eta_{\PrP}$ & $3/10$ & pressure player's value for the output-side bit battlefield \\
        $k_{\NOT}$ & 2 & number of pressure players on each side of \NOT{} \\
        $k_{\AND,0}$ & 3 & number of pressure players on the $0$-side of each input of \AND{} \\
        $k_{\AND,1}$ & 1 & number of pressure players on the $1$-side of each input of \AND{} \\
        $k_{\PURIFY}$ & 4 & number of pressure players on each side of \PURIFY{} \\
        $k_{\BIAS}$ & 2 & number of bias players on each side of \PURIFY{} \\
        $R$ & 4 & number of reserve battlefields per player \\
    \end{tabular}
    \caption{Parameters used in the reduction}
    \label{tab:reduc-params}
\end{table}

Given the instance of the \PURE{} problem, we construct an instance of the Blotto game that encodes a solution to the \PURE{} instance in every $\eps$-WSNE of the game, for sufficiently small constant $\eps > 0$.
For the analysis below, we first consider a slightly modified tie-breaking convention: if all players allocate zero resources to a battlefield, then every player receives payoff zero from that battlefield; otherwise, ties among the highest allocators are broken uniformly.
For now, we also use unnormalized payoffs in the construction.
These two assumptions make the analysis cleaner.
Let the unnormalized approximation parameter be $\eps = 10^{-5}$.
In the construction, we ensure that every player’s total payoff is bounded by $3$, so normalization only changes the approximation parameter by a constant factor.
At the end of the proof, we explain how to transfer the result back to the usual uniform tie-breaking convention, where an empty battlefield is split uniformly among all players.

Our constructed game will have $\Theta(|V|)$ players and battlefields.
The players will be of two types: {primary} and {auxiliary}; the auxiliary players will further have five roles: detectors, blockers, guards, pressure players, and bias players.
The battlefields will be of three types: {bit}, {auxiliary}, and {reserve}.
There will be exactly $|V|$ primary players and $2|V|$ bit battlefields, and the interaction of the primary players with the bit battlefields will encode the solution of the \PURE{} instance.
There will be $\Theta(|V|)$ auxiliary players and $4|V|$ auxiliary battlefields, and the interaction of the auxiliary players with the bit and auxiliary battlefields will help implement the logic of the \NOT{}, \AND{}, and \PURIFY{} gates. Finally, $\Theta(|V|)$ reserve battlefields give outside options to players and ensure that no player unnecessarily interacts with battlefields for which it has a zero value, which makes the proof cleaner.

We now describe the primary players and the bit battlefields.
For each node $v \in V$ of the \PURE{} instance, we have a primary player named $\PrimP^v$. Corresponding to node $v \in V$, there are also two bit battlefields named $\BitB^v_0$ and $\BitB^v_1$.
The player $\PrimP^v$ has $3$ resources and has a value of $1$ for each of the two battlefields $\BitB^v_0$ and $\BitB^v_1$. The player $\PrimP^v$ also has a value of $\rho$ for a few specific reserve battlefields; we will discuss this in more detail when we describe the reserve battlefields. The player $\PrimP^v$ has zero value for all other battlefields, which includes all auxiliary battlefields and the bit battlefields of the form $\BitB^{u}_s$ for $u \in V \setminus \{v\}$ and $s \in \{0,1\}$.

We now describe the auxiliary players and auxiliary battlefields, first at a high level to provide an overview of the overall construction, then formally.
There are four auxiliary battlefields per input variable of each gate; hence, there are four auxiliary battlefields for each \NOT{} and \PURIFY{} gate, and eight auxiliary battlefields for each \AND{} gate.
Further, for each input variable $u$ of a gate, two of the four auxiliary battlefields help implement the $\val{u}=0$ case, and two help implement the $\val{u}=1$ case.

The auxiliary players have five roles: detectors $\DeP$, blockers $\BlP$, guards $\GuP$, pressure players $\PrP$, and bias players $\BiP$. The names roughly correspond to their behaviors.
Each auxiliary player is associated with exactly one gate; and within that gate, it is associated with exactly one input variable, say $u$; and for that input variable $u$, it is associated with exactly one of the cases among $\val{u}=0$ and $\val{u}=1$.
Detectors, blockers, guards, and pressure players are present in each of the \NOT{}, \AND{}, and \PURIFY{} gates. Bias players are used only in the \PURIFY{} gate. The behavior of each role is largely similar across the three types of gates. We next formally describe the auxiliary players and battlefields for each of the three gate types.
The values of the players are also summarized in \cref{tab:cb-players} for easier reference. The values for the reserve battlefields will be described later in the proof. All unspecified values for bit and auxiliary battlefields are zero.

\medskip
\noindent
$v = \NOT(u)$.
We have already introduced the primary players $\PrimP^u$ and $\PrimP^v$ and the bit battlefields $\BitB^u_0$, $\BitB^u_1$, $\BitB^v_0$, and $\BitB^v_1$ that interact with this \NOT{} gate.

We have $2 \times 2 = 4$ auxiliary battlefields named $\AuxB^{u}_{s,t}$ for $s \in \{0,1\}$ and $t \in [2]$, where the $\AuxB^{u}_{0,t}$ battlefields help implement the logic for the $\val{u} = 0$ case and $\AuxB^{u}_{1,t}$ for the $\val{u} = 1$ case.

We have $2 (1 + 1 + 1 + k_{\NOT})$ auxiliary players: for each side $s \in \{0,1\}$, one detector $\DeP^u_s$, one blocker $\BlP^u_s$, one guard $\GuP^u_s$, and $k_{\NOT}$ identical pressure players of class $\PrP^u_s$.
The auxiliary players on side $s=0$ help implement the $\val{u}=0$ case, and those on side $s=1$ help implement the $\val{u}=1$ case.
We next specify the resources and values of all the auxiliary players in the gate, except for their values for the reserve battlefields, which we will do later. Let $s \in \{0,1\}$; the construction is symmetric across $s$.
\begin{itemize}
    \item Detector $\DeP^u_s$. It has $1$ resource and has a value of $1$ for the bit battlefield $\BitB^u_s$ and a value of $\eta_{\DeP}$ for the auxiliary battlefield $\AuxB^{u}_{s,1}$.

    \item Blocker $\BlP^u_s$. It has $2$ resources and has a value of $1$ for the auxiliary battlefield $\AuxB^{u}_{s,1}$ and a value of $\eta_{\BlP}$ for the auxiliary battlefield $\AuxB^{u}_{s,2}$.

    \item Guard $\GuP^u_s$. It has $1$ resource and has a value of $1$ for the auxiliary battlefield $\AuxB^{u}_{s,2}$.

    \item Pressure players $\PrP^u_s$; there are $k_{\NOT}$ of them. Each player has $1$ resource and has a value of $1$ for the auxiliary battlefield $\AuxB^{u}_{s,2}$ and a value of $\eta_{\PrP}$ for the bit battlefield $\BitB^v_{1-s}$.
\end{itemize}

\medskip
\noindent
$w = \AND(u,v)$.
The high-level ideas in the construction are similar to those of the \NOT{} gate. We have two inputs $u$ and $v$, so we make two copies of the auxiliary construction in the \NOT{} gate, but connect them to the same output $w$. Additionally, as the \AND{} gate outputs $0$ even if only one of the inputs is $0$ (and the other can be in $[0,1]$), we adjust the number of pressure players on the $0$ side and the $1$ side of the inputs to accommodate this bias towards $0$ output. The formal details are below.

We have already introduced the primary players $\PrimP^\sigma$ and the bit battlefields $\BitB^\sigma_s$ for $\sigma \in \{u,v,w\}$ and $s \in \{0,1\}$ that interact with this \AND{} gate.
We have $2 \times 2 \times 2 = 8$ auxiliary battlefields named $\AuxB^{\sigma}_{s,t}$ for $\sigma \in \{u,v\}$, $s \in \{0,1\}$, and $t \in [2]$.
For each input $\sigma \in \{u,v\}$ and side $s \in \{0,1\}$, the battlefields $\AuxB^\sigma_{s,1}$ and $\AuxB^\sigma_{s,2}$ help implement the $\val{\sigma}=s$ case.

We have $4(1+1+1) + 2(k_{\AND,0}+k_{\AND,1})$ auxiliary players: for each $\sigma \in \{u,v\}$ and $s \in \{0,1\}$, one detector $\DeP^\sigma_s$, one blocker $\BlP^\sigma_s$, one guard $\GuP^\sigma_s$, and pressure players of class $\PrP^\sigma_s$.
There are $k_{\AND,0}$ pressure players of class $\PrP^\sigma_0$ and $k_{\AND,1}$ pressure players of class $\PrP^\sigma_1$, for each $\sigma \in \{u,v\}$.
We next formally specify the resources and values of all the auxiliary players in the gate, except for their values for the reserve battlefields. Let $\sigma \in \{u,v\}$ and $s \in \{0,1\}$; the construction is symmetric across $\sigma$ and $s$, except for the number of pressure players.
\begin{itemize}
    \item Detector $\DeP^\sigma_s$. It has $1$ resource and has a value of $1$ for the bit battlefield $\BitB^\sigma_s$ and a value of $\eta_{\DeP}$ for the auxiliary battlefield $\AuxB^\sigma_{s,1}$.

    \item Blocker $\BlP^\sigma_s$. It has $2$ resources and has a value of $1$ for the auxiliary battlefield $\AuxB^\sigma_{s,1}$ and a value of $\eta_{\BlP}$ for the auxiliary battlefield $\AuxB^\sigma_{s,2}$.

    \item Guard $\GuP^\sigma_s$. It has $1$ resource and has a value of $1$ for the auxiliary battlefield $\AuxB^\sigma_{s,2}$.

    \item Pressure players $\PrP^\sigma_s$; there are $k_{\AND,0}$ of them if $s = 0$ and $k_{\AND,1}$ of them if $s = 1$, for each $\sigma$. Each player has $1$ resource and has a value of $1$ for the auxiliary battlefield $\AuxB^\sigma_{s,2}$ and a value of $\eta_{\PrP}$ for the bit battlefield $\BitB^w_s$.
\end{itemize}

\medskip
\noindent
$(v, w) = \PURIFY(u)$.
The overall construction is again similar to the \NOT{} gate.
We have only one input $u$ like the \NOT{} gate, but we now have two outputs $v$ and $w$. The design of the auxiliary battlefields, detectors, blockers, and guards is identical to the \NOT{} gate. The differences are in the count and behavior of the pressure players and the introduction of bias players.
At a high level, the pressure players associated with the $\val{u}=0$ side influence only the $\val{v}=0$ side, but do not affect the $\val{w}=0$ side of the construction. Similarly, the pressure players associated with the $\val{u}=1$ side influence only the $\val{w}=1$ side, but not the $\val{v}=1$ side. The bias players interact with the remaining $\val{w}=0$ and $\val{v}=1$ sides of the construction. We next specify the formal details.

We have already introduced the primary players $\PrimP^\sigma$ and the bit battlefields $\BitB^\sigma_s$ for $\sigma \in \{u,v,w\}$ and $s \in \{0,1\}$ that interact with this \PURIFY{} gate.
We have $2 \times 2 = 4$ auxiliary battlefields named $\AuxB^u_{s,t}$ for $s \in \{0,1\}$ and $t \in [2]$.
For each side $s \in \{0,1\}$, the battlefields $\AuxB^u_{s,1}$ and $\AuxB^u_{s,2}$ help implement the $\val{u}=s$ case.

We have $2(1+1+1+k_{\PURIFY}+k_{\BIAS})$ auxiliary players: for each side $s \in \{0,1\}$, one detector $\DeP^u_s$, one blocker $\BlP^u_s$, one guard $\GuP^u_s$, $k_{\PURIFY}$ pressure players of class $\PrP^u_s$, and $k_{\BIAS}$ bias players of class $\BiP^u_s$.
Half of the players of each role help implement the $\val{u}=0$ case and half the $\val{u}=1$ case.
We next formally specify the resources and values of all the auxiliary players in the gate, except for their values for the reserve battlefields. Let $s \in \{0,1\}$; the construction is symmetric across $s$ for the detectors, blockers, and guards.
\begin{itemize}
    \item Detector $\DeP^u_s$, blocker $\BlP^u_s$, and guard $\GuP^u_s$ have exactly the same properties as in the \NOT{} gate described earlier.

    \item Pressure players $\PrP^u_s$; there are $k_{\PURIFY}$ of them. Each player has $1$ resource and has a value of $1$ for the auxiliary battlefield $\AuxB^u_{s,2}$. In addition, if $s=0$, it has a value of $\eta_{\PrP}$ for the bit battlefield $\BitB^v_0$; if $s=1$, it has a value of $\eta_{\PrP}$ for the bit battlefield $\BitB^w_1$.

    \item Bias players $\BiP^u_s$; there are $k_{\BIAS}$ of them. Each player has $1$ resource and has a value of $1$ for the bit battlefield $\BitB^w_0$ if $s = 0$, and a value of $1$ for the bit battlefield $\BitB^v_1$ if $s = 1$.
\end{itemize}
Remember that the auxiliary players of all three gates have zero values for bit and auxiliary battlefields for which the values have not been explicitly specified above.

\begin{table}[htbp]
    \small
    \centering
    \begin{tabular}{| c || c || c || c | c |}
        \hline
        \multicolumn{5}{| c |}{primary players; \quad $v \in V$} \\
        \hline
         & count & resources & $\BitB^v_0$ & $\BitB^v_1$ \\
        \hline
        $\PrimP^v$ & 1 & 3 & 1 & 1 \\
        \hline
    \end{tabular}
    \medskip

    \begin{tabular}{| c || c || c || c | c | c | c |}
        \hline
        \multicolumn{7}{|c|}{auxiliary players for $v = \NOT(u)$; \quad $s \in \{0,1\}$} \\
        \hline
         & count & resources & $\BitB^u_s$ & $\AuxB^{u}_{s,1}$ & $\AuxB^{u}_{s,2}$ & $\BitB^v_{1-s}$ \\
        \hline
        $\DeP^u_s$ & 1 & 1 & 1 & $\eta_{\DeP}$ & 0 & 0 \\
        $\BlP^u_s$ & 1 & 2 & 0 & 1 & $\eta_{\BlP}$ & 0 \\
        $\GuP^u_s$ & 1 & 1 & 0 & 0 & 1 & 0 \\
        $\PrP^u_s$ & $k_{\NOT}$ & 1 & 0 & 0 & 1 & $\eta_{\PrP}$ \\
        \hline
    \end{tabular}
    \medskip

    \begin{tabular}{| c || c || c || c | c | c | c | c | c | c | c |}
        \hline
        \multicolumn{11}{|c|}{auxiliary players for $w = \AND(u, v)$; \quad $\sigma \in \{ u, v \}$} \\
        \hline
         & count & resources
         & $\BitB^{\sigma}_0$ & $\AuxB^{\sigma}_{0,1}$ & $\AuxB^{\sigma}_{0,2}$ & $\BitB^w_0$
         & $\BitB^{\sigma}_1$ & $\AuxB^{\sigma}_{1,1}$ & $\AuxB^{\sigma}_{1,2}$ & $\BitB^w_1$ \\
        \hline
        $\DeP^\sigma_0$ & 1 & 1
         & 1 & $\eta_{\DeP}$ & 0 & 0 & 0 & 0 & 0 & 0 \\
        $\BlP^\sigma_0$ & 1 & 2
         & 0 & 1 & $\eta_{\BlP}$ & 0 & 0 & 0 & 0 & 0 \\
        $\GuP^\sigma_0$ & 1 & 1
         & 0 & 0 & 1 & 0 & 0 & 0 & 0 & 0 \\
        $\PrP^\sigma_0$ & $k_{\AND,0}$ & 1
         & 0 & 0 & 1 & $\eta_{\PrP}$ & 0 & 0 & 0 & 0 \\
        $\DeP^\sigma_1$ & 1 & 1
         & 0 & 0 & 0 & 0 & 1 & $\eta_{\DeP}$ & 0 & 0 \\
        $\BlP^\sigma_1$ & 1 & 2
         & 0 & 0 & 0 & 0 & 0 & 1 & $\eta_{\BlP}$ & 0 \\
        $\GuP^\sigma_1$ & 1 & 1
         & 0 & 0 & 0 & 0 & 0 & 0 & 1 & 0 \\
        $\PrP^\sigma_1$ & $k_{\AND,1}$ & 1
         & 0 & 0 & 0 & 0 & 0 & 0 & 1 & $\eta_{\PrP}$ \\
        \hline
    \end{tabular}
    \medskip

    \begin{tabular}{| c || c || c || c | c | c | c | c | c | c | c | c | c |}
        \hline
        \multicolumn{13}{|c|}{auxiliary players for $(v, w) = \PURIFY(u)$} \\
        \hline
         & count & resources
         & $\BitB^u_0$ & $\AuxB^{u}_{0,1}$ & $\AuxB^{u}_{0,2}$
         & $\BitB^u_1$ & $\AuxB^{u}_{1,1}$ & $\AuxB^{u}_{1,2}$
         & $\BitB^v_0$ & $\BitB^v_1$ & $\BitB^w_0$ & $\BitB^w_1$ \\
        \hline
        $\DeP^u_0$ & 1 & 1
         & 1 & $\eta_{\DeP}$ & 0 & 0 & 0 & 0 & 0 & 0 & 0 & 0 \\
        $\BlP^u_0$ & 1 & 2
         & 0 & 1 & $\eta_{\BlP}$ & 0 & 0 & 0 & 0 & 0 & 0 & 0 \\
        $\GuP^u_0$ & 1 & 1
         & 0 & 0 & 1 & 0 & 0 & 0 & 0 & 0 & 0 & 0 \\
        $\DeP^u_1$ & 1 & 1
         & 0 & 0 & 0 & 1 & $\eta_{\DeP}$ & 0 & 0 & 0 & 0 & 0 \\
        $\BlP^u_1$ & 1 & 2
         & 0 & 0 & 0 & 0 & 1 & $\eta_{\BlP}$ & 0 & 0 & 0 & 0 \\
        $\GuP^u_1$ & 1 & 1
         & 0 & 0 & 0 & 0 & 0 & 1 & 0 & 0 & 0 & 0 \\
        $\PrP^u_0$ & $k_{\PURIFY}$ & 1
         & 0 & 0 & 1 & 0 & 0 & 0 & $\eta_{\PrP}$ & 0 & 0 & 0 \\
        $\PrP^u_1$ & $k_{\PURIFY}$ & 1
         & 0 & 0 & 0 & 0 & 0 & 1 & 0 & 0 & 0 & $\eta_{\PrP}$ \\
        $\BiP^u_0$ & $k_{\BIAS}$ & 1
         & 0 & 0 & 0 & 0 & 0 & 0 & 0 & 0 & 1 & 0 \\
        $\BiP^u_1$ & $k_{\BIAS}$ & 1
         & 0 & 0 & 0 & 0 & 0 & 0 & 0 & 1 & 0 & 0 \\
        \hline
    \end{tabular}
    \caption{Player counts, resources, and values for bit and auxiliary battlefields. Values for reserve battlefields are described separately.}
    \label{tab:cb-players}
\end{table}

\medskip
We now describe the reserve battlefields. We have already added all the players in the game; we only need to specify the values of these players for the reserve battlefields.
Let $n$ be the total number of players in our construction.
Notice that we have $|V|$ primary players corresponding to the $|V|$ variables of the \PURE{} instance. Further, as every variable is the input to exactly one gate, there are at most $|V|$ gates. For each gate, we added a constant number of auxiliary players. Therefore, the total number of players is $n = \Theta(|V|)$.
For each player $i \in [n]$, we add a set $\cR_i$ of reserve battlefields, where $|\cR_i| = R = 4$. So, we add a total of $R n = \Theta(|V|)$ reserve battlefields.
If a player $i$ is not a guard, then it has value $\rho$ for each battlefield in $\cR_i$ and value $0$ for all other reserve battlefields. If player $i$ is a guard, then it has value $\rho_{\GuP}$ for each battlefield in $\cR_i$ and value $0$ for all other reserve battlefields.

\medskip
This completes our construction of the Blotto game based on the \PURE{} instance. We now describe how to map an $\eps$-WSNE of the Blotto game to a solution of the \PURE{} instance; correctness will be proven later.
Fix an $\eps$-WSNE of this game. Let $p^v_1$ be the probability that the primary player $\PrimP^v$ for $v \in V$ puts at least two resources on the bit battlefield $\BitB^v_1$. We set the solution $\valonly$ of the \PURE{} instance as
\begin{equation}\label{eq:mapping}
    \val{v} = \frac{ \min(1-\alpha, \max(\alpha, p^v_1)) - \alpha }{1 - 2\alpha}.
\end{equation}

\paragraph{Analysis.}
We show that the Blotto game we constructed encodes a solution to the \PURE{} instance in every $\eps$-WSNE.
Assume that the players are playing an $\eps$-WSNE.
First, we show that no player allocates resources to battlefields for which they have zero value.

\begin{lemma}\label{lm:reserve}
No player puts positive probability on any action that allocates resources to battlefields with zero value.
\end{lemma}
\begin{proof}
We will make use of the reserve battlefields to prove this claim.
For the sake of contradiction, assume there are players who assign positive probability to pure actions that allocate resources to zero-valued battlefields, and let the set of such players be $S$.
Fix an $i \in S$. Let $\ba_i$ be a supported pure action of player $i$ that allocates a resource to a zero-valued battlefield, say $j$.

By construction, player $i$ has the set $\cR_i$ of $R = 4$ reserve battlefields for which it has a positive value of at least $\rho_{\GuP}$. Since player $i$ has zero value for battlefield $j$, we get $j \notin \cR_i$. Further, as every player has at most $3$ resources, and as player $i$ has allocated one of those resources to battlefield $j$ in the pure action $\ba_i$, there must be at least $R - (3 - 1) = 2$ reserve battlefields in $\cR_i$ without any resources allocated by player $i$ in action $\ba_i$.

Let $q_i$ be the expected number of opponent resources placed on battlefields in $\cR_i$.
As argued above, at least two battlefields in $\cR_i$ do not get any resource in action $\ba_i$.
Among these two battlefields, at least one of them, say $j'$, gets expected number of opponent resources at most $q_i/2$.
Consider the deviation from $\ba_i$ that moves the resource from zero-valued battlefield $j$ to reserve battlefield $j'$.
Let $X$ be the random variable denoting the number of opponent resources on $j'$. If $X \le 1$, then after the deviation, player $i$ obtains a payoff of at least $\rho_{\GuP}/2$ from $j'$. Since $\Exp[X] \le q_i/2$, Markov's inequality gives
\[
    \Prob[X \ge 2] \le \frac{\Exp[X]}{2} \le \frac{q_i}{4} \implies \Prob[X \le 1] \ge 1 - \frac{q_i}{4}.
\]
Therefore, the deviation gains at least
\[
    \frac{\rho_{\GuP}}{2} \left( 1 - \frac{q_i}{4} \right),
\]
which must be at most $\eps$. Hence
\[
    \frac{\rho_{\GuP}}{2} \left( 1 - \frac{q_i}{4} \right) \le \eps
    \iff q_i \ge 4 \left( 1 - \frac{2 \eps}{\rho_{\GuP}} \right)
    = 4 (1 - 2 \cdot 10^{-5} \cdot 10^4)
    = 4 (1 - 0.2) = 3.2,
\]
where we substituted $\eps = 10^{-5}$ and $\rho_{\GuP} = 10^{-4}$. Thus, every player $i \in S$ receives strictly more than $3$ expected opponent resources on its reserve battlefields $\cR_i$. As the reserve battlefields for which each player has positive values are disjoint, i.e., $\cR_i \cap \cR_{i'} = \emptyset$ for $i \neq i'$, summing over $i \in S$, we get
\[
    \sum_{i \in S} q_i > 3 |S|.
\]
On the other hand, every resource counted in $\sum_{i \in S} q_i$ is placed by some opponent on another player's positive-valued reserve battlefields. Such a battlefield has value $0$ for the player placing the resource. Therefore, every player contributing resources to this sum also belongs to $S$.
Since every player has at most three resources, the total expected number of resources that players in $S$ can place on other players' positive-valued reserve battlefields is at most $3|S|$, which gives us a contradiction.
\end{proof}

Given \cref{lm:reserve}, let's upper bound the maximum number of auxiliary players that can allocate resources to any bit battlefield $\BitB^u_s$ for $u \in V$ and $s \in \{0,1\}$. We claim it is $\Delta = 7$.
First, notice that blockers and guards have zero value for all bit battlefields.
For detectors, there is exactly one player $\DeP^u_s$ that has positive value for $\BitB^u_s$.
For pressure players, a maximum of six players can have a positive value for $\BitB^u_s$, which happens if $u$ is the output of an \AND{} gate and $s = 0$.
Finally, for bias players, a maximum of two players can have a positive value for $\BitB^u_s$, which can happen if $u$ is the output of a \PURIFY{} gate, but in that case there will be no pressure players who value $\BitB^u_s$. Overall, the upper bound of $\Delta = 7$ holds.

For $u \in V$ and $s \in \{0,1\}$, let $p^u_s$ be the probability that the primary player $\PrimP^u$ puts at least two resources on the bit battlefield $\BitB^u_s$. As player $\PrimP^u$ has three resources, we have
\[
    p^u_0 + p^u_1 \le 1, \text{ for all $u \in V$. }
\]
The next lemma proves that in fact $p^u_0 + p^u_1 = 1$.
\begin{lemma}\label{lm:primary}
For any $u \in V$, the primary player $\PrimP^u$ plays only the following two actions with positive probability: allocate resources $(2, 1)$ or $(1, 2)$ to the bit battlefields $(\BitB^u_0, \BitB^u_1)$.
\end{lemma}
\begin{proof}
From \cref{lm:reserve}, we know that player $\PrimP^u$ allocates resources only to $\BitB^u_0$, $\BitB^u_1$, or its positive-valued reserve battlefields.

First, we argue that $\PrimP^u$ will not allocate three resources to $\BitB^u_0$ or $\BitB^u_1$ with positive probability. For contradiction, say $\PrimP^u$ allocates three resources to $\BitB^u_s$ with positive probability. For this action, it gets a payoff of $1$ from $\BitB^u_s$ but $0$ from $\BitB^u_{1-s}$. However, it can still get a payoff of $1$ from $\BitB^u_s$ by putting just two resources, as every other player that has a positive value for $\BitB^u_s$ has only one resource. On the other hand, by allocating the third resource to $\BitB^u_{1-s}$, it gets an additional payoff of at least $1/(1 + \Delta) = 1/8 > \eps$, where $\Delta = 7$ is the maximum number of auxiliary players that have a positive value for $\BitB^u_{1-s}$, as argued earlier.

Second, we claim that $\PrimP^u$ will not allocate any resources to reserve battlefields with positive probability. For contradiction, say it does so in pure action $a$.
Notice that in the overall mixed strategy, as $\PrimP^u$ has only three resources, with probability at least $1/2$, $\PrimP^u$ allocates at most one resource to either $\BitB^u_0$ or $\BitB^u_1$; say $\BitB^u_s$ is that battlefield.

We claim that the detector $\DeP^u_s$ must allocate its resource to $\BitB^u_s$ with probability $1$. With probability at least $1/2$, $\DeP^u_s$ has to compete with at most one resource from player $\PrimP^u$ and at most one resource each from at most $\Delta - 1$ auxiliary players. So, $\DeP^u_s$ gets an expected payoff of at least $1/(2(1+\Delta)) = 1/16$, which is more than $\eps$ greater than its value of $\eta_{\DeP} = 1/20$ for the auxiliary battlefield $\AuxB^u_{s,1}$.

Let's now focus again on the primary player and the action $a$. It allocates at most two resources to $\BitB^u_0$ and $\BitB^u_1$ in pure action $a$.
We have three cases.
If $\BitB^u_s$ gets zero resources in $a$, then moving the resource allocated to the reserve battlefield to $\BitB^u_s$ increases the payoff by at least $1/(1+\Delta) - \rho > \eps$.
Similarly, if $\BitB^u_s$ gets only one resource in $a$, then moving the resource allocated to the reserve battlefield to $\BitB^u_s$ increases the payoff by at least $(1 - 1/2) - \rho > \eps$.
On the other hand, if $\BitB^u_s$ gets two resources in $a$, then $\BitB^u_{1-s}$ gets zero resources in $a$; moving the resource allocated to the reserve battlefield to $\BitB^u_{1-s}$ increases the payoff by at least $1/(1+\Delta) - \rho > \eps$.
\end{proof}

\cref{lm:primary} tells us that player $\PrimP^u$ does resource allocation $(2,1)$ and $(1,2)$ to battlefields $(\BitB^u_0, \BitB^u_1)$ with probabilities $p^u_0$ and $p^u_1 = 1 - p^u_0$, respectively.
Combining this with \eqref{eq:mapping}, we know that if $p^u_1 = 1 - p^u_0 \ge 1 - \alpha$, then $\val{u} = 1$, and if $p^u_0 = 1 - p^u_1 \ge 1 - \alpha$, then $\val{u} = 0$.
We now analyze the behavior of the detectors given the actions of the primary players.

\begin{lemma}\label{lm:detector}
For any $u \in V$ and $s \in \{0,1\}$, the detector $\DeP^u_s$ plays the following action depending upon $p^u_s$:
\begin{itemize}
    \item if $p^u_s \le 11/20$, then $\DeP^u_s$ puts its only resource on the bit battlefield $\BitB^u_s$ with probability $1$;
    \item if $p^u_s \ge 20/21$, then $\DeP^u_s$ puts its only resource on the auxiliary battlefield $\AuxB^u_{s,1}$ with probability $1$.
\end{itemize}
\end{lemma}
\begin{proof}
Fix $p^u_s$.
Let's compute bounds on the maximum and minimum payoffs that player $\DeP^u_s$ gets for choosing the bit battlefield $\BitB^u_s$.
Notice that with probability $p^u_s$ the primary player $\PrimP^u$ allocates $2$ resources to $\BitB^u_s$, and in this case $\DeP^u_s$ gets a payoff of $0$ for playing $\BitB^u_s$.
On the other hand, with probability $1 - p^u_s$, player $\PrimP^u$ allocates $1$ resource to $\BitB^u_s$.
In the best-case scenario, no other auxiliary player allocates any resources to $\BitB^u_s$, and the detector $\DeP^u_s$ gets an expected payoff of $(1 - p^u_s)/2$.
In the worst-case scenario, all auxiliary players who have positive value for $\BitB^u_s$ allocate resources to it, and the detector $\DeP^u_s$ gets an expected payoff of $(1 - p^u_s)/(1+\Delta) = (1 - p^u_s)/8$.

Let us now focus on the auxiliary battlefield $\AuxB^u_{s,1}$. Only the detector $\DeP^u_s$ and blocker $\BlP^u_s$ have positive value for this battlefield.
Although the blocker $\BlP^u_s$ has two resources, we claim that it never allocates both of them to $\AuxB^u_{s,1}$, because allocating both resources to $\AuxB^u_{s,1}$ gives it payoff $1$, while the alternative action of allocating both resources to $\AuxB^u_{s,2}$ gives it payoff $\eta_{\BlP} = 11/10 > 1 + \eps$.
So, detector $\DeP^u_s$ competes with either zero or one resource on battlefield $\AuxB^u_{s,1}$ and gets a payoff of either $\eta_{\DeP} = 1/20$ or $\eta_{\DeP}/2 = 1/40$, respectively.

Player $\DeP^u_s$ must play $\BitB^u_s$ with probability $1$ if the minimum possible payoff of $\BitB^u_s$ is $\eps$ more than the maximum possible payoff of $\AuxB^u_{s,1}$, i.e.,
\[
    \frac{1 - p^u_s}{1 + \Delta} > \eta_{\DeP} + \eps
    \iff \frac{1 - p^u_s}{8} > \frac{1}{20} + 10^{-5}
    \iff p^u_s < 1 - \frac{8}{20} - 8 \cdot 10^{-5}
    \impliedby p^u_s \le \frac{11}{20}.
\]
On the other hand, player $\DeP^u_s$ must play $\AuxB^u_{s,1}$ with probability $1$ if the maximum possible payoff of $\BitB^u_s$ is $\eps$ less than the minimum possible payoff of $\AuxB^u_{s,1}$, i.e.,
\[
    \frac{1 - p^u_s}{2} + \eps < \frac{\eta_{\DeP}}{2}
    \iff \frac{1 - p^u_s}{2} + 10^{-5} < \frac{1}{40}
    \iff p^u_s > 1 - \frac{1}{20} + 2 \cdot 10^{-5}
    \impliedby p^u_s \ge \frac{20}{21}. \qedhere
\]
\end{proof}

We now analyze the behavior of blockers and guards given the action of the detectors.

\begin{lemma}\label{lm:blocker}
For any $u \in V$ and $s \in \{0,1\}$,
the blocker $\BlP^u_s$ randomizes only over the two resource allocations $(1,1)$ and $(0,2)$ over the auxiliary battlefields $(\AuxB^u_{s,1}, \AuxB^u_{s,2})$.
Further, the blocker $\BlP^u_s$ and guard $\GuP^u_s$ play the following actions depending upon the action of the detector $\DeP^u_s$:
\begin{itemize}
    \item if detector $\DeP^u_s$ puts its only resource on the bit battlefield $\BitB^u_s$ with probability $1$,
    then the blocker $\BlP^u_s$ allocates one resource each to $\AuxB^u_{s,1}$ and $\AuxB^u_{s,2}$ with probability $1$,
    and the guard $\GuP^u_s$ allocates its only resource to $\AuxB^u_{s,2}$ with probability $1$;
    \item if detector $\DeP^u_s$ puts its only resource on the auxiliary battlefield $\AuxB^u_{s,1}$ with probability $1$,
    then the blocker $\BlP^u_s$ allocates both resources to $\AuxB^u_{s,2}$ with probability $\beta \ge 1 - 6 (\rho_{\GuP} + \eps)$ and allocates one resource each to $\AuxB^u_{s,1}$ and $\AuxB^u_{s,2}$ with probability $1 - \beta$, and the guard $\GuP^u_s$ mixes among the strategies of allocating its resource to $\AuxB^u_{s,2}$ and its reserve battlefields.
\end{itemize}
\end{lemma}

\begin{proof}
First, notice that the blocker $\BlP^u_s$ never allocates any of its two resources to a reserve battlefield. If it allocates resources to reserve battlefields, then it does not allocate any resources to at least one of the two auxiliary battlefields $\AuxB^u_{s,1}$ or $\AuxB^u_{s,2}$.
Detector $\DeP^u_s$, with one resource, is the only other player that has a positive value for $\AuxB^u_{s,1}$, so the blocker $\BlP^u_s$ gets a payoff of at least $1/2 > \rho + \eps$ from $\AuxB^u_{s,1}$ by allocating a resource to it.
Similarly, guard $\GuP^u_s$ and at most four pressure players of class $\PrP^u_s$ (four is the maximum number of pressure players that have positive value for any auxiliary battlefield, which in particular happens for the \PURIFY{} gate), each with one resource, are the only other players that have a positive value for $\AuxB^u_{s,2}$, so the blocker $\BlP^u_s$ gets a payoff of at least $\eta_{\BlP}/6 = 11/60 > \rho + \eps$ from $\AuxB^u_{s,2}$ by allocating a resource to it.
Overall, moving the resource allocated to a reserve battlefield to one of the two auxiliary battlefields, whichever has not been allocated a resource, increases the payoff by more than $\eps$.
In the proof of \cref{lm:detector}, we argued that the blocker never allocates both of its resources to $\AuxB^u_{s,1}$; the argument uses the fact that allocating both resources to $\AuxB^u_{s,2}$ is more than $\eps$ better than allocating both resources to $\AuxB^u_{s,1}$. In summary, the blocker $\BlP^u_s$ randomizes only over the two actions $(1,1)$ and $(0,2)$ of resource allocations over the auxiliary battlefields $(\AuxB^u_{s,1}, \AuxB^u_{s,2})$.

Let us look at the case when the detector $\DeP^u_s$ puts its only resource on the bit battlefield $\BitB^u_s$ with probability $1$.
Then the blocker $\BlP^u_s$ has no competition in the auxiliary battlefield $\AuxB^u_{s,1}$, and can get a payoff of $1$ from it by allocating one resource.
On the other hand, in the battlefield $\AuxB^u_{s,2}$, blocker $\BlP^u_s$ has to compete with at most one guard $\GuP^u_s$ and at most four pressure players of class $\PrP^u_s$, and receives a payoff of at least $\eta_{\BlP}/6$ by allocating a resource. So, by playing action $(1,1)$, the blocker gets a payoff of at least $1 + \eta_{\BlP}/6 = 1 + 11/60 = 71/60$. On the other hand, by playing action $(0,2)$, the blocker gets a payoff of $11/10 = 66/60$.
Therefore, the blocker plays $(1,1)$ over $(\AuxB^u_{s,1}, \AuxB^u_{s,2})$ if the detector plays $\BitB^u_s$.
Finally, if the blocker allocates only one resource to $\AuxB^u_{s,2}$, then the guard has to compete with the blocker and at most four pressure players for $\AuxB^u_{s,2}$, and gets a payoff of at least $1/6$ by allocating its resource to $\AuxB^u_{s,2}$. On the other hand, a reserve battlefield gives a payoff of just $\rho_{\GuP} < 1/6 - \eps$. So, the guard allocates its resource to $\AuxB^u_{s,2}$ with probability $1$.

Let us now look at the case when the detector $\DeP^u_s$ puts its only resource on the auxiliary battlefield $\AuxB^u_{s,1}$ with probability $1$.
Let $\beta$ be the probability that blocker $\BlP^u_s$ plays the action $(0,2)$ over battlefields $(\AuxB^u_{s,1}, \AuxB^u_{s,2})$.
We claim that $\beta \ge 1 - 6 (\rho_{\GuP} + \eps)$.
For contradiction, let us assume that $\beta < 1 - 6 (\rho_{\GuP} + \eps)$. So, with probability $1 - \beta > 6 (\rho_{\GuP} + \eps)$ the blocker plays action $(1,1)$ and allocates only one resource to $\AuxB^u_{s,2}$.
In this case, the guard $\GuP^u_s$ has to compete with one resource of the blocker and at most four pressure players for battlefield $\AuxB^u_{s,2}$, and gets an expected payoff of at least $(1 - \beta)/6 > \rho_{\GuP} + \eps$, which is $\eps$ more than the payoff of $\rho_{\GuP}$ for the reserve battlefields. So, the guard plays $\AuxB^u_{s,2}$ with probability $1$ if $\beta < 1 - 6 (\rho_{\GuP} + \eps)$. Further, if the guard plays $\AuxB^u_{s,2}$ with probability $1$, then the expected payoff of the blocker for playing $(0,2)$ is $\eta_{\BlP} = 11/10$, and the expected payoff for playing $(1,1)$ is at most $1/2 + \eta_{\BlP}/2 = 21/20$, where $1/2$ comes from $\AuxB^u_{s,1}$, shared with $\DeP^u_s$, and at most $\eta_{\BlP}/2$ comes from $\AuxB^u_{s,2}$, shared with at least $\GuP^u_s$. As $11/10 - 21/20 = 1/20 > \eps$, the blocker must play $(0,2)$ with probability $\beta = 1$, which contradicts that $\beta < 1 - 6 (\rho_{\GuP} + \eps)$.
\end{proof}

We next analyze the behavior of the pressure players as a function of the blockers.

\begin{lemma}\label{lm:pressure}
For any $u \in V$ and $s \in \{0,1\}$, the pressure players $\PrP^u_s$ play the following actions depending upon the action of the blocker $\BlP^u_s$:
\begin{itemize}
    \item if blocker $\BlP^u_s$ allocates one resource each to $\AuxB^u_{s,1}$ and $\AuxB^u_{s,2}$ with probability $1$,
    then all pressure players $\PrP^u_s$ allocate their resources to $\AuxB^u_{s,2}$ with probability $1$;
    \item if blocker $\BlP^u_s$ allocates two resources to $\AuxB^u_{s,2}$ with probability $\beta \ge 1 - 6(\rho_{\GuP} + \eps)$,
    then no pressure player of class $\PrP^u_s$ allocates any resources to $\AuxB^u_{s,2}$; instead, they mix over their positive-valued reserve battlefields and their positive-valued bit battlefield on the output side of the gate.
\end{itemize}
\end{lemma}
\begin{proof}
Let us consider the case where the blocker $\BlP^u_s$ allocates one resource each to $\AuxB^u_{s,1}$ and $\AuxB^u_{s,2}$ with probability $1$.
Then any pressure player of class $\PrP^u_s$ has to compete with one resource of the blocker, at most one guard, and at most three other pressure players on the battlefield $\AuxB^u_{s,2}$; hence, it gets a payoff of at least $1/6$ for playing $\AuxB^u_{s,2}$.
This payoff is clearly $\eps$ more than the value $\rho$ for the positive-valued reserve battlefields.
For the positive-valued bit battlefield on the output side, given \cref{lm:primary}, any pressure player of class $\PrP^u_s$ has to compete with at least one resource of the primary player, and therefore gets a payoff of at most $\eta_{\PrP}/2 = 3/20$. As $1/6 - 3/20 > \eps$, all these pressure players must play $\AuxB^u_{s,2}$ with probability $1$.

Let us now consider the case where the blocker $\BlP^u_s$ allocates two resources to $\AuxB^u_{s,2}$ with probability $\beta \ge 1 - 6(\rho_{\GuP} + \eps)$. Then with probability $1 - \beta$ it allocates one resource to $\AuxB^u_{s,2}$ by \cref{lm:blocker}.
Then any pressure player of class $\PrP^u_s$ gets $0$ with probability $\beta$ and $\le 1/2$ with probability $1 - \beta$ for playing $\AuxB^u_{s,2}$, so its expected utility for playing $\AuxB^u_{s,2}$ is at most $(1 - \beta)/2 \le 3(\rho_{\GuP} + \eps)$. On the other hand, it gets a payoff of $\rho$ for playing a positive-valued reserve battlefield, and $\rho - 3(\rho_{\GuP} + \eps) = 1/1400 - 3 \cdot 10^{-4} - 3 \cdot 10^{-5} > 10^{-5} = \eps$, so it will not allocate any resource to $\AuxB^u_{s,2}$.
\end{proof}

\cref{lm:pressure} tells us that if the blocker $\BlP^u_s$ puts two resources on the auxiliary battlefield $\AuxB^u_{s,2}$ with high probability, then the pressure players $\PrP^u_s$ play the auxiliary battlefield $\AuxB^u_{s,2}$ with probability $0$ and mix only over their positive-valued reserve battlefields and their positive-valued bit battlefield on the output side.
For this case, we now analyze the behavior of these pressure players as a function of the primary player on the output side.

\begin{lemma}\label{lm:pressure2}
For any $u, v \in V$ and $s, r \in \{0,1\}$ such that $u$ is an input and $v$ an output of a gate, consider the pressure players $\PrP^u_s$ who have positive values for the auxiliary battlefield $\AuxB^u_{s,2}$ and the bit battlefield $\BitB^v_r$.
If the blocker $\BlP^u_s$ allocates two resources to $\AuxB^u_{s,2}$ with probability $\beta \ge 1 - 6(\rho_{\GuP} + \eps)$
and the primary player $\PrimP^v$ has $p^v_r < 1 - \alpha$,
then every pressure player of class $\PrP^u_s$ puts its only resource on the bit battlefield $\BitB^v_r$ with probability $1$.
\end{lemma}

\begin{proof}
First, as the blocker $\BlP^u_s$ allocates two resources to battlefield $\AuxB^u_{s,2}$ with probability $\beta \ge 1 - 6(\rho_{\GuP} + \eps)$,
using \cref{lm:pressure} we know that all pressure players of class $\PrP^u_s$ mix only over their positive-valued reserve battlefields and their positive-valued bit battlefield $\BitB^v_r$.

The primary player $\PrimP^v$ allocates only one resource to $\BitB^v_r$ with probability $1 - p^v_r$.
So, with probability $1 - p^v_r$, any pressure player of class $\PrP^u_s$ has to compete with one resource of $\PrimP^v$,
at most one detector $\DeP^v_r$,
and at most five other pressure players in $\BitB^v_r$; the maximum occurs on the $s=r=0$ side of an \AND{} gate.
Therefore, it gets an expected payoff of at least $\frac{(1 - p^v_r)\eta_{\PrP}}{8}$ for playing $\BitB^v_r$. On the other hand, it gets a payoff of $\rho$ for playing a suitable reserve battlefield. As $(1 - p^v_r)\eta_{\PrP}/8 > \alpha \eta_{\PrP} / 8 > \rho + \eps$, where the last inequality follows from the constants in \cref{tab:reduc-params}, every pressure player of class $\PrP^u_s$ plays $\BitB^v_r$ with probability $1$.
\end{proof}

In a similar spirit to \cref{lm:pressure2}, we next analyze the behavior of the bias players as a function of the action of the primary player on the output side.

\begin{lemma}\label{lm:bias}
Consider $u, v \in V$ and $s \in \{0,1\}$ such that $u$ is the input of a \PURIFY{} gate and $v$ is the output for which the bias players of class $\BiP^u_s$ have value $1$ for the bit battlefield $\BitB^v_s$.
Each bias player of class $\BiP^u_s$ puts its only resource on the bit battlefield $\BitB^v_s$ with probability $1$ if $p^v_s < 1 - \alpha$.
\end{lemma}
\begin{proof}
The primary player $\PrimP^v$ allocates only one resource to $\BitB^v_s$ with probability $1 - p^v_s$.
So, with probability $1 - p^v_s$, any bias player of class $\BiP^u_s$ has to compete with one resource of $\PrimP^v$, at most one detector $\DeP^v_s$, and at most one other bias player of class $\BiP^u_s$ on the battlefield $\BitB^v_s$. Therefore, it gets an expected payoff of at least $(1 - p^v_s)/4$ for playing $\BitB^v_s$. On the other hand, it gets a payoff of $\rho$ for playing a suitable reserve battlefield. As $(1 - p^v_s)/4 > \alpha/4 > \rho + \eps$, the player plays $\BitB^v_s$ with probability $1$.
\end{proof}

In the previous lemmas, we analyzed the behavior of each auxiliary role as a function of the strategies of the primary players and other auxiliary players.
What remains is to connect these lemmas, suitably adapted to the individual gates, and prove that the gates work correctly.

\medskip
\noindent
$v = \NOT(u)$.
We prove that our construction correctly implements the \NOT{} gate.
Let us consider the case when $\val{u} = 0$.
In our game, this corresponds to $p^u_0 \ge 1 - \alpha = 49/50$ and $p^u_1 = 1 - p^u_0 \le \alpha = 1/50$.
Let us analyze the behavior of the auxiliary players and the output player $\PrimP^v$.

Let us first focus on the $s=0$ side of the gate.
As $p^u_0 \ge 49/50 \ge 20/21$, applying \cref{lm:detector}, we know that the 0-side detector $\DeP^u_0$ chooses the auxiliary battlefield $\AuxB^u_{0,1}$.
Given this, using \cref{lm:blocker}, we know that the 0-side blocker $\BlP^u_0$ plays action $(0,2)$ over $(\AuxB^u_{0,1}, \AuxB^u_{0,2})$ with probability $\beta \ge 1 - 6(\rho_{\GuP} + \eps)$.
Then, using \cref{lm:pressure}, we know that the 0-side pressure players $\PrP^u_0$ choose either the bit battlefield $\BitB^v_1$ or reserve battlefields with probability $1$.

Now, focusing on the $s=1$ side of the gate.
As $p^u_1 \le 1/50 \le 11/20$, using \cref{lm:detector}, the 1-side detector $\DeP^u_1$ chooses the bit battlefield $\BitB^u_1$.
Given this, using \cref{lm:blocker}, we know that the 1-side blocker $\BlP^u_1$ plays $(1,1)$ over $(\AuxB^u_{1,1}, \AuxB^u_{1,2})$ with probability $1$.
Then, using \cref{lm:pressure}, we know that all the 1-side pressure players $\PrP^u_1$ choose $\AuxB^u_{1,2}$ with probability $1$.

Summarizing the behavior of the pressure players, we know that the 1-side pressure players $\PrP^u_1$ play $\AuxB^u_{1,2}$, but the 0-side pressure players $\PrP^u_0$ may play either $\BitB^v_1$ or reserve battlefields. Note that there are $k_{\NOT} = 2$ pressure players of class $\PrP^u_0$.

Let us now focus on the behavior of the primary player $\PrimP^v$; we will prove that $p^v_1 \ge 1 - \alpha$. Let us assume $p^v_1 < 1 - \alpha$ for contradiction.
From \cref{lm:pressure2}, we know that if $p^v_1 < 1 - \alpha$, then the pressure players $\PrP^u_0$ choose $\BitB^v_1$ with probability $1$.
In that case, $\PrimP^v$'s payoff for playing action $(2,1)$ over $(\BitB^v_0, \BitB^v_1)$ will be $1$ from $\BitB^v_0$ and at most $1/(k_{\NOT}+1) = 1/3$ from $\BitB^v_1$, so a total of at most $4/3$.
On the other hand, $\PrimP^v$'s payoff for playing action $(1,2)$ will be $1$ from $\BitB^v_1$ and at least $1/2$ from $\BitB^v_0$, because on $\BitB^v_0$ it may have to compete with the detector $\DeP^v_0$ but does not have to compete with any pressure players $\PrP^u_1$. So, it gets a total of at least $3/2 > 4/3 + \eps$ for playing $(1,2)$. Therefore, $\PrimP^v$ must play $(1,2)$ with probability $p^v_1 = 1$ in any $\eps$-WSNE, which contradicts that $p^v_1 < 1 - \alpha$. So, $p^v_1 \ge 1 - \alpha$, which implies $\val{v} = 1$, as required.

Notice that our constructed \NOT{} gate is symmetric across $s = 0$ and $s = 1$. So, for the case $\val{u} = 1$, a symmetric argument proves that $\val{v} = 0$. Finally, \PURE{} puts no restriction on the output of the \NOT{} when the input $\val{u} \in (0, 1)$, so we do not need to analyze this case.

\medskip
\noindent
$w = \AND(u, v)$. We do a case analysis on the input values $\val{u}$ and $\val{v}$.

Case $\val{u} = 0$.
As we did in the analysis of the \NOT{} gate above, applying \cref{lm:detector,lm:blocker,lm:pressure} in sequence, we get that the pressure players $\PrP^u_1$ play the auxiliary battlefield $\AuxB^u_{1,2}$, but the pressure players $\PrP^u_0$ may play either $\BitB^w_0$ or reserve battlefields.
We claim that the primary player $\PrimP^w$ has $p^w_0 \ge 1 - \alpha$.
For contradiction, assume $p^w_0 < 1 - \alpha$.
From \cref{lm:pressure2}, we know that if $p^w_0 < 1 - \alpha$, then all the $k_{\AND,0} = 3$ pressure players of class $\PrP^u_0$ choose $\BitB^w_0$ with probability $1$.
Then, in the battlefield $\BitB^w_0$, $\PrimP^w$ has to compete with at least these $k_{\AND,0}$ pressure players, so action $(1,2)$ over $(\BitB^w_0, \BitB^w_1)$ gives it payoff at most $1/4 + 1 = 5/4$.
On the other hand, in the battlefield $\BitB^w_1$, $\PrimP^w$ has to compete with at most one detector $\DeP^w_1$ and at most $k_{\AND,1} = 1$ pressure player of class $\PrP^v_1$ from the other input $v$. So, action $(2,1)$ over $(\BitB^w_0, \BitB^w_1)$ gives it payoff at least $1 + 1/3 = 4/3 > 5/4 + \eps$.
So, $\PrimP^w$ must play $(2,1)$ with probability $p^w_0 = 1$, which contradicts $p^w_0 < 1 - \alpha$. Therefore, $p^w_0 \ge 1 - \alpha$, which implies $\val{w} = 0$, as required.

Case $\val{v} = 0$.
The behavior of the \AND{} gate and our construction of the corresponding gadget are symmetric across the two inputs $u$ and $v$, and an analysis similar to the $\val{u} = 0$ case holds.

Case $\val{u} = \val{v} = 1$.
As we did earlier, applying \cref{lm:detector,lm:blocker,lm:pressure} in sequence, we get that the pressure players $\PrP^u_0$ and $\PrP^v_0$ play auxiliary battlefields $\AuxB^u_{0,2}$ and $\AuxB^v_{0,2}$, respectively, but the pressure players $\PrP^u_1$ and $\PrP^v_1$ may play either $\BitB^w_1$ or reserve battlefields depending upon the action of the primary player $\PrimP^w$.
We claim that the primary player $\PrimP^w$ has $p^w_1 \ge 1 - \alpha$.
For contradiction, assume $p^w_1 < 1 - \alpha$.
From \cref{lm:pressure2}, we know that if $p^w_1 < 1 - \alpha$, then all $2 k_{\AND,1} = 2$ pressure players of classes $\PrP^u_1$ and $\PrP^v_1$ choose $\BitB^w_1$ with probability $1$.
Then, in the battlefield $\BitB^w_1$, $\PrimP^w$ has to compete with at least these $2 k_{\AND,1}$ pressure players, so action $(2,1)$ over $(\BitB^w_0, \BitB^w_1)$ gives it payoff at most $1 + 1/3 = 4/3$.
On the other hand, in the battlefield $\BitB^w_0$, $\PrimP^w$ has to compete with at most one detector $\DeP^w_0$, so action $(1,2)$ over $(\BitB^w_0, \BitB^w_1)$ gives it payoff at least $1/2 + 1 = 3/2 > 4/3 + \eps$.
So, $\PrimP^w$ must play $(1,2)$ with probability $p^w_1 = 1$, which contradicts $p^w_1 < 1 - \alpha$. Therefore, $p^w_1 \ge 1 - \alpha$, which implies $\val{w} = 1$, as required.

Other cases. For other values of $\val{u}$ and $\val{v}$, the output of the \AND{} gate is unrestricted.

\medskip
\noindent
$(v, w) = \PURIFY(u)$.
For $s \in \{0,1\}$, let us call the detector $\DeP^u_s$ \textit{active} if it allocates its only resource to the auxiliary battlefield $\AuxB^u_{s,1}$ with probability $1$ and \textit{inactive} if it allocates its only resource to the bit battlefield $\BitB^u_s$ with probability $1$. For the cases where we are unsure about the action of $\DeP^u_s$, we call it \textit{unspecified}.
Let us revisit the statement of \cref{lm:detector}. For $s \in \{0,1\}$, if $p^u_s \le 11/20$, then the detector $\DeP^u_s$ is inactive, but if $p^u_s \ge 20/21$, then the detector is active. When $p^u_s \in (11/20, 20/21)$, the action of the detector $\DeP^u_s$ is unspecified; however, knowing this will not be necessary for the analysis below.
Consider the following exhaustive cases depending upon the value of $p^u_1 = 1 - p^u_0$ in $[0,1]$,
\begin{align}
    p^u_1 = 1 - p^u_0 &\le \frac{1}{21}, &\text{then $\DeP^u_0$ is active and $\DeP^u_1$ is inactive,} \label{eq:purify:1} \\
    \frac{1}{21} < p^u_1 = 1 - p^u_0 &< \frac{9}{20}, &\text{then $\DeP^u_0$ is unspecified and $\DeP^u_1$ is inactive,} \\
    \frac{9}{20} \le p^u_1 = 1 - p^u_0 &\le \frac{11}{20}, &\text{then $\DeP^u_0$ and $\DeP^u_1$ are both inactive,} \\
    \frac{11}{20} < p^u_1 = 1 - p^u_0 &< \frac{20}{21}, &\text{then $\DeP^u_0$ is inactive and $\DeP^u_1$ is unspecified,} \\
    \frac{20}{21} \le p^u_1 = 1 - p^u_0 &, &\text{then $\DeP^u_0$ is inactive and $\DeP^u_1$ is active.} \label{eq:purify:2}
\end{align}
In each of the above cases, at least one of the two detectors in the gate, $\DeP^u_0$ or $\DeP^u_1$, is either active or inactive but not unspecified.
As we did in the analysis of the \NOT{} and \AND{} gates, applying \cref{lm:blocker,lm:pressure}, we get the following:
if the detector $\DeP^u_s$ is inactive, then the pressure players $\PrP^u_s$ play auxiliary battlefield $\AuxB^u_{s,2}$ with probability $1$;
if the detector $\DeP^u_s$ is active, then the pressure players $\PrP^u_s$ play either their positive-valued bit battlefield on the output side or their positive-valued reserve battlefields.
Using this, we can prove the following properties about the behavior of the primary players $\PrimP^v$ and $\PrimP^w$:
\begin{enumerate}
    \item If $\DeP^u_0$ is active, then $p^v_0 = 1 - p^v_1 \ge 1 - \alpha$, equivalently $\val{v} = 0$.

    For contradiction, assume $p^v_0 < 1 - \alpha$.
    From \cref{lm:pressure2}, we know that if $p^v_0 < 1 - \alpha$, then all the $k_{\PURIFY} = 4$ pressure players $\PrP^u_0$ choose $\BitB^v_0$ with probability $1$.
    Then, in the battlefield $\BitB^v_0$, $\PrimP^v$ has to compete with at least these $k_{\PURIFY}$ pressure players, so an action $(1,2)$ over $(\BitB^v_0, \BitB^v_1)$ gives it a payoff of at most $1/5 + 1 = 6/5$.
    On the other hand, in the battlefield $\BitB^v_1$, $\PrimP^v$ has to compete with at most one detector $\DeP^v_1$ and at most $k_{\BIAS} = 2$ bias players of class $\BiP^u_1$. So, an action $(2,1)$ over $(\BitB^v_0, \BitB^v_1)$ gives it a payoff of at least $1 + 1/4 = 5/4 > 6/5 + \eps$.
    So, $\PrimP^v$ must play $(2,1)$ with probability $p^v_0 = 1$, which contradicts $p^v_0 < 1 - \alpha$, so $p^v_0 \ge 1 - \alpha \iff \val{v} = 0$, as required.

    \item If $\DeP^u_1$ is inactive, then $p^w_0 = 1 - p^w_1 \ge 1 - \alpha$, equivalently $\val{w} = 0$.

    For contradiction, assume $p^w_0 < 1 - \alpha$.
    From \cref{lm:bias}, we know that if $p^w_0 < 1 - \alpha$, then all the $k_{\BIAS} = 2$ bias players $\BiP^u_0$ choose $\BitB^w_0$ with probability $1$.
    Then, in the battlefield $\BitB^w_0$, $\PrimP^w$ has to compete with at least these $k_{\BIAS}$ players, so an action $(1,2)$ over $(\BitB^w_0, \BitB^w_1)$ gives it a payoff of at most $1/3 + 1 = 4/3$.
    On the other hand, in the battlefield $\BitB^w_1$, $\PrimP^w$ has to compete only with at most one detector $\DeP^w_1$. Note that $\DeP^u_1$ is inactive, so the pressure players $\PrP^u_1$ play the auxiliary battlefield $\AuxB^u_{1,2}$ and not the bit battlefield $\BitB^w_1$, as discussed earlier.
    So, action $(2,1)$ over $(\BitB^w_0, \BitB^w_1)$ gives it a payoff of at least $1 + 1/2 = 3/2 > 4/3 + \eps$.
    So, $\PrimP^w$ must play $(2,1)$ with probability $p^w_0 = 1$, which contradicts $p^w_0 < 1 - \alpha$, so $p^w_0 \ge 1 - \alpha \implies \val{w} = 0$, as required.

    \item If $\DeP^u_1$ is active, then $\val{w} = 1$.
    Our construction of this gate is symmetric if we switch the roles of the $s=0$ side and $s=1$ side and the roles of the output players $\PrimP^v$ and $\PrimP^w$.
    In particular, the proof of this case---$\DeP^u_1$ is active then $\val{w} = 1$---is analogous to the case above---$\DeP^u_0$ is active then $\val{v} = 0$.

    \item If $\DeP^u_0$ is inactive, then $\val{v} = 1$. Again due to the symmetry in the construction, the proof is analogous to the case where $\DeP^u_1$ is inactive and $\val{w} = 0$.
\end{enumerate}
Given the above results, we can now check the correctness of the \PURIFY{} gate using a simple case analysis on the input value $\val{u}$.
\begin{itemize}
    \item $\val{u} = 0 \iff p^u_0 \ge 1 - \alpha = 49/50$. As $p^u_0 \ge 49/50 \ge 20/21$, using \eqref{eq:purify:1}, we know that $\DeP^u_0$ is active and $\DeP^u_1$ is inactive, which further imply that $\val{v} = 0$ and $\val{w} = 0$, respectively.
    \item $\val{u} = 1 \iff p^u_1 \ge 1 - \alpha = 49/50$. As $p^u_1 \ge 49/50 \ge 20/21$, using \eqref{eq:purify:2}, we know that $\DeP^u_0$ is inactive and $\DeP^u_1$ is active, which further imply that $\val{v} = 1$ and $\val{w} = 1$, respectively.
    \item $\val{u} \in (0, 1)$. From \eqref{eq:purify:1}--\eqref{eq:purify:2}, we know that at least $\DeP^u_0$ or $\DeP^u_1$ is either active or inactive, i.e., both detectors are never unspecified simultaneously.
    Thus, in every case at least one of the four implications above applies, and hence at least one of $\val{v}$ or $\val{w}$ lies in $\{ 0, 1 \}$.
\end{itemize}

We have now shown the correctness of all the gates. Our constructed Blotto game encodes a solution to the given \PURE{} instance in every $\eps$-WSNE of the game.

Notice that all players have at most three resources in the constructed game.
We can slightly tweak the game to ensure every player has exactly three resources:
For every player $i$ that has $B_i < 3$ resources, we add $3 - B_i$ additional battlefields to the game, denoted by the set $\cS_i$, where $|\cS_i| = 3 - B_i$, and increase the budget of player $i$ to $3$. For each battlefield $j \in \cS_i$, player $i$ has value $3/2$ and all other players have value $0$.
The proof of \cref{lm:reserve} continues to apply after adding these new battlefields: all players still have at most three resources and the reserve-battlefield structure is unchanged.
Consequently, in any $\eps$-WSNE, no opponent of player $i$ allocates resources to the new private battlefields $\cS_i$ or to player $i$'s four reserve battlefields $\cR_i$, since these battlefields have value $0$ for every other player.
We claim that every supported pure action of player $i$ allocates exactly one resource to each battlefield in $\cS_i$.
First, if such an action leaves a battlefield in $\cS_i$ unused, player $i$ can move a resource to it from an old battlefield, or from a new private battlefield that already receives at least two resources.
Since every old battlefield has individual value at most $\eta_{\BlP} = 11/10$, this deviation gains at least $3/2 - 11/10 = 2/5 > \eps$, contradicting the well-supported condition.
Conversely, suppose a supported pure action allocates at least two resources to a battlefield in $\cS_i$.
Since player $i$ has three resources and four reserve battlefields, some battlefield in $\cR_i$ receives no resource from $i$.
Moving one resource from the former battlefield to the latter leaves the private-battlefield payoff unchanged and increases the reserve-battlefield payoff by at least $\rho_{\GuP} = 10^{-4} > \eps = 10^{-5}$, again contradicting the well-supported condition.
Thus, every player $i$ allocates the remaining $3 - |\cS_i| = B_i$ resources to the old battlefields and simulates the original game.

Finally, we return to the standard uniform tie-breaking convention.
Recall that so far we assumed that, if all players allocate zero resources to a battlefield, then every player receives payoff zero from that battlefield.
Let $G$ be the game analyzed above with this modified no-resource tie-breaking convention, and let $G'$ be the same game with the standard uniform tie-breaking convention, where an empty battlefield is split uniformly among all players.
Let $n$ be the number of players, and let $C$ be an upper bound on the total value of any player across all battlefields.
For the base construction, one may take $C = 3$; after the exactly-three-resources extension above, one may take $C = 6$.
For every player $i$, every pure action $\ba_i$, and every mixed profile $\bx_{-i}$, the payoff of $\ba_i$ in $G'$ differs from its payoff in $G$ by at most $C/n$.
Indeed, the two games differ only on battlefields to which all players allocate zero resources, and such a battlefield $j$ contributes at most $v_{ij}/n$ to player $i$ in $G'$.
Therefore, every $\eps'$-WSNE of $G'$ is an $(\eps' + 2C/n)$-WSNE of $G$.
Taking $\eps' = \eps/2$ and assuming $n \ge 4C/\eps$, the above analysis applies.
We may ensure this lower bound on $n$ by padding the \PURE{} instance with disconnected satisfiable dummy components.
Thus, after normalizing payoffs to $[0,1]$, the standard uniform tie-breaking version is PPAD-hard for $(\eps/2C)$-WSNE.
\end{proof}

We can combine \cref{thm:hardness-unif-wsne} with \cref{lm:wsne2ne} to get the following hardness result for computing an approximate NE.

\begin{corollary}\label{cor:ne}
It is PPAD-hard to compute a $(\kappa/n)$-NE in Colonel Blotto games with uniform tie-breaking for a sufficiently small constant $\kappa > 0$, where $n$ is the number of players. The result holds even if the players in the Colonel Blotto games have exactly three resources.
\end{corollary}
\begin{proof}
Let $\delta > 0$ be the constant from \cref{thm:hardness-unif-wsne}, after normalizing payoffs to lie in $[0,1]$. Set $\kappa = \delta^2/8$. Suppose there is a polynomial-time algorithm that computes a $(\kappa/n)$-NE for the Blotto games constructed in \cref{thm:hardness-unif-wsne}.
In these games, each player has polynomially many pure strategies, so the output profile has polynomial support and the payoffs of supported pure strategies and the best-response values can be computed in polynomial time.
Since $\kappa/n = \delta^2/8n$, \cref{lm:wsne2ne} converts the output of this algorithm into a $\delta$-WSNE in polynomial time.
This contradicts \cref{thm:hardness-unif-wsne}. Hence, computing a $(\kappa/n)$-NE is PPAD-hard.
\end{proof}

A natural question is whether we can prove PPAD-hardness for computing an $\eps$-NE for constant $\eps$.
The main obstacle is that the proof of \cref{thm:hardness-unif-wsne} uses the well-supported condition in an essential way.
In several places, we show that if one pure action is worse than another pure action by more than the approximation parameter, then the worse action is not played at all.
This conclusion is valid in an approximate WSNE, but it is not valid in an approximate NE.
In an approximate NE, a player may still put a small probability on a significantly suboptimal action, as long as the contribution of that action to the player’s expected regret is small.
This issue already appears in \cref{lm:reserve}.
There, the WSNE condition is used to prove that no player allocates resources to battlefields for which it has zero value.
If we only had an $\eps$-NE, then each player could still put probability roughly $O(\eps)$ on such actions.
Since the Blotto game allows every player to allocate resources to every battlefield, these small probabilities may accumulate over all $n$ players and perturb the incentives in a battlefield by as much as $O(n\eps)$.
This is exactly why \cref{lm:wsne2ne} and \cref{cor:ne} only give hardness for $\eps = O(1/n)$, rather than for constant $\eps$.
Hence, our current techniques do not easily give hardness for $\eps$-NE for constant $\eps$.
However, this does not rule out a different proof of the same.

We can, however, recover PPAD-hardness for $\eps$-NE for constant $\eps$ if we restrict the action spaces so that players are allowed to allocate resources only to battlefields for which they have positive value. Under this restriction, the interaction graph of our constructed game has constant degree. Indeed, every player has positive value for only constantly many battlefields, and every battlefield has positive value for only constantly many players. Since every player has at most three resources, every player also has only constantly many pure strategies. The hardness result is formally stated in the corollary below. The proof of the corollary is provided in \cref{app:proof:ne2}, and uses an extension of \cref{lm:wsne2ne} for constant-degree interaction graphs.

\begin{corollary}\label{cor:ne2}
It is PPAD-hard to compute a $\kappa$-NE, for constant $\kappa > 0$, in Colonel Blotto games with uniform tie-breaking and
action spaces restricted so that each player $i$ may allocate resources only to battlefields $j$ with $v_{ij} > 0$.
The result holds even if the players in the Colonel Blotto games have exactly three resources.
\end{corollary}

\paragraph{Number of resources.}
The hardness results in \cref{thm:hardness-unif-wsne,cor:ne,cor:ne2} use players with at most three resources (\cref{thm:hardness-unif-wsne} also shows how to make every player have exactly three resources).
It is natural to ask whether the hardness result holds when every player has a smaller resource budget.

The case in which every player has a single resource is tractable. In this case, each pure strategy simply chooses one battlefield, and the resulting best-response structure is equivalent to that of a singleton congestion game with player-specific payoff functions.\footnote{Under the standard uniform tie-breaking convention, a player may also receive value from battlefields on which all players allocate zero resources. If the player deviates to one of these battlefields, this contribution changes. This creates a small bookkeeping issue, but it does not affect the reduction to a singleton congestion game; see \cref{app:singleton-congestion}.} Player-specific singleton congestion games admit pure Nash equilibria, and such equilibria can be computed in polynomial time~\cite{milchtaich1996congestion,ackermann2009pure}. Therefore, when $B_i=1$ for every player $i$, a pure Nash equilibrium of the corresponding Blotto game exists and can be computed in polynomial time.
However, even in this single-resource setting, if the tie-breaking is non-uniform, the problem is PPAD-hard (\cref{thm:nonunif}).

The case in which every player has at most two resources remains unresolved.
This case no longer reduces to a singleton congestion game, since a player can either concentrate both resources on one battlefield or split them across two battlefields.
At the same time, the reduction in \cref{thm:hardness-unif-wsne} uses the third resource in an essential way: by \cref{lm:primary}, each primary player is forced to keep one resource on each of its two bit battlefields, while the remaining resource encodes the value of the corresponding circuit variable through the choice between the allocations $(2,1)$ and $(1,2)$.
In these two actions, the allocation of two resources to a bit battlefield is essential to induce the detectors, and as a consequence, the other auxiliary players, to change their actions.
On the other hand, the allocation of a single resource to the other bit battlefield ensures that the primary player still gets a positive payoff from it, which is essential to ensure that the primary player chooses the suitable action among the two when influenced by the pressure and bias players on the output side.
We currently do not know how to implement this logic using only two resources for the primary player.
We leave the complexity of computing approximate equilibria in Blotto games with uniform tie-breaking, player-specific values, and at most two resources per player as an open problem.

\paragraph{Non-uniform tie-breaking.}
With non-uniform tie-breaking, the tie-breaking rule itself can create the payoff differences needed for hardness.
Indeed, computing a constant-approximate WSNE, and even a constant-approximate NE, is PPAD-hard even when every player has one resource and all values are identical.
Thus all payoff differences in this construction come from the tie-breaking rules, rather than from player-specific values or larger budgets.

We use monotone tie-breaking rules, as defined in \cref{sec:prelim}.
Monotonicity means that if an additional player joins a tie on a battlefield, then the shares of the previously tied players weakly decrease.
The proof has a similar high-level structure as the uniform tie-breaking reduction: we reduce from \PURE{} and encode the values of the circuit variables by the choices of players in the Blotto game.
The main difference is that the gate logic is implemented by the non-uniform tie-breaking rules.
This gives the following theorem.
The proof is deferred to \cref{app:proof:nonunif}.

\begin{theorem}\label{thm:nonunif}
It is PPAD-hard to compute a $\delta$-WSNE in Colonel Blotto games with non-uniform monotone tie-breaking, even when every player has one resource and all values are identical, for a sufficiently small constant $\delta > 0$.
Moreover, it is PPAD-hard to compute a $\delta'$-NE in the same class of games, for a sufficiently small constant $\delta' > 0$.
\end{theorem}

\section{Membership in PPAD}
\label{sec:membership}

We now study PPAD membership for the equilibrium computation problem in multiplayer Blotto games with uniform tie-breaking.
We discuss non-uniform tie-breaking at the end of the section.
Our main result is that computing an inverse-exponential approximate equilibrium is in PPAD.
Before stating the formal theorem, we discuss two points that are important for the formulation of the result: first, why we work with approximate equilibria rather than exact equilibria, and second, how the representation conventions from \cref{sec:prelim} are used in the membership proof.

\paragraph{Irrational equilibria.}
We first explain why the membership result is formulated for approximate equilibria.
In the standard two-player settings, rational payoffs lead to a rational polyhedral equilibrium structure.
For a bimatrix game with rational payoffs, once the supports of the two players are fixed, the equilibrium conditions are linear equalities and inequalities with rational coefficients.
Thus each support pair describes a rational polyhedron: isolated equilibria are rational, and if the equilibrium set has positive dimension, the extreme points are rational.
Similarly, in two-player constant-sum Blotto games, equilibria can be described by linear programs with rational coefficients, and hence have rational extreme points~\cite{ahmadinejad2019duels,behnezhad2023fast}.

The situation changes for three or more players.
The expected payoff of a pure action is multilinear in the mixed strategies of the other players.
Consequently, after fixing supports, the indifference conditions are polynomial rather than linear.
Even with rational payoffs, such systems may have isolated irrational solutions; this is a standard phenomenon in multiplayer finite games~\cite{etessami2010complexity}.
The same phenomenon already appears in small multiplayer Blotto games with rational values and uniform tie-breaking.

\begin{example}
\label{ex:isolated-irrational}
Consider a uniform tie-breaking Blotto game with three players and three battlefields.
Player $1$ has budget $B_1 = 3$, while players $2$ and $3$ have budgets $B_2 = B_3 = 1$.
The value vectors are
\[
    v_1 = (4, 7, 6), \qquad
    v_2 = (3, 0, 6), \qquad
    v_3 = (0, 7, 4).
\]
The values are left unnormalized for readability; scaling a player's value vector by a positive constant does not change the exact equilibria.

We now describe the isolated irrational equilibrium. Let
\[
    x^* = \frac{-73 + \sqrt{8689}}{70}, \qquad
    y^* = \frac{121 - \sqrt{8689}}{32}, \qquad
    z^* = \frac{109 - \sqrt{8689}}{84}.
\]
Consider the mixed strategy profile:
player $1$ plays $L=(0,1,2)$ with probability $x^*$ and $R=(0,2,1)$ with probability $1-x^*$;
player $2$ plays $C_2=(0,0,1)$ with probability $y^*$ and $A=(1,0,0)$ with probability $1-y^*$;
and player $3$ plays $C_3=(0,0,1)$ with probability $z^*$ and $B=(0,1,0)$ with probability $1-z^*$.
This mixed profile is an isolated irrational Nash equilibrium.
The details are given in \cref{app:proof:isolated-irrational}.
\end{example}

The example shows that irrational isolated equilibria can occur in rational multiplayer Blotto games.
This particular instance also has a rational equilibrium (\cref{app:proof:isolated-irrational}), so the example does not show that rational equilibria are absent.
We leave open whether there is a rational uniform-tie-breaking Blotto game with no rational equilibrium; such examples are known for general multiplayer normal-form games~\cite{etessami2010complexity}.
The example nevertheless shows that the rational-polyhedral structure present in two-player bimatrix games and two-player constant-sum Blotto games does not extend to the multiplayer model.
We focus on $\eps$-approximate equilibria, and allow $\eps$ to be inverse exponential in the input size, which gives arbitrarily high-precision approximations while keeping the output representable with polynomially many bits.

\paragraph{Input representation.}
We use the input representation introduced in \cref{sec:prelim}, and recall the details that are important for membership.
The input of a discrete Blotto instance consists of the $n$ budgets $(B_i)_{i \in [n]}$, the $n \times m$ values $(v_{ij})_{i \in [n], j \in [m]}$, and the approximation parameter $\eps > 0$.
The values $v_{ij}$ are rational numbers, represented in binary by their numerators and denominators.
The approximation parameter $\eps$ is also represented in binary.
The budgets are represented in unary.

The unary representation of budgets is important for the membership result.
The compact representations used below have size polynomial in $m$ and in the budgets $B_i$.
This is polynomial in the input length only when the budgets are part of the input in unary.
For example, even the marginal distribution of player $i$ on a single battlefield has one probability for each possible amount in $\{ 0, 1, \ldots, B_i \}$.
Thus, if $B_i$ were encoded in binary, this marginal vector could already have length exponential in the input size.

For this reason, throughout the membership section we use the standard pseudo-polynomial representation for discrete Blotto games: budgets are represented in unary, or equivalently the input is of length $\Omega(\sum_i B_i)$.
Under this convention, quantities whose natural size depends polynomially on the budgets are polynomial in the input length.
For example, the marginal distribution of player $i$ on a single battlefield, which may require $B_i+1$ probabilities as discussed earlier, is now polynomial in input size.
This is also the convention used in previous polynomial-time algorithms for discrete Blotto games: their running times are polynomial in the budget values $B_i$, rather than polynomial in $\log(B_i)$~\cite{ahmadinejad2019duels,behnezhad2023fast,beaglehole2023sampling,kontogiannis2025efficient}.

Throughout this section, utilities are normalized to lie in $[0,1]$. For
example, this can be enforced by scaling each player's values so that $\sum_{j\in[m]} v_{ij}\le 1$.


\paragraph{Output (flow) representation.}
As discussed in \cref{sec:prelim}, explicitly listing a mixed strategy over all pure allocations may be too large.
A mixed strategy of player $i$ is a distribution over pure allocations in $\cA_i$.
Even under unary budgets, listing such a distribution explicitly over all pure allocations is not the right representation for the membership result, since $|\cA_i|=\binom{m+B_i-1}{B_i}$ may still be exponentially large when both $m$ and $B_i$ are polynomial in input size.
Instead, we represent a mixed strategy compactly by the induced flow in the layered Blotto graph~\cite{behnezhad2023fast}.


Fix a player $i$.
Define the vertex set $V_i = \{(0,0),(m,B_i)\} \cup \{(j,b) : 1 \le j \le m-1,\ 0 \le b \le B_i\}$.
The directed acyclic graph $D_i=(V_i,E_i)$ has source $s_i=(0,0)$ and sink $t_i=(m,B_i)$.
For every battlefield $j \in [m]$,
budget level $b \in \{0\} \cup [B_i]$,
and allocation amount $r \in \{0\} \cup [B_i-b]$
such that both $(j-1,b)$ and $(j,b+r)$ belong to $V_i$,
we add the edge $((j-1,b),(j,b+r))$ to $E_i$.
This edge represents allocating $r$ resources to battlefield $j$,
after allocating $b$ resources to the first $j-1$ battlefields.
By construction, every vertex and edge of $D_i$ lies on an $s_i$--$t_i$ path.
Moreover, every $s_i$--$t_i$ path corresponds to a pure allocation of player $i$, and every pure allocation corresponds to such a path.

Let $\cF_i$ be the polytope of unit $s_i$--$t_i$ flows in $D_i$.
The graph $D_i$ has $O(mB_i)$ vertices and $O(mB_i^2)$ edges, and therefore $\cF_i$ has a polynomial-size linear description.
Let
\[
    B = \max_{i \in [n]} B_i,
    \qquad
    d_i = |E_i|,
    \qquad
    d = \sum_{i \in [n]} d_i.
\]
Thus, $d = O(m n B^2)$, which is polynomial in the input length under the unary-budget representation.
For notational convenience, we extend every flow $f_i \in \cF_i$ by setting $f_i(e)=0$ for each edge $e \notin E_i$.
For a flow $f_i \in \cF_i$, define
\[
    p_{ijr}(f_i)
    =
    \sum_{b=0}^{B_i-r}
    f_i\bigl((j-1,b),(j,b+r)\bigr).
\]
This is the marginal probability that player $i$ allocates exactly $r$ resources to battlefield $j$.

The flow representation is equivalent to a mixed Blotto strategy.
Given any mixed strategy over pure allocations, send the probability of each pure allocation along its corresponding $s_i$--$t_i$ path in $D_i$.
This gives a unit flow $f_i$, and the quantities $p_{ijr}(f_i)$ are exactly the battlefield marginals of the mixed strategy.
Conversely, since $D_i$ is acyclic, every unit flow $f_i\in\cF_i$ can be decomposed in polynomial time into a convex combination of $s_i$--$t_i$ paths~\cite{ahuja1993network}.
This decomposition gives an explicit mixed strategy over pure allocations that induces the same flow.

We will focus on the compact flow representation rather than the path decomposition.
This loses no payoff-relevant information.
In a Blotto game, payoffs are additive over battlefields, and the payoff contribution from battlefield $j$ depends only on how many resources each player allocates to battlefield $j$.
Therefore expected payoffs are determined by the battlefield marginals $p_{ijr}$.
Thus a flow profile $f=(f_i)_{i\in[n]}$ is a polynomial-size representation of the mixed-strategy profile for the purposes of computing utilities and best responses.

\paragraph{Payoff computation.}
We next show that expected utilities can be computed efficiently from the flow representation under uniform tie-breaking.
Fix player $i$, battlefield $j$, and an amount $r \in \{0\} \cup [B_i]$.
We interpret $p_{kjr}(f_k) = 0$ whenever $r > B_k$.
For an opponent $k \neq i$, define
\[
    L_{kjr}(f_k) = \sum_{\ell < r} p_{kj\ell}(f_k).
\]
Thus $L_{kjr}$ is the probability that opponent $k$ allocates less than $r$ to battlefield $j$.

If player $i$ allocates $r$ resources to battlefield $j$, then she receives a positive share of battlefield $j$ exactly when no opponent allocates more than $r$ resources to $j$.
Conditional on this event, if exactly $t$ opponents allocate exactly $r$ resources, then uniform tie-breaking gives player $i$ a share $1 / (t + 1)$ of the battlefield.

Order the opponents of player $i$ as $k_1, \ldots, k_{n - 1}$.
For $s = 0, \ldots, n - 1$ and $t = 0, \ldots, s$,
let $c_t^{(s)}$ be the probability that,
among the first $s$ opponents,
exactly $t$ allocate exactly $r$ resources to battlefield $j$,
and all remaining opponents among these allocate less than $r$.
Initialize
\[
    c_0^{(0)} = 1,
    \qquad
    c_t^{(0)} = 0 \quad \text{for } t > 0.
\]
For $s = 1, \ldots, n - 1$, set
\[
    c_t^{(s)}
    = c_t^{(s - 1)} L_{k_sjr}(f_{k_s})
    + c_{t - 1}^{(s - 1)} p_{k_sjr}(f_{k_s}),
\]
where out-of-range terms are interpreted as zero.
The two terms of the above formula capture the two cases where player $k_s$ puts less than $r$ and equal to $r$ resources, respectively.
Then the expected fraction of battlefield $j$ won by player $i$, conditional on allocating $r$ to $j$, is
\[
    h_{ijr}(f_{-i}) = \sum_{t = 0}^{n - 1} \frac{c_t^{(n - 1)}}{t + 1}.
\]
For $r = 0$, this formula gives $h_{ij0} = \prod_{k \neq i} p_{kj0}/n$, so it also accounts for the uniform tie when every player allocates zero resources.
For an edge $e = ((j - 1, b), (j, b + r))$ of $D_i$, define
\[
    C_{i, e}(f_{-i}) = v_{ij} h_{ijr}(f_{-i}).
\]
Then player $i$'s expected utility in the flow profile $f = (f_i)_{i \in [n]}$ can be written as
\[
    U_i(f) = \sum_{e \in E(D_i)} f_i(e) C_{i, e}(f_{-i}),
\]
where $E(D_i)$ denotes the set of edges of $D_i$.
In particular, for fixed $f_{-i}$, the function $U_i(\cdot, f_{-i})$ is linear in player $i$'s own flow.

\begin{lemma}\label{lm:compact-payoff}
For every player $i$, arithmetic circuits for $U_i$ and $\nabla U_i$ can be constructed in polynomial time and have size $O(m n B^2 + m n^2 B)$.
Their values at rational inputs can be computed exactly in time polynomial in the game input length and the encoding length of the query.
\end{lemma}

\begin{proof}
Fix player $i$.
Each opponent's marginal $p_{kjr}$ is a sum of at most $B + 1$ edge flows.
There are $O(m n B)$ such marginals, so computing all of them uses $O(m n B^2)$ arithmetic gates.
For each opponent $k$ and battlefield $j$, computing the prefix sums $L_{kjr}$ successively as $r$ increases uses $O(B)$ further gates, giving $O(m n B)$ additional gates in total.
Thus all the marginals and prefix sums together use $O(m n B^2)$ gates.

For a fixed pair $(j, r)$, the dynamic program defining $h_{ijr}$ has $\sum_{s = 0}^{n - 1} (s + 1) = n(n + 1)/2$ states.
Each state requires only a constant number of additions and multiplications, and the final sum defining $h_{ijr}$ uses $O(n)$ more gates.
Thus each conditional share requires $O(n^2)$ gates.
There are at most $m (B + 1) = O(m B)$ pairs $(j, r)$, so computing all these shares uses $O(m n^2 B)$ gates.
Once the shares are available, forming the coefficients $C_{i, e}$ and taking their weighted sum with the edge flows uses $O(d_i)$ gates.
The total circuit size for $U_i$ is therefore
\[
    O(m n B^2 + m n^2 B + d_i) = O(m n B^2 + m n^2 B),
\]
where we use $d_i = O(m B^2)$.
The circuit can be constructed in polynomial time, and all its rational constants have polynomial bit length.
The Baur--Strassen theorem~\cite[Theorem~2]{baur1983complexity} gives a circuit for $\nabla U_i$ with only a constant-factor increase in size.

At a rational query, each opponent's marginals and prefix sums have polynomial bit length and can be written over a common denominator, obtained by multiplying the denominators of her flow coordinates.
For each fixed pair $(j, r)$, the recurrence for $c_t^{(s)}$ has $n - 1$ stages.
At each stage, the common denominator is multiplied by that of the new opponent, and each new numerator is a sum of two products of previous numerators with that opponent's numerators.
Numerator and denominator lengths therefore increase by at most a polynomial amount per stage, so every state in each of these recurrences has polynomial bit length.
The sums defining $h_{ijr}$ and $U_i$ combine polynomially many such values with rational coefficients of polynomial bit length, so the bound extends to the entire utility circuit.
The reverse differentiation construction uses only additions and multiplications by already computed values of the utility circuit.
The same bit-length argument therefore applies to the derivative circuit, since there are only polynomially many such operations.
Thus both utilities and their derivatives can be evaluated exactly in polynomial time.
\end{proof}

\paragraph{Membership for approximate NE.}

Each player's utility is linear in her own flow, so the game is concave.
We apply the following result of Papadimitriou, Vlatakis-Gkaragkounis, and Zampetakis~\cite[Theorem~4.9 and Lemma~E.3]{papadimitriou2023computational}.
We state it for a product of strategy sets, with utilities and gradients represented by circuits.\footnote{The original search problem also allows certificates that some of the required conditions fail, instead of an approximate equilibrium.
We verify all the required conditions for our game below, so these alternative outputs cannot occur.}

\begin{theorem}
\label{thm:concave-games-so}
Consider a game with feasible region $\mathcal Z = \prod_i \mathcal Z_i \subseteq [-1, 1]^q$, where $q$ is polynomial in the input length.
Suppose $\mathcal Z$ is compact and convex, contains the origin, and is given by a polynomial-time strong separation oracle.
Suppose $\mathcal Z$ contains a Euclidean ball of radius $r > 0$ and lies in the ball of radius $R \le 1$ centered at the origin.
The radii and the center coordinates of the inner ball are given as rationals of polynomial bit length.
On $[-1, 1]^q$, each utility is continuously differentiable, takes values in $[0, 1]$, and is concave in the corresponding player's own variables.
The utilities and their gradients are given by polynomial-size circuits that can be evaluated in polynomial time at rational queries.
A common Lipschitz bound for the utilities on this box, with respect to the Euclidean norm, is given as a rational of polynomial bit length.
Then computing a rational $\eps$-NE in $\mathcal Z$ is in PPAD for every positive rational $\eps$ given in binary.
\end{theorem}

With the strong separation oracle given by the linear flow constraints, the theorem returns an exactly feasible profile, as required for valid Blotto strategies.
We now show that the flow game satisfies the assumptions of \cref{thm:concave-games-so}.

\smallskip
\noindent\emph{Reduced coordinates and separation.}
The inner-ball assumption does not hold in the edge-flow coordinates: flow conservation confines the feasible region to a lower-dimensional affine subspace.
We instead keep only the independent coordinates, using a spanning tree to recover the other edge flows from flow conservation.
This gives another description of the same feasible strategies.
Scaling these coordinates then puts the feasible region inside the unit ball.
The following lemma makes this construction precise.

\begin{lemma}
\label{lm:membership:coordinates}
For each player $i$, let $q_i = d_i - |V_i| + 1$, and set $q = \sum_i q_i \le d$.
In polynomial time, we can construct polytopes $\mathcal Z_i \subseteq \mathbb R^{q_i}$ and affine bijections from $\mathcal Z_i$ to $\cF_i$ of the form
\[
    z_i \longmapsto f_i = d T_i z_i + f_i^0,
\]
for some $d_i \times q_i$ matrices $T_i$ with entries in $\{-1, 0, 1\}$ and unit path flows $f_i^0 \in \cF_i$.
The product $\mathcal Z = \prod_i \mathcal Z_i$ is compact and convex, contains the origin, lies in the unit ball, and admits a polynomial-time strong separation oracle.
If $q > 0$, it also contains a Euclidean ball of radius $1/(d^2 (B + 1)^m)$
whose rational center can be computed in polynomial time.
The center coordinates and radius have polynomial bit length.
\end{lemma}

\begin{proof}
For each player $i$, choose a directed $s_i$--$t_i$ path in $D_i$ and let $f_i^0$ denote the $s_i$--$t_i$ flow where the entire one unit of flow goes through this path.
Extend the edges of this path to a spanning tree of the underlying undirected graph of $D_i$.
We use the flows on the $q_i$ edges of $D_i$ outside this tree as the independent coordinates.
Let $z_i$ consist of the flows on these non-tree edges divided by $d$.

Given any vector $z_i \in \mathbb R^{q_i}$, assign the values $d z_i$ to the non-tree edges.
To determine the value on an edge $e$ of the tree, delete $e$ from the tree and sum the required flow conservation equations over either resulting component.
Only $e$ and non-tree edges cross this cut.
If $e$ lies on the chosen $s_i$--$t_i$ path, the cut separates $s_i$ from $t_i$ and must carry one unit of net flow in the direction of $e$; otherwise, the net flow is zero.
Set $f_i(e)$ to the value required by this equation: $f_i^0(e)$ plus a signed sum of the assigned values on the non-tree edges crossing the cut.
Each assignment satisfies its cut equation independently of the values assigned to the other tree edges.
Together these assignments give $f_i = d T_i z_i + f_i^0$, with every entry of $T_i$ in $\{-1, 0, 1\}$.
Both $T_i$ and $f_i^0$ can be constructed in polynomial time.

The reconstructed vector $f_i$ also satisfies the flow conservation equations at every vertex.
To verify the equation at a vertex $v$, delete $v$ from the tree.
Each remaining component was joined to $v$ by a single tree edge, so the reconstructed flows satisfy the cut equation for that component.
When we sum these equations, the terms for edges between components cancel, leaving only the terms for edges entering or leaving $v$.
Since total supply equals total demand, this sum gives the conservation equation at $v$.

Thus, for every $z_i \in \mathbb R^{q_i}$, the reconstructed vector $f_i$ satisfies the flow conservation equations for a unit $s_i$--$t_i$ flow.
To require nonnegative flow on every edge, we define
\[
    \mathcal Z_i = \{z_i \in \mathbb R^{q_i} : d T_i z_i + f_i^0 \ge 0\}.
\]
For every $z_i \in \mathcal Z_i$, the defining inequalities give $f_i(e) \ge 0$ for every edge $e \in E_i$.
Together with the flow conservation equations established above, this shows that the reconstructed vector $f_i$ belongs to $\cF_i$.
Conversely, every $f_i \in \cF_i$ satisfies the cut equations, so the reconstruction recovers $f_i$ when $z_i$ consists of the flows of $f_i$ on the non-tree edges divided by $d$.
The reconstruction is therefore a bijection from $\mathcal Z_i$ to $\cF_i$.

At $z = 0$, the reconstruction gives $f_i = f_i^0 \in \cF_i$ for every player $i$, so $0 \in \mathcal Z$.
Each $D_i$ is acyclic, so every feasible unit flow assigns at most $1$ to each edge.
Hence
\[
    \mathcal Z \subseteq [0, 1/d]^q,
    \qquad
    \|z\|_2 \le \sqrt q / d \le 1
    \quad\text{for all $z \in \mathcal Z$.}
\]
Thus $\mathcal Z$ is compact and convex and lies in the unit ball.

The inequalities defining the polytopes $\mathcal Z_i$ also give a strong separation oracle for $\mathcal Z$.
Given a query $z$, we reconstruct each player's flow and check whether all edge flows are nonnegative.
If they are, then $z \in \mathcal Z$.
Otherwise, choose a player $i$ and an edge $e \in E_i$ with negative reconstructed flow.
We form a vector in $\mathbb R^q$ by placing the negated row of $T_i$ for $e$ in player $i$'s coordinates and zeros in all other coordinates.
This vector separates $z$ from $\mathcal Z$.
The row for $e$ is nonzero, since a zero row would give the nonnegative flow $f_i^0(e)$ at every query.
Its entries are in $\{-1, 0, 1\}$, so the separating vector has infinity norm one.
The oracle uses $O(d^2) = O(m^2 n^2 B^4)$ arithmetic operations and comparisons and runs in polynomial time at rational queries.

To obtain an inner ball when $q > 0$, we construct for each player a feasible flow that is positive on every edge.
We will use their reduced coordinates as the center and bound how far these coordinates can change while keeping every edge flow nonnegative.
For each player $i$, form a unit flow $\bar f_i$ by starting with one unit at $s_i$ and splitting the flow reaching each nonterminal vertex equally among its outgoing edges.
Every vertex has at most $B + 1$ outgoing edges, so each split sends at least a $1 / (B + 1)$ fraction of the flow along each outgoing edge.
Every edge lies on an $s_i$--$t_i$ path of $m$ edges, so following this path gives the lower bound
\[
    \bar f_i(e) \ge \delta := (B + 1)^{-m}
    \qquad\text{for every $e \in E_i$.}
\]

Let $\bar z = (\bar z_i)_i$, where $\bar z_i$ consists of the flows $\bar f_i(e)$ on the non-tree edges divided by $d$.
Since the entries of $T_i$ have absolute value at most $1$, the flow reconstructed from any point $z$ satisfies
\[
    |f_i(e) - \bar f_i(e)| \le d \sqrt q \, \|z - \bar z\|_2 \le d^2 \|z - \bar z\|_2
    \qquad\text{for every $e \in E_i$.}
\]
If $\|z - \bar z\|_2 \le \delta / d^2$, each edge flow changes by at most $\delta$ and therefore remains nonnegative.
The ball with center $\bar z$ and radius $\delta / d^2$ is thus contained in $\mathcal Z$.

We compute the flows $\bar f_i$ layer by layer using polynomially many additions and divisions.
Along any path, the splitting procedure makes at most $m$ divisions by integers at most $B + 1$.
Each of these integers divides $(B + 1)!$, so every intermediate flow value can be written over the denominator $((B + 1)!)^m$.
These values lie in $[0, 1]$, so their numerators and denominators have polynomial bit length under the unary encoding of the budgets.
Thus the flows $\bar f_i$ and their reduced coordinates $\bar z_i$ can be computed in polynomial time.
The center coordinates can be written over the common denominator $d ((B + 1)!)^m$, and the radius is $1 / (d^2 (B + 1)^m)$, so both have polynomial bit length.
\end{proof}

Under this change of coordinates, we regard $U_i$ as a function of the reduced profile $z$, obtained by substituting $f_k = d T_k z_k + f_k^0$ for every player $k$.
Computing all reconstructed edge flows uses $O(d^2)$ gates.
By the chain rule, each derivative with respect to a reduced coordinate is a linear combination of the derivatives with respect to the edge flows; computing all these combinations for one utility uses $O(d^2)$ further gates.
Together with \cref{lm:compact-payoff}, this gives circuits for each player's utility and gradient of size
\[
    O(m n B^2 + m n^2 B + d^2) = O(m^2 n^2 B^4),
\]
where we use $d = O(m n B^2)$.
The reconstruction and the linear combinations use rational coefficients of polynomial bit length, so these circuits can still be evaluated exactly in polynomial time at rational queries.
The utility $U_i$ is a polynomial, hence continuously differentiable, and is affine in $z_i$ for every fixed $z_{-i}$.
In particular, the required concavity holds throughout $[-1, 1]^q$.

\smallskip
\noindent\emph{Bounds on utility values and Lipschitz constants.}
The normalization of Blotto payoffs only bounds the utilities on feasible profiles, whereas \cref{thm:concave-games-so} requires bounds on utility values and Lipschitz constants on the whole box $[-1, 1]^q$.
Outside $\mathcal Z$, the reconstructed edge flows may be negative, and the utility polynomials need not take values in $[0, 1]$.
We therefore bound these polynomials and their derivatives on the whole box, then rescale the utilities.

Fix an arbitrary point $z \in [-1, 1]^q$, which need not belong to $\mathcal Z$.
In the reconstruction $f_k = d T_k z_k + f_k^0$, each edge flow is a constant in $\{0, 1\}$ plus at most $q \le d$ terms of absolute value at most $d$, since every coordinate of $z$ has absolute value at most $1$.
Thus every reconstructed edge flow has absolute value at most $d^2 + 1$, and each of its partial derivatives with respect to a reduced coordinate has absolute value at most $d$.
Each marginal $p_{kjr}$ or prefix sum $L_{kjr}$ sums at most $(B + 1)^2$ edge flows.
Thus the marginals, prefix sums, and their first partial derivatives all have absolute value at most
\[
    K = (B + 1)^2 (d^2 + 1).
\]
Fix a player $i$, a battlefield $j$, and an allocation amount $r$.
Applying the recurrence for $c_t^{(s)}$ to the reconstructed flows, induction on the number $s$ of opponents gives, for every reduced coordinate $z_a$,
\[
    |c_t^{(s)}| \le (2 K)^s,
    \qquad
    \left|\frac{\partial c_t^{(s)}}{\partial z_a}\right| \le s (2 K)^s.
\]
For the derivative bound, differentiating each of the two products in the recurrence gives one term of absolute value at most $(s - 1) K (2 K)^{s - 1}$ and one at most $K (2 K)^{s - 1}$.
Summing the final states with weights $1 / (t + 1)$ therefore gives
\[
    |h_{ijr}| \le n (2 K)^n,
    \qquad
    \left|\frac{\partial h_{ijr}}{\partial z_a}\right| \le n^2 (2 K)^n.
\]

To bound $U_i$, we sum at most $d$ terms of the form $v_{ij} f_i(e) h_{ijr}$, where $e$ represents allocating $r$ resources to battlefield $j$.
Set
\[
    M = d (d^2 + 1) n (2 K)^n \max\{1, \max_{i, j} v_{ij}\}.
\]
The bounds on the edge flows and shares, together with the product rule, give, for every $z \in [-1, 1]^q$,
\[
    |U_i| \le M,
    \qquad
    \left|\frac{\partial U_i}{\partial z_a}\right| \le 2 M n,
    \qquad
    \|\nabla_z U_i\|_2 \le 2 M n \sqrt q \le 2 M d n.
\]
The number $M$ has polynomial bit length: the power $(2 K)^n$ contributes $O(n \log(K))$ bits, and the other factors have polynomial bit length.

Define the rescaled utilities and accuracy by
\[
    \widehat U_i = \frac{U_i + M}{2 M},
    \qquad
    \widehat\eps = \frac{\eps}{2 M}.
\]
These utilities take values in $[0, 1]$ on the whole box and are $d n$-Lipschitz with respect to the Euclidean norm.
The transformation preserves concavity and the polynomial bounds on circuit size and exact evaluation time.
The bit length of $\widehat\eps$ is polynomial in the game input length and the encoding length of $\eps$, so inverse-exponential accuracy is allowed.
Since rescaling divides each payoff difference by $2 M$, every $\widehat\eps$-NE of the rescaled game is an $\eps$-NE under the original utilities.

\begin{theorem}\label{thm:membership-ne}
Consider a Colonel Blotto game with uniform tie-breaking, unary budgets, and rational battlefield values given in binary, with utilities normalized to $[0, 1]$.
For every positive rational $\eps$ given in binary, computing an $\eps$-NE is in PPAD.
The output can be represented by rational unit flows or by a distribution over at most $d_i = O(m B_i^2)$ pure allocations for each player $i$.
\end{theorem}

\begin{proof}
If $m = 1$, every player must allocate her entire budget to the only battlefield, so the unique strategy profile is an exact equilibrium.
We therefore assume $m \ge 2$.
Since every budget $B_i$ is positive, each graph $D_i$ then contains at least two distinct $s_i$--$t_i$ paths, so $q_i > 0$ for every player $i$.
By \cref{lm:membership:coordinates}, we can take $R = 1$ and $r = 1/(2 d^2 (B + 1)^m)$ in \cref{thm:concave-games-so}.

Apply \cref{thm:concave-games-so} to the rescaled game constructed above with accuracy $\widehat\eps = \eps / (2 M)$.
It returns a rational feasible profile $z \in \mathcal Z$ of polynomial bit length whose regret under the original utilities is at most $\eps$.
By the affine bijection in \cref{lm:membership:coordinates}, every feasible flow deviation corresponds to a deviation in reduced coordinates with the same payoff.
The reconstructed flow profile is therefore an $\eps$-NE of the flow game.

Decompose each player's flow into paths by repeatedly choosing an $s_i$--$t_i$ path whose edges all carry positive flow and subtracting the minimum flow along that path from each of its edges.
Assign the subtracted amount as the probability of the corresponding pure allocation.
Each step eliminates a positive edge, so the decomposition uses at most $d_i$ paths and takes polynomial time.
The path weights are rational with polynomial bit length.
The resulting distributions over pure Blotto allocations preserve the battlefield marginals and hence the payoffs.
Every mixed Blotto deviation also induces a feasible flow, so these distributions form an $\eps$-NE of the Blotto game.
\end{proof}

\paragraph{Membership for approximate WSNE.}

We obtain PPAD membership for $\eps$-WSNE by pruning a sufficiently accurate approximate NE using \cref{lm:wsne2ne}.
The output in this case is an explicit polynomial-support distribution over pure Blotto allocations.

\begin{corollary}\label{cor:membership-wsne}
Consider a Colonel Blotto game with uniform tie-breaking, unary budgets, and rational battlefield values given in binary, with utilities normalized to $[0, 1]$.
For every positive rational $\eps$ given in binary, computing an $\eps$-WSNE is in PPAD.
\end{corollary}

\begin{proof}[Proof sketch.]
It suffices to consider $0 < \eps < 1$, since every profile is an $\eps$-WSNE when $\eps \ge 1$.
Set $\eta = \eps^2 / (8 n)$, which has polynomial bit length in the input.
By \cref{thm:membership-ne}, computing an $\eta$-NE with polynomial support is in PPAD.

To implement the pruning procedure, we need the payoffs of supported allocations and the best-response values against the original $\eta$-NE.
Fix player $i$, and let $f_{-i}$ be the flows induced by the opponents' strategies in this profile.
Writing the utility in flow coordinates, her payoff from any $g_i \in \cF_i$ is
\[
    U_i(g_i, f_{-i}) = \sum_{e \in E(D_i)} g_i(e) C_{i, e}(f_{-i}).
\]
The payoff computation in \cref{lm:compact-payoff} gives the edge weights $C_{i, e}(f_{-i})$ exactly in polynomial time.
A pure allocation's payoff is the weight of its corresponding path, and the best-response value is the maximum weight of an $s_i$--$t_i$ path in $D_i$.
Since $D_i$ is acyclic, the latter can be computed in polynomial time by the longest-path dynamic program.
We can therefore delete, for each player, the supported allocations whose payoff against the original opponents' profile is more than $\eps / 2$ below her best-response value and renormalize in polynomial time.
By \cref{lm:wsne2ne} and the choice of $\eta$, the resulting profile is an $\eps$-WSNE.
The full proof is provided in \cref{app:proof:membership-wsne}.
\end{proof}

\paragraph{Remark on non-uniform tie-breaking.}

For non-uniform tie-breaking, the main obstacle to the membership proof is computing expected payoffs efficiently.
Under uniform tie-breaking, the dynamic program in \cref{lm:compact-payoff} only needs to keep track of the number of tied opponents.
For a general rule $\tau_j$, the shares may also depend on their identities.
When player $i$ allocates $r$ resources to battlefield $j$, her expected share is
\[
    h_{ijr}
    = \sum_{S \subseteq [n] \setminus \{i\}}
        \tau_j(S \cup \{i\}, i)
        \prod_{k \in S} p_{kjr}
        \prod_{k \notin S \cup \{i\}} L_{kjr}.
\]
Here $S$ is the set of opponents who also allocate $r$ resources, while every other opponent allocates less.
Using the earlier bounds on $p_{kjr}$, $L_{kjr}$, and their derivatives, the displayed formula gives the same bounds on utility values and derivatives on the whole box $[-1, 1]^q$.

The flow polytopes and separation oracles also remain unchanged, and each utility remains affine in the player's own coordinates.
To extend the membership argument, it therefore suffices to construct polynomial-size circuits for the expected-share functions $h_{ijr}$ and their first derivatives in polynomial time, with exact polynomial-time evaluation at rational queries.
Under this condition, both the NE argument and the WSNE pruning procedure apply.
The tie-breaking rules used in the proof of \cref{thm:nonunif} satisfy this condition.
The displayed sum has $2^{n - 1}$ terms and does not by itself supply such circuits.

Computing a player's expected share against mixed strategies requires averaging the shares awarded by the rule over the possible tied sets, weighted by their probabilities.
An efficient procedure for evaluating the rule on one specified tied set does not by itself give an efficient way to compute this average.
In \cref{app:nonuniform-tiebreaking-hardness}, we construct a monotone rule that can be evaluated in polynomial time on any given tied set, but for which computing expected shares exactly allows us to count the independent sets in a graph.

\section{Conclusion and Open Problems}

We studied equilibrium computation in discrete multiplayer Blotto games with player-specific battlefield values.
For uniform tie-breaking, we showed PPAD-completeness of approximate equilibrium computation, with hardness already for three resources per player.
We also showed that the one-resource case is polynomial-time solvable under uniform tie-breaking, but becomes hard under non-uniform monotone tie-breaking.
These results show that the tractability of two-player constant-sum Blotto games does not extend to the multiplayer model with player-specific values.

Our hardness result for uniform tie-breaking uses both multiple players and player-specific battlefield values.
It remains open whether the problem remains hard with only one of these generalizations, either with multiple players or with player-specific values.
The case with player-specific values but a constant number of players also remains open.
Finally, our hardness result uses three resources per player, while the one-resource case is polynomial-time solvable. The two-resource case remains open.

\section*{Acknowledgements}
This project has received funding from the European Research Council (ERC) under the European Union’s Horizon Europe research and innovation programme (grant agreement No.\ 101198689).
ChatGPT 5.5 was used after the initial draft to identify possible errors and typos and to suggest improvements to language, notation, and the presentation of technical results.
Some of these suggestions were incorporated after verification by the authors.

\clearpage
\appendix

\section{Omitted Proofs}
\subsection{Proof of \cref{lm:wsne2ne}}
\label[appendix]{app:proof:wsne2ne}
\begin{proof}
The proof is direct by setting the trivial bound of $d = n$ in \cref{lm:wsne2ne2}.
\end{proof}

\subsection{Proof of \cref{cor:ne2}}
\label[appendix]{app:proof:ne2}
\begin{proof}
We first state the bounded-interaction version of \cref{lm:wsne2ne}.

\begin{lemma}\label{lm:wsne2ne2}
Let $G$ be a finite game with utilities normalized to lie in $[0,1]$.
Suppose that $G$ has dependency degree at most $d$, meaning that for every player $i$, there is a set $N_i \subseteq [n] \setminus \{i\}$ of at most $d$ other players such that player $i$'s payoff depends only on its own action and the actions of the players in $N_i$.
Then, for every $\delta \in (0,1]$, given a $(\delta^2/8d)$-NE of $G$, removing from each player's support every pure strategy whose payoff is more than $\delta/2$ below the best-response payoff, and then renormalizing, gives a $\delta$-WSNE of $G$.
\end{lemma}

\begin{proof}
Let $\bx$ be an $\eps$-NE of the game, where
\[
    \eps \le \frac{\delta^2}{8d}.
\]
For every player $i \in [n]$, let $N_i \subseteq [n]\setminus\{i\}$ be a set of
players, with $|N_i|\le d$, such that the utility of player $i$ depends only on
her own action and the actions of players in $N_i$.

For every player $i \in [n]$, let
\[
    M_i = \max_{\ba_i' \in \cA_i} u_i(\ba_i', \bx_{-i})
\]
be the payoff of a best response to $\bx_{-i}$. We call a pure strategy
$\ba_i \in \cA_i$ bad for player $i$ if
\[
    u_i(\ba_i, \bx_{-i}) < M_i - \frac{\delta}{2}.
\]
Let $\cA_i^{\textrm{bad}} \subseteq \cA_i$ be the set of bad strategies of
player $i$, and let
\[
    \gamma_i = \sum_{\ba_i \in \cA_i^{\textrm{bad}}} \bx_i(\ba_i)
\]
be the probability that $\bx_i$ assigns to bad strategies.

Since $\bx$ is an $\eps$-NE, the expected regret of player $i$ is at most
$\eps$. Therefore,
\[
\eps \ge M_i - u_i(\bx)
    = \sum_{\ba_i \in \cA_i} \bx_i(\ba_i)
        \bigl(M_i - u_i(\ba_i, \bx_{-i})\bigr)
    \ge \frac{\delta}{2}\gamma_i \implies \gamma_i \le \frac{2\eps}{\delta}.
\]

We now delete all bad strategies from the support of each player and renormalize.
This gives a new mixed strategy profile $\by = (\by_i)_{i \in [n]}$, where
\[
    \by_i(\ba_i) =
    \begin{cases}
        \dfrac{\bx_i(\ba_i)}{1-\gamma_i}, & \text{if } \ba_i \notin \cA_i^{\textrm{bad}}, \\[1.2ex]
        0, & \text{if } \ba_i \in \cA_i^{\textrm{bad}}.
    \end{cases}
\]
This is well-defined because
\[
    \gamma_i \le \frac{2\eps}{\delta}
    \le \frac{\delta}{4d}
    \le \frac{1}{4},
\]
where we use $d \ge 1$ and $\delta \le 1$.

For each player $k$, the $\ell_1$ distance between $\bx_k$ and $\by_k$ is at
most $2\gamma_k$, and hence their total variation distance is at most
$\gamma_k$. Since player $i$'s utility depends only on the actions of players in
$N_i$, changing the mixed strategies from $\bx$ to $\by$ changes the payoff of
any fixed pure strategy $\ba_i \in \cA_i$ by at most
\[
    \left|u_i(\ba_i, \by_{-i}) - u_i(\ba_i, \bx_{-i})\right|
    \le \sum_{k \in N_i} \gamma_k
    \le d \cdot \frac{2\eps}{\delta}
    =
    \frac{2d\eps}{\delta}.
\]

Now fix a player $i \in [n]$, a pure strategy $\ba_i \in \cA_i$ such that
$\by_i(\ba_i)>0$, and an arbitrary deviation $\ba_i'\in \cA_i$. Since
$\by_i(\ba_i)>0$, the strategy $\ba_i$ is not bad, and therefore
\[
    u_i(\ba_i, \bx_{-i}) \ge M_i - \frac{\delta}{2}.
\]
Also, by the definition of $M_i$,
\[
    u_i(\ba_i', \bx_{-i}) \le M_i.
\]
Using the payoff perturbation bound above, we get
\[
    u_i(\ba_i', \by_{-i}) - u_i(\ba_i, \by_{-i})
    \le
    \left(M_i + \frac{2d\eps}{\delta}\right)
    -
    \left(M_i - \frac{\delta}{2} - \frac{2d\eps}{\delta}\right)
    =
    \frac{\delta}{2} + \frac{4d\eps}{\delta}
    \le \delta,
\]
where the last inequality follows from $\eps \le \delta^2/(8d)$.
Thus every pure strategy in the support of $\by_i$ is a $\delta$-best response to $\by_{-i}$.
Since this holds for every player $i \in [n]$, the mixed strategy profile $\by$ is a $\delta$-WSNE.
\end{proof}

In the restricted version of the Blotto game we consider, the pure strategy set of player $i$ consists only of allocations $\ba_i$ satisfying the usual budget constraint, together with the additional restriction that $a_{ij} = 0$ whenever $v_{ij} = 0$ for any battlefield $j$. In other words, players may allocate resources only to battlefields for which they have positive value.

The proof of \cref{thm:hardness-unif-wsne} continues to apply. In fact, the role of \cref{lm:reserve} was to exactly show that, in any approximate WSNE, players do not allocate resources to battlefields for which they have zero value. Under the restricted action spaces, such allocations are not allowed by definition. The rest of the proof is unchanged.

Moreover, in the constructed games, every player has positive value for only constantly many battlefields, and every battlefield has positive value for only constantly many players.
Hence, under the restricted action spaces, the dependency degree of the game is bounded by an absolute constant $d$.
Further, since every player has only polynomially many pure strategies in these games, and the relevant payoffs and best-response values can be computed in polynomial time, the pruning construction in \cref{lm:wsne2ne2} can be implemented in polynomial time.

Let $\delta > 0$ be the constant from \cref{thm:hardness-unif-wsne}, after normalizing payoffs to lie in $[0,1]$. Set $\kappa = \delta^2/8d$.
Since $d$ is an absolute constant, $\kappa$ is also an absolute constant. Suppose there is a polynomial-time algorithm that computes a $\kappa$-NE for the restricted-action Blotto games considered above. By \cref{lm:wsne2ne2}, the output of this algorithm can be converted in polynomial time into a $\delta$-WSNE. This contradicts \cref{thm:hardness-unif-wsne}. Hence computing a $\kappa$-NE is PPAD-hard for constant $\kappa$ in the Blotto games with restricted actions.
\end{proof}

\subsection{Proof of \cref{thm:nonunif}}
\label[appendix]{app:proof:nonunif}

\begin{proof}
As in the proof of \cref{thm:hardness-unif-wsne}, we reduce from the \PURE{} problem (\cref{def:pure}).
We are given an instance of \PURE{}.
Let $V$ be the set of variables/nodes.
Every variable is an input to and an output of exactly one gate.
Throughout the proof, we will use several constants. These are listed in \cref{tab:reduc-params-nonunif}. Their use will be explained when they are introduced in the proof.
\begin{table}[htbp]
    \centering
    \begin{tabular}{c|c|c}
        parameter & value & description \\
        \hline
        $\eps$  & $1/100$ & unnormalized approximation parameter for $\eps$-WSNE \\
        $\beta_I$     & $1/8$   & payoff of auxiliary player on input side when primary player is absent \\
        $\beta_O$     & $1/20$  & payoff of auxiliary player on output side \\
        $\rho$     & $11/50$ & default payoff of primary player \\
        $\lambda_{\NOT}$    & $1/25$ & payoff reduction for primary player used in \NOT{} gate \\
        $\lambda_{\AND,0}$  & $1/25$ & payoff reduction for primary player used in 0-side of \AND{} gate \\
        $\lambda_{\AND,1}$  & $1/50$ & payoff reduction for primary player used in 1-side of \AND{} gate \\
        $\lambda_{\PURIFY}$ & $1/25$ & payoff reduction for primary player used in \PURIFY{} gate \\
        $\lambda_{\BIAS}$   & $1/50$ & payoff bias for primary player used in \PURIFY{} gate \\
    \end{tabular}
    \caption{Parameters used in the reduction with non-uniform tie-breaking}
    \label{tab:reduc-params-nonunif}
\end{table}

We first prove the hardness result for approximate WSNE. At the end of the proof, we extend this result to approximate NE.
For now, we use unnormalized payoffs and an unnormalized approximation parameter $\eps = 1/100$.
We will later show that the payoffs of the players are bounded by a constant, and therefore, we can scale the approximation parameter by a suitable constant to get hardness for normalized payoffs.

Our constructed game will have exactly $7|V|$ players and $2|V|$ battlefields.
The players will be of three types: {primary}, \textit{auxiliary}, and \textit{dummy}.
All battlefields in this construction are bit battlefields.
For every $v \in V$, we create a primary player $\PrimP^v$ and two bit battlefields $\BitB^v_0$ and $\BitB^v_1$.
The action of $\PrimP^v$ across battlefields $\BitB^v_0$ and $\BitB^v_1$ will encode the value of the variable $v \in V$.
For every bit battlefield $\BitB^v_s$, where $v \in V$ and $s \in \{0,1\}$, we create the auxiliary player $\AuxP^v_s$ and two dummy players $\DumP^v_{s,1}$ and $\DumP^v_{s,2}$.
The auxiliary players help enforce the logic of the \PURE{} gates.
The dummy players create competition in the battlefields and help implement our intended tie-breaking rules.
All players have exactly one resource, and all values are identical and equal to $1$.
Thus, a player's pure strategy is simply the choice of one battlefield.
All payoff differences in the construction will come from the non-uniform tie-breaking rules.

\paragraph{Tie-breaking.}
We now describe the tie-breaking rules. Fix an arbitrary bit battlefield $\BitB^v_s$, where $v \in V$ and $s \in \{0,1\}$.
For this battlefield $\BitB^v_s$, we call the following players \textit{local} to the battlefield, and the battlefield local to the players:
\begin{itemize}
    \item The two dummy players $\DumP^v_{s,1}$ and $\DumP^v_{s,2}$ associated with this battlefield.
    \item The primary player $\PrimP^v$ associated with this battlefield and the sibling battlefield $\BitB^v_{1-s}$.
    \item The auxiliary player $\AuxP^v_s$ who considers this battlefield $\BitB^v_s$ its \textit{input} battlefield. We will call $\AuxP^v_s$ the input-side auxiliary player of $\BitB^v_s$.

    \item Up to two more auxiliary players $\AuxP^u_r$ and $\AuxP^{u'}_{r'}$, for some $u, u' \in V$ and $r, r' \in \{0,1\}$, who consider this battlefield $\BitB^v_s$ as their \textit{output} battlefield.
    Note that $v$ is a variable of \PURE{}, so it must be the output of a gate in the circuit. The variable $u$, and possibly also $u'$, are the one or two inputs to this gate, depending upon the gate. We will explain it properly below when we describe the tie-breaking rules.
    We will call $\AuxP^u_r$ and $\AuxP^{u'}_{r'}$ the output-side auxiliary players of $\BitB^v_s$.
\end{itemize}

Let $S$ denote the set of up to six players described above, local to the battlefield $\BitB^v_s$. Let $W$ be the set of players, among all players in the game, who tie for this battlefield $\BitB^v_s$. The tie-breaking rule goes as follows:
\begin{itemize}
    \item If $|W| = 1$, then there is a unique winner of the battlefield, and this player gets the entire payoff of $1$.
    \item If $|W| > 1$ and $W \cap S = \emptyset$, then all players in $W$ equally share the payoff and get $1/|W|$ each.
    Note that this case and the case above are the only two possible cases with $W \cap S = \emptyset$, i.e., where all the players in $W$ are non-local. All cases after this have $W \cap S \neq \emptyset$, i.e., there will be at least one local player in $W$.

    \item If $|W| > 1$ and $W \cap S \neq \emptyset$, then all non-local players $W \setminus S$ get zero payoff.
    In other words, if there is at least one local player in $W$, then the non-local players in $W$ get zero.

    \item If $|W| > 1$, $W \cap S \neq \emptyset$, and $W \cap \{ \DumP^v_{s,1}, \DumP^v_{s,2} \} = \emptyset$, then all players in $W \cap S$ equally share the payoff and get $1/|W \cap S|$ each.
    In other words, if there is no local dummy player in $W$, then all local players in $W$ equally share the prize.

    \item We now consider the remaining case of $W \cap \{ \DumP^v_{s,1}, \DumP^v_{s,2} \} \neq \emptyset$, i.e., when there is at least one local dummy player in $W$. This case will be the most important for analysis later.

    The payoffs of the local players of $\BitB^v_s$ in $W$ will depend upon whether $v$ is the output of a \NOT{}, \AND{}, or \PURIFY{} gate of the \PURE{} instance. However, for some of the local players, the payoffs will be largely consistent across the gates.

    If there is only one local dummy player in $W$, then that player will get a payoff of at least $7/10$. If both local dummy players are in $W$, each will receive an equal payoff of at least $7/20$.
    The exact payoffs of these local dummy players will be the share remaining from $1$ after allocating the payoffs to the local primary and auxiliary players, as described below.

    The payoff of the input-side local auxiliary player $\AuxP^v_s$ depends only on the primary player $\PrimP^v$. If $\AuxP^v_s \in W$, then $\AuxP^v_s$ gets a payoff of $\beta_I$ if the primary player $\PrimP^v \notin W$ but a payoff of $0$ if $\PrimP^v \in W$.

    We briefly discussed earlier the output-side auxiliary players who consider $\BitB^v_s$ as their local output battlefield.
    The exact auxiliary players who consider $\BitB^v_s$ as their local output battlefield depend upon the gate.
    However, all these output-side local auxiliary players get a payoff of $\beta_O$ if they are in $W$, irrespective of the other players in $W$.

    We have now defined the payoffs of the local dummy, input-side auxiliary, and output-side auxiliary players of battlefield $\BitB^v_s$. It remains to define the payoffs of the local primary player and to specify the local output-side auxiliary players. We do it below for each type of gate.
    \begin{itemize}
        \item $v = \NOT(u)$.
        There is one output-side auxiliary player $\AuxP^u_s$ who considers $\BitB^v_s$ as its output battlefield.
        Assuming the primary player $\PrimP^v \in W$, it gets a payoff of
        \[
            \rho - \lambda_{\NOT} \cdot \Ind[\AuxP^u_s \in W].
        \]
        In other words, $\PrimP^v$ gets a payoff of $\rho$ if the output-side auxiliary player $\AuxP^u_s \notin W$ but a payoff of $\rho - \lambda_{\NOT}$ if $\AuxP^u_s \in W$.

        \item $v = \AND(u, u')$.
        There are two output-side auxiliary players $\AuxP^{u}_{1-s}$ and $\AuxP^{u'}_{1-s}$ who consider $\BitB^v_s$ as their output battlefield.
        If $\PrimP^v \in W$, $\PrimP^v$ gets a payoff of
        \[
            \rho - \lambda_{\AND,1-s} \cdot \left( \Ind[\AuxP^{u}_{1-s} \in W] + \Ind[\AuxP^{u'}_{1-s} \in W] \right).
        \]
        In other words, if $s=1$, then each of the two local output-side auxiliary players $\AuxP^u_0$ and $\AuxP^{u'}_0$ decreases the payoff of $\PrimP^v$ by $\lambda_{\AND,0}$ from the baseline $\rho$, if they are present in $W$.
        The case $s=0$ is analogous, as given in the formula above.

        \item $(v, \cdot) = \PURIFY(u)$.
        If $s = 0$, no output-side auxiliary player considers $\BitB^v_0$ as its output battlefield, but if $s = 1$, the output-side auxiliary player $\AuxP^u_0$ considers $\BitB^v_1$ as its output battlefield.
        If $\PrimP^v \in W$, $\PrimP^v$ gets a payoff of
        \[
            \rho - \lambda_{\BIAS} \cdot \Ind[s=0] - \lambda_{\PURIFY} \cdot \Ind[s=1] \cdot \Ind[\AuxP^u_0 \in W].
        \]
        In other words, $\PrimP^v$'s payoff decreases from its baseline of $\rho$ by $\lambda_{\BIAS}$ if $s=0$ and decreases by $\lambda_{\PURIFY}$ if $s=1$ and $\AuxP^u_0 \in W$.

        \item $(\cdot, v) = \PURIFY(u)$.
        This case is symmetric to the previous case.
        If $s = 0$, the output-side auxiliary player $\AuxP^u_1$ considers $\BitB^v_0$ as its output battlefield, but if $s = 1$,
        no output-side auxiliary player considers $\BitB^v_1$ as its output battlefield.
        If $\PrimP^v \in W$, $\PrimP^v$ gets a payoff of
        \[
            \rho - \lambda_{\BIAS} \cdot \Ind[s=1] - \lambda_{\PURIFY} \cdot \Ind[s=0] \cdot \Ind[\AuxP^u_1 \in W].
        \]
    \end{itemize}
\end{itemize}
With the formal specification of the tie-breaking rule above, our construction is complete.
We now describe how to map an $\eps$-WSNE of the constructed game to a solution of the \PURE{} instance; correctness will be proven later.
Fix an $\eps$-WSNE of this game. Let $p^v_1$ be the probability that the primary player $\PrimP^v$, for $v \in V$, chooses the bit battlefield $\BitB^v_1$.
We set the solution $\valonly$ of the \PURE{} instance as
\begin{equation}\label{eq:mapping-nonunif}
    \val{v} = p^v_1.
\end{equation}

\paragraph{Analysis.} We now prove that the constructed Blotto game encodes a solution to the \PURE{} instance in every $\eps$-WSNE. First, we show that the tie-breaking rule is valid.

\begin{lemma}
\label{lm:tiebreaking-monotone-nonunif}
The tie-breaking rule described above is well-defined and monotone.
\end{lemma}
\begin{proof}
Fix a bit battlefield $\BitB^v_s$.
We first show that the tie-breaking rule is well-defined.
The only case where this is not immediate is the final case, where at least one local dummy player is in $W$.
In this case, we first allocate payoffs to the local primary and auxiliary players, and then give the remaining payoff to the local dummy players. So we need to show that the total payoff allocated to the local primary and auxiliary players is at most $1$.

In fact, we show a stronger bound: the total payoff allocated to the local primary and auxiliary players is at most $7/25$.
First suppose $\PrimP^v \notin W$. Then the input-side auxiliary player can get payoff $\beta_I$, and at most two output-side auxiliary players can get payoff $\beta_O$ each. Thus the total payoff allocated to the local primary and auxiliary players is at most
\[
    \beta_I + 2 \beta_O
    = \frac{1}{8} + \frac1{10}
    = \frac{9}{40}
    < \frac{7}{25}.
\]
Now suppose $\PrimP^v \in W$. Then the input-side auxiliary player gets payoff $0$. The largest possible total payoff to the local primary player and the output-side auxiliary players occurs in the \AND{} gate, when two output-side auxiliary players with pressure $\lambda_{\AND,1}$ are present. In this case, the total payoff is
\[
    \rho - 2\lambda_{\AND,1} + 2\beta_O
    =
    \frac{11}{50} - \frac2{50} + \frac1{10}
    =
    \frac7{25}.
\]
Therefore, the residual payoff for the local dummy players is always at least
\[
    1-\frac7{25}=\frac{18}{25}.
\]
In particular, if exactly one local dummy player is in $W$, then that player gets payoff at least $18/25>7/10$.
If both local dummy players are in $W$, then each gets payoff at least $9/25>7/20$.
Thus the tie-breaking rule is well-defined.

We now prove monotonicity. It is enough to consider the case where we add one player to the tied set.
Fix two tied sets $W$ and $W \cup \{j\}$, where $j \notin W$.
We show that the payoff of every player $i \in W$ weakly decreases when we move from $W$ to $W \cup \{j\}$.

First consider a non-local player $i \in W \setminus S$.
If $W$ contains no local player, then all players in $W$ share the payoff equally. Adding a non-local player decreases $i$'s payoff from $1/|W|$ to $1/(|W|+1)$, and adding a local player makes $i$'s payoff zero.
If $W$ already contains a local player, then $i$ already gets payoff zero and continues to get payoff zero.
Hence the payoff of a non-local player weakly decreases.

Next consider a local dummy player $i \in W$.
In this case, $W$ already contains a local dummy player.
Adding a non-local player does not affect the payoffs of local players, so the residual payoff for the local dummy players does not change.
Adding the second local dummy player splits the same residual payoff between two dummy players instead of giving all of it to one dummy player.
Therefore the payoff of the old local dummy player weakly decreases.

It remains to consider adding a local non-dummy player.
If we add the input-side auxiliary player,
then either $\PrimP^v \notin W$, in which case the input-side auxiliary player receives payoff $\beta_I$ and the residual payoff for the dummy players decreases,
or $\PrimP^v \in W$, in which case the input-side auxiliary player receives payoff zero and the residual payoff does not change.
If we add an output-side auxiliary player, then this new player receives payoff $\beta_O$.
If $\PrimP^v \notin W$, the residual payoff decreases by $\beta_O$.
If $\PrimP^v \in W$, then the payoff of $\PrimP^v$ decreases by the corresponding pressure parameter $\lambda$, while the new output-side auxiliary player receives payoff $\beta_O$. Since $\beta_O > \lambda$ for every pressure parameter $\lambda \in \{\lambda_{\NOT},\lambda_{\AND,0},\lambda_{\AND,1},\lambda_{\PURIFY}\}$, the total payoff allocated to the local primary and auxiliary players increases, and hence the residual payoff for the dummy players decreases.
Finally, if we add the primary player $\PrimP^v$, then the input-side auxiliary player's payoff, if it is present, drops from $\beta_I$ to zero. The new payoff of the primary player is at least
\[
    \rho - 2\lambda_{\AND,0}
    = \frac{11}{50}-\frac2{25}
    = \frac7{50}
    > \frac18
    = \beta_I.
\]
Thus the total payoff allocated to the local primary and auxiliary players weakly increases, and the residual payoff for the dummy players weakly decreases.
Hence the payoff of a local dummy player weakly decreases in all cases.

Next consider the local primary player $\PrimP^v \in W$.
If $W$ contains no local dummy player, then all local players in $W$ equally share the payoff.
Adding a non-local player does not change the payoff of $\PrimP^v$.
Adding another local non-dummy player decreases the equal share of $\PrimP^v$.
Adding a local dummy player moves us to the final case of the tie-breaking rule.
Before adding the local dummy player, there are at most four local non-dummy players in $W$, and hence $\PrimP^v$ gets payoff at least $1/4$.
After adding the local dummy player, $\PrimP^v$ gets payoff at most $\rho = 11/50 < 1/4$.
Hence the payoff of $\PrimP^v$ weakly decreases.

Now suppose $W$ already contains a local dummy player.
Adding a non-local player or the second local dummy player does not affect the payoff of $\PrimP^v$.
Adding the input-side auxiliary player also does not affect the payoff of $\PrimP^v$, because the input-side auxiliary player's payoff depends on $\PrimP^v$, but not vice versa.
Adding an output-side auxiliary player weakly decreases the payoff of $\PrimP^v$ by the corresponding pressure parameter.
Thus the payoff of the local primary player weakly decreases in all cases.

Next consider the input-side auxiliary player $\AuxP^v_s \in W$.
The argument when $W$ contains no local dummy player is the same as in the case of the primary player above.
Indeed, all local players in $W$ equally share the payoff.
Adding a non-local player does not change this payoff, and adding another local non-dummy player decreases it.
If we add a local dummy player, then we move to the final case of the tie-breaking rule.
Before adding the local dummy player, there are at most four local non-dummy players in $W$, and hence $\AuxP^v_s$ gets payoff at least $1/4$.
After adding the local dummy player, $\AuxP^v_s$ gets payoff either $\beta_I=1/8$ or $0$, depending on whether $\PrimP^v$ is absent or present. Hence its payoff weakly decreases.

Now suppose $W$ already contains a local dummy player.
The payoff of $\AuxP^v_s$ is $\beta_I$ if $\PrimP^v \notin W$, and $0$ if $\PrimP^v \in W$.
Adding a non-local player, a local dummy player, or an output-side auxiliary player does not change this payoff.
Adding the primary player changes the payoff of $\AuxP^v_s$ from $\beta_I$ to $0$.
Hence the payoff of the input-side auxiliary player weakly decreases in all cases.

Finally, consider an output-side auxiliary player $i \in W$.
Again, when $W$ contains no local dummy player, the argument is the same as above.
Adding a non-local player does not change the equal share of $i$, and adding another local non-dummy player decreases it.
If we add a local dummy player, then before adding it, $i$ gets payoff at least $1/4$, while after adding it, $i$ gets payoff $\beta_O=1/20$.
Hence its payoff weakly decreases.

If $W$ already contains a local dummy player, then the output-side auxiliary player gets payoff $\beta_O$, irrespective of the other players in $W$. Thus adding any player does not change its payoff. Therefore the payoff of an output-side auxiliary player weakly decreases in all cases.

We have shown that for every old player $i\in W$, adding a new player to the tied set weakly decreases the payoff of $i$. This proves monotonicity.
\end{proof}

We now start the analysis of the $\eps$-WSNE of the constructed game, where $\eps=1/100$.
The first step is to show that the dummy players indeed behave as intended: each dummy player puts its resource on its own battlefield.
This will imply that in every battlefield, there is always at least one local player present, and therefore non-local players cannot profit from interacting with unrelated battlefields.

For convenience, let $A = 7/10$ and $a = 7/20$.
Recall from the proof of \cref{lm:tiebreaking-monotone-nonunif} that if exactly one local dummy player is in the tied set of a battlefield, then this dummy player gets payoff at least $A$, and if both local dummy players are in the tied set, then each of them gets payoff at least $a$.

\begin{lemma}
\label{lm:dummy-home-nonunif}
In every $\eps$-WSNE of the constructed game, every dummy player $\DumP^v_{s,t}$, where $v \in V$, $s \in \{0,1\}$, and $t \in [2]$, plays its own bit battlefield $\BitB^v_s$ with probability $1$.
\end{lemma}

\begin{proof}
For every $v \in V$, $s \in \{0,1\}$, and $t \in [2]$, let
\[
    q^v_{s,t} = \Prob[\DumP^v_{s,t} \text{ does not play } \BitB^v_s],
\]
and let
\[
    \mu^v_s = q^v_{s,1} q^v_{s,2}.
\]
Thus, $\mu^v_s$ is the probability that both local dummy players of $\BitB^v_s$ do not allocate their resource to $\BitB^v_s$.

Suppose, for contradiction, that some dummy player $\DumP^v_{s,t}$ plays a foreign battlefield $\BitB^{v'}_{s'} \neq \BitB^v_s$ with positive probability.
Consider the deviation that moves this dummy player from $\BitB^{v'}_{s'}$ to its own battlefield $\BitB^v_s$.
We first lower bound the gain from $\BitB^v_s$.
If the other local dummy player $\DumP^v_{s,3-t}$ allocates its resource to $\BitB^v_s$, then after the deviation $\DumP^v_{s,t}$ gets payoff at least $a$ from $\BitB^v_s$.
If $\DumP^v_{s,3-t}$ does not allocate its resource to $\BitB^v_s$, then after the deviation $\DumP^v_{s,t}$ is the only local dummy player allocating to $\BitB^v_s$, and hence gets payoff at least $A \ge a$ from $\BitB^v_s$.
The only small subtlety is the no-resource case, where before the deviation the dummy player may already receive some payoff from $\BitB^v_s$.
But in the no-resource case, the two local dummy players of $\BitB^v_s$ split the residual payoff, so $\DumP^v_{s,t}$ gets payoff at most $1/2$ before the deviation, while after the deviation it is the unique positive allocator and gets payoff $1$.
Thus the gain from $\BitB^v_s$ is at least $1/2 \ge a$ in this case as well.

Now consider the possible loss from leaving $\BitB^{v'}_{s'}$.
If at least one local dummy player of $\BitB^{v'}_{s'}$ allocates its resource to $\BitB^{v'}_{s'}$, then $\DumP^v_{s,t}$ is a non-local player on $\BitB^{v'}_{s'}$ in the presence of a local player, and therefore gets payoff zero from $\BitB^{v'}_{s'}$.
Thus the loss from $\BitB^{v'}_{s'}$ can be positive only when both local dummy players of $\BitB^{v'}_{s'}$ do not allocate their resource to $\BitB^{v'}_{s'}$.
This event happens with probability $\mu^{v'}_{s'}$, and the loss on this event is at most $1$.
Therefore, the deviation from $\BitB^{v'}_{s'}$ to $\BitB^v_s$ improves the payoff of $\DumP^v_{s,t}$ by at least
\[
    a - \mu^{v'}_{s'}.
\]
Since $\BitB^{v'}_{s'}$ is in the support of $\DumP^v_{s,t}$ in an $\eps$-WSNE, this improvement must be at most $\eps$. Hence
\[
    \mu^{v'}_{s'} \ge a - \eps.
\]

Call a battlefield $\BitB^v_s$ bad if $\mu^v_s \ge a-\eps$.
We have shown that every foreign battlefield played by a dummy player with positive probability must be bad.
Suppose there is some dummy player who plays a foreign battlefield with positive probability. Then there is at least one bad battlefield.
Build a directed graph whose vertices are the bad battlefields.
Put an edge $\BitB^v_s \to \BitB^{v'}_{s'}$ if one of the local dummy players of $\BitB^v_s$ plays the foreign battlefield $\BitB^{v'}_{s'}$ with positive probability.
By the argument above, every edge points to a bad battlefield.
Moreover, every bad battlefield has an outgoing edge. Indeed, if $\BitB^v_s$ is bad, then
\[
    \mu^v_{s} = q^v_{s,1} q^v_{s,2} \ge a - \eps > 0,
\]
and hence both local dummy players of $\BitB^v_s$ leave $\BitB^v_s$ with positive probability.

Take a sink strongly connected component $T$ of this graph. The local dummy players of the battlefields in $T$ never play outside $T$; they either play their own battlefield or another battlefield in $T$.
Therefore, all payoff obtained by these dummy players comes from battlefields in $T$.
Since each battlefield distributes total payoff at most $1$, we get
\[
    \sum_{\BitB^v_s \in T}
    \left(
        u_{\DumP^v_{s,1}} + u_{\DumP^v_{s,2}}
    \right)
    \le |T|,
\]
where $u_{\DumP^v_{s,t}}$ is the expected payoff of the dummy player $\DumP^v_{s,t}$.

On the other hand, for any battlefield $\BitB^v_s \in T$, the deviation of $\DumP^v_{s,1}$ to its own battlefield $\BitB^v_s$ gives payoff at least
\[
    a (1-q^v_{s,2}) + A q^v_{s,2}
    = a + (A-a)q^v_{s,2}.
\]
Similarly, the deviation of $\DumP^v_{s,2}$ to $\BitB^v_s$ gives payoff at least
\[
    a + (A-a)q^v_{s,1}.
\]
Since we are in an $\eps$-WSNE, the expected payoff of each dummy player is within $\eps$ of the payoff from deviating to its own battlefield. Therefore,
\[
    u_{\DumP^v_{s,1}} + u_{\DumP^v_{s,2}}
    \ge 2a + (A-a)(q^v_{s,1}+q^v_{s,2}) - 2\eps.
\]
Since every battlefield in $T$ is bad, we have
\[
    q^v_{s,1} q^v_{s,2} \ge a-\eps.
\]
Thus,
\[
    q^v_{s,1} + q^v_{s,2} \ge 2\sqrt{a-\eps}.
\]
Combining the last two inequalities gives
\[
    u_{\DumP^v_{s,1}} + u_{\DumP^v_{s,2}}
    \ge 2a + 2(A-a)\sqrt{a-\eps} - 2\eps.
\]
Substituting $A=7/10$, $a=7/20$, and $\eps=1/100$, the right-hand side is
\[
    \frac7{10} + \frac7{10} \sqrt{\frac7{20} - \frac1{100}} - \frac1{50}
    > 1.
\]
Therefore, for every $\BitB^v_s \in T$, the two local dummy players of $\BitB^v_s$ get total payoff strictly larger than $1$.
Summing over all $\BitB^v_s \in T$, we get
\[
    \sum_{\BitB^v_s \in T}
    \left(
        u_{\DumP^v_{s,1}} + u_{\DumP^v_{s,2}}
    \right)
    > |T|,
\]
which contradicts the upper bound above.
Hence no dummy player plays a foreign battlefield with positive probability.
Therefore every dummy player plays its own battlefield with probability $1$.
\end{proof}

\cref{lm:dummy-home-nonunif} implies that, in every $\eps$-WSNE, every battlefield always has its two local dummy players present. Therefore, whenever a non-local player allocates its resource to a battlefield, it gets payoff zero from that battlefield. We next use this observation to show that primary and auxiliary players also behave locally.

\begin{lemma}
\label{lm:local-actions-nonunif}
In every $\eps$-WSNE of the constructed game, every primary player $\PrimP^v$, where $v \in V$, only plays the two local bit battlefields $\BitB^v_0$ and $\BitB^v_1$ with positive probability.
Further, every auxiliary player $\AuxP^v_s$, where $v \in V$ and $s \in \{0,1\}$, only plays its input battlefield $\BitB^v_s$ and its output battlefield, depending upon the gate, with positive probability.
\end{lemma}

\begin{proof}
By \cref{lm:dummy-home-nonunif}, both local dummy players of every battlefield play their own local battlefield with probability $1$.
Hence, if any player allocates its resource to a battlefield to which it is non-local, then this player gets payoff zero from that battlefield.

We first consider a primary player $\PrimP^v$. The player $\PrimP^v$ is local to exactly two battlefields, namely $\BitB^v_0$ and $\BitB^v_1$. If $\PrimP^v$ allocates its resource to any other battlefield, then it gets payoff zero. On the other hand, by allocating its resource to either $\BitB^v_0$ or $\BitB^v_1$, it gets payoff at least
\[
    \rho - 2\lambda_{\AND,0}
    = \frac{11}{50} - \frac2{25}
    = \frac7{50}
    > \eps.
\]
Indeed, this is the smallest possible payoff of a primary player on a local battlefield: it occurs when two output-side auxiliary players are present and each decreases the primary player's payoff by the larger pressure parameter $\lambda_{\AND,0}$. Since $7/50 > \eps$, no non-local battlefield can be in the support of $\PrimP^v$ in an $\eps$-WSNE.

Now consider an auxiliary player $\AuxP^v_s$.
This player is local to exactly two battlefields: its input battlefield $\BitB^v_s$ and its output battlefield, determined by the gate in which $v$ is an input.
If $\AuxP^v_s$ allocates its resource to any other battlefield, then it gets payoff zero.
On the other hand, by allocating its resource to its output battlefield, it gets payoff exactly $\beta_O = 1/20$.
Since $\beta_O > \eps$, no non-local battlefield can be in the support of $\AuxP^v_s$ in an $\eps$-WSNE.
\end{proof}

For every variable $v \in V$ and $s \in \{0,1\}$, let $p^v_s$ be the probability that the primary player $\PrimP^v$ chooses the bit battlefield $\BitB^v_s$.
Recall that $p^v_1$ and $\val{v} = p^v_1$ were defined in equation \eqref{eq:mapping-nonunif}.
By \cref{lm:local-actions-nonunif}, the primary player $\PrimP^v$ only plays $\BitB^v_0$ and $\BitB^v_1$ with positive probability. Therefore, $p^v_0 + p^v_1 = 1$.
Thus, if $\PrimP^v$ plays $\BitB^v_0$ with probability $1$, then $\val{v}=0$, and if $\PrimP^v$ plays $\BitB^v_1$ with probability $1$, then $\val{v}=1$.

We next analyze the auxiliary players.
Recall that the auxiliary player $\AuxP^v_s$ has two local battlefields: the input battlefield $\BitB^v_s$ and its output battlefield.
The purpose of $\AuxP^v_s$ is to detect whether $\PrimP^v$ is absent from $\BitB^v_s$.
If $\PrimP^v$ is absent from $\BitB^v_s$ with large enough probability, then $\AuxP^v_s$ prefers to stay on its input battlefield.
On the other hand, if $\PrimP^v$ is present on $\BitB^v_s$ with large probability, then $\AuxP^v_s$ gets small payoff from its input battlefield and moves to its output battlefield.

\begin{lemma}
\label{lm:aux-behavior-nonunif}
Fix $v \in V$ and $s \in \{0,1\}$.
If $p^v_s = 1$, then the auxiliary player $\AuxP^v_s$ plays its output battlefield with probability $1$.
If $p^v_s \le 1/2$, then $\AuxP^v_s$ plays its input battlefield $\BitB^v_s$ with probability $1$.
Moreover, if $\AuxP^v_s$ plays its output battlefield with positive probability, then $p^v_s > 1/2$.
\end{lemma}

\begin{proof}
By \cref{lm:dummy-home-nonunif}, both local dummy players of every battlefield play their own battlefield with probability $1$.
Therefore, if $\AuxP^v_s$ plays its output battlefield, then it gets payoff exactly $\beta_O$.
On the other hand, if $\AuxP^v_s$ plays its input battlefield $\BitB^v_s$, then it gets payoff $\beta_I$ exactly when $\PrimP^v$ is absent from $\BitB^v_s$, and payoff $0$ when $\PrimP^v$ is present.
Thus the expected payoff of $\AuxP^v_s$ from its input battlefield is $\beta_I (1 - p^v_s)$.

We now compare these two payoffs.
If $p^v_s = 1$, then the input battlefield gives payoff $0$, whereas the output battlefield gives payoff
$\beta_O = 1/20 > \eps$.
Hence the input battlefield cannot be in the support of $\AuxP^v_s$.
By \cref{lm:local-actions-nonunif}, $\AuxP^v_s$ only plays its input and output battlefields, and so $\AuxP^v_s$ plays its output battlefield with probability $1$.

If $p^v_s \le 1/2$, then the input battlefield gives payoff at least
$\beta_I/2 = 1/16$,
whereas the output battlefield gives payoff $\beta_O = 1/20$.
Since $1/16 - 1/20 = 1/80 > \eps$,
the output battlefield cannot be in the support of $\AuxP^v_s$.
Again by \cref{lm:local-actions-nonunif}, $\AuxP^v_s$ plays its input battlefield with probability $1$.

Finally, suppose $\AuxP^v_s$ plays its output battlefield with positive probability. Since we are in an $\eps$-WSNE, the output battlefield must be an $\eps$-best response. Therefore,
\[
    \beta_O + \eps \ge \beta_I (1-p^v_s).
\]
Rearranging,
\[
    p^v_s \ge 1 - \frac{\beta_O + \eps}{\beta_I}
    = 1 - \frac{1/20 + 1/100}{1/8}
    > \frac12.
\]
This proves the lemma.
\end{proof}

We now show that the assignment $\valonly$ defined in equation \eqref{eq:mapping-nonunif} satisfies all the gates of the \PURE{} instance. We verify the three types of gates one by one.

\medskip
\noindent
$v=\NOT(u)$.
Suppose $\val{u} = 0$. Then $p^u_0 = 1$ and $p^u_1 = 0$.
By \cref{lm:aux-behavior-nonunif}, the auxiliary player $\AuxP^u_0$ plays its output battlefield $\BitB^v_0$ with probability $1$, while the auxiliary player $\AuxP^u_1$ plays its input battlefield $\BitB^u_1$ with probability $1$.
Therefore, the payoff of $\PrimP^v$ from $\BitB^v_0$ is $\rho - \lambda_{\NOT}$, while its payoff from $\BitB^v_1$ is $\rho$.
Since $\lambda_{\NOT} = 1/25 > \eps$, the battlefield $\BitB^v_0$ cannot be in the support of $\PrimP^v$ in an $\eps$-WSNE.
Hence $\PrimP^v$ plays $\BitB^v_1$ with probability $1$, and $\val{v} = 1$.
The case $\val{u}=1$ is symmetric: $\AuxP^u_1$ plays its output battlefield $\BitB^v_1$ with probability $1$, while $\AuxP^u_0$ stays on its input battlefield, and therefore $\PrimP^v$ plays $\BitB^v_0$ with probability $1$.
Thus $\val{v} = 0$.
If $\val{u} \in (0,1)$, then the \NOT{} gate imposes no condition.

\medskip
\noindent
$w = \AND(u, v)$.
Suppose $\val{u} = 0$.
Then $p^u_0 = 1$ and $p^u_1 = 0$.
By \cref{lm:aux-behavior-nonunif}, the auxiliary player $\AuxP^u_0$ plays its output battlefield $\BitB^w_1$ with probability $1$, while $\AuxP^u_1$ stays on its input battlefield.
Therefore, the payoff of $\PrimP^w$ from $\BitB^w_1$ is at most $\rho -\lambda_{\AND,0}$.
On the other hand, the payoff of $\PrimP^w$ from $\BitB^w_0$ is at least $\rho - \lambda_{\AND,1}$, since at most one output-side auxiliary player, namely $\AuxP^v_1$, can create pressure on $\BitB^w_0$.
Since $ \lambda_{\AND,0} - \lambda_{\AND,1} = 1/50 > \eps$, the battlefield $\BitB^w_1$ cannot be in the support of $\PrimP^w$ in an $\eps$-WSNE. Hence $\PrimP^w$ plays $\BitB^w_0$ with probability $1$, and $\val{w} = 0$. The same argument applies if $\val{v} = 0$.

Now suppose $\val{u} = \val{v} = 1$. Then $p^u_1 = p^v_1 = 1$ and $p^u_0 = p^v_0 = 0$.
By \cref{lm:aux-behavior-nonunif}, the auxiliary players $\AuxP^u_1$ and $\AuxP^v_1$ play their output battlefield $\BitB^w_0$ with probability $1$, while $\AuxP^u_0$ and $\AuxP^v_0$ stay on their input battlefields.
Therefore, the payoff of $\PrimP^w$ from $\BitB^w_0$ is $\rho - 2\lambda_{\AND,1}$, while its payoff from $\BitB^w_1$ is $\rho$.
Since $2\lambda_{\AND,1} = 1/25 > \eps$, the battlefield $\BitB^w_0$ cannot be in the support of $\PrimP^w$ in an $\eps$-WSNE.
Hence $\PrimP^w$ plays $\BitB^w_1$ with probability $1$, and $\val{w} = 1$.

All other input cases are unconstrained by the \AND{} gate.

\medskip
\noindent
$(v,w) = \PURIFY(u)$.
We first consider the case $\val{u} = 0$.
Then $p^u_0 = 1$ and $p^u_1 = 0$. By \cref{lm:aux-behavior-nonunif}, the auxiliary player $\AuxP^u_0$ plays its output battlefield $\BitB^v_1$ with probability $1$, while $\AuxP^u_1$ stays on its input battlefield.
For the primary player $\PrimP^v$, the payoff from $\BitB^v_0$ is $\rho - \lambda_{\BIAS}$, while the payoff from $\BitB^v_1$ is $\rho - \lambda_{\PURIFY}$.
Since $\lambda_{\PURIFY} - \lambda_{\BIAS} = 1/50 > \eps$, the battlefield $\BitB^v_1$ cannot be in the support of $\PrimP^v$.
Hence $\PrimP^v$ plays $\BitB^v_0$ with probability $1$, and $\val{v} = 0$.
For the primary player $\PrimP^w$, the payoff from $\BitB^w_0$ is $\rho$, while the payoff from $\BitB^w_1$ is $\rho - \lambda_{\BIAS}$.
Since $\lambda_{\BIAS} = 1/50 > \eps$, the player $\PrimP^w$ plays $\BitB^w_0$ with probability $1$, and $\val{w} = 0$.

The case $\val{u} = 1$ is symmetric.
In this case, $\AuxP^u_1$ plays its output battlefield $\BitB^w_0$ with probability $1$, while $\AuxP^u_0$ stays on its input battlefield.
For $\PrimP^v$, the payoff from $\BitB^v_0$ is $\rho - \lambda_{\BIAS}$, while the payoff from $\BitB^v_1$ is $\rho$. Therefore $\PrimP^v$ plays $\BitB^v_1$ with probability $1$, and $\val{v}=1$.
For $\PrimP^w$, the payoff from $\BitB^w_0$ is $\rho - \lambda_{\PURIFY}$, while the payoff from $\BitB^w_1$ is $\rho - \lambda_{\BIAS}$.
Since $\lambda_{\PURIFY} - \lambda_{\BIAS} > \eps$, the player $\PrimP^w$ plays $\BitB^w_1$ with probability $1$, and $\val{w}=1$.

It remains to consider the case $\val{u} \in (0,1)$.
Since $p^u_0 + p^u_1 = 1$, at least one of $p^u_0$ or $p^u_1$ is at most $1/2$.
If $p^u_0 \le 1/2$, then by \cref{lm:aux-behavior-nonunif}, the auxiliary player $\AuxP^u_0$ stays on its input battlefield.
Thus there is no output-side auxiliary player on $\BitB^v_1$.
The payoff of $\PrimP^v$ from $\BitB^v_0$ is $\rho - \lambda_{\BIAS}$, while the payoff from $\BitB^v_1$ is $\rho$.
Since $\lambda_{\BIAS} > \eps$, the player $\PrimP^v$ plays $\BitB^v_1$ with probability $1$, and $\val{v} = 1$.
Similarly, if $p^u_1 \le 1/2$, then $\AuxP^u_1$ stays on its input battlefield, so there is no output-side auxiliary player on $\BitB^w_0$.
The payoff of $\PrimP^w$ from $\BitB^w_0$ is $\rho$, while the payoff from $\BitB^w_1$ is $\rho - \lambda_{\BIAS}$.
Hence $\PrimP^w$ plays $\BitB^w_0$ with probability $1$, and $\val{w} = 0$.
Therefore, when $\val{u} \in (0,1)$, at least one of the two outputs is pure, as required by the \PURIFY{} gate.

\medskip
We have shown that the assignment $\valonly$ satisfies every \NOT{}, \AND{}, and \PURIFY{} gate of the \PURE{} instance. Therefore, every $\eps$-WSNE of the constructed game gives a solution to the \PURE{} instance. Since \PURE{} is PPAD-hard, computing an $\eps$-WSNE of the constructed game is PPAD-hard for the unnormalized payoff scale used above.

We now discuss normalization. Although the constructed game has $2|V|$ battlefields, every player's total payoff is bounded by a constant. Indeed, if a player chooses a non-local battlefield, then it can get payoff from that battlefield only through the resource it allocates there, and this payoff is at most $1$. On the other hand, the player can receive payoff from at most two local battlefields: for a primary player $\PrimP^v$ these are $\BitB^v_0$ and $\BitB^v_1$, for an auxiliary player $\AuxP^v_s$ these are its input battlefield $\BitB^v_s$ and output battlefield, depending upon the gate, and for a dummy player $\DumP^v_{s,t}$ this is only its local battlefield $\BitB^v_s$. Each battlefield contributes payoff at most $1$. Thus every player's total payoff is at most $3$.

Therefore, by scaling all values down by a factor of $3$, we obtain an equivalent normalized game with utilities in $[0,1]$. This scaling preserves the fact that all players have identical values and all payoff differences come only from the tie-breaking rules. It only scales the approximation parameter by the same factor. Hence the construction gives PPAD-hardness for computing an $(\eps/3)$-WSNE in the normalized version.

\paragraph{Approximate-NE.}

We now prove the hardness result for approximate NE.
We work with the original unnormalized game constructed above for the approximate-WSNE hardness result.
The same normalization argument used for approximate WSNE also applies to approximate NE, so we omit it here.

Let
\[
    \Delta = \frac{1}{400}, \qquad
    \eta=\frac{1}{10000}, \qquad
    \eps' = \Delta \eta = \frac{1}{4000000}.
\]
The next lemma shows that an $\eps'$-NE of our constructed game can be converted into an $\eps$-WSNE of the same game. Since we have already shown PPAD-hardness of computing an $\eps$-WSNE, this will imply PPAD-hardness of computing an $\eps'$-NE.

\begin{lemma}
\label{lm:ne-to-wsne-nonunif}
Given an $\eps'$-NE of the constructed game, one can compute an $\eps$-WSNE of the same game in polynomial time.
\end{lemma}

\begin{proof}
Let $\bx$ be an $\eps'$-NE of the constructed game. For every player $i$, let
\[
    M_i = \max_{a_i \in \cA_i} u_i(a_i, \bx_{-i})
\]
be the payoff of a best response to $\bx_{-i}$, where $\cA_i$ is the pure strategy set of player $i$. We call an action $a_i$ bad if
\[
    u_i(a_i, \bx_{-i}) < M_i - \Delta.
\]
Since $\bx$ is an $\eps'$-NE, player $i$ puts probability at most
\[
    \frac{\eps'}{\Delta} = \eta
\]
on bad actions. Delete all bad actions from the support of every player and renormalize the remaining probabilities. Let $\by$ be the resulting mixed strategy profile.

We first show that every dummy player plays its own battlefield with probability $1$ in $\by$.
The argument extends the ideas in the proof of \cref{lm:dummy-home-nonunif}.
For every $v \in V$, $s \in \{0,1\}$, and $t \in [2]$, let
\[
    q^v_{s,t} = \Prob_{\bx}[\DumP^v_{s,t} \text{ does not play } \BitB^v_s],
\]
and let
\[
    \mu^v_s = q^v_{s,1} q^v_{s,2}.
\]
Thus $\mu^v_s$ is the probability, under the profile $\bx$, that both local dummy players of $\BitB^v_s$ are absent from $\BitB^v_s$.

Suppose, for contradiction, that some foreign dummy action survives the deletion step.
That is, suppose some dummy player $\DumP^v_{s,t}$ has a surviving action to a foreign battlefield $\BitB^{v'}_{s'} \neq \BitB^v_s$.
As in the proof of \cref{lm:dummy-home-nonunif}, deviating from $\BitB^{v'}_{s'}$ to its own battlefield $\BitB^v_s$ improves the payoff of $\DumP^v_{s,t}$ by at least
\[
    a - \mu^{v'}_{s'}.
\]
Since the action $\BitB^{v'}_{s'}$ survived the deletion step, this improvement must be at most $\Delta$. Hence
\[
    \mu^{v'}_{s'} \ge a-\Delta.
\]

Call a battlefield $\BitB^v_s$ heavy if $\mu^v_s \ge a - \Delta$. We have shown that every surviving foreign dummy action points to a heavy battlefield. Suppose there is some surviving foreign dummy action. Then there is at least one heavy battlefield. Build a directed graph whose vertices are the heavy battlefields. Put an edge
\[
    \BitB^v_s \to \BitB^{v'}_{s'}
\]
if one of the local dummy players of $\BitB^v_s$ has a surviving foreign action to $\BitB^{v'}_{s'}$. By the argument above, every edge points to a heavy battlefield.

Moreover, every heavy battlefield has an outgoing edge. Indeed, if $\BitB^v_s$ is heavy, then
\[
    q^v_{s,1}q^v_{s,2}\ge a-\Delta.
\]
Since $q^v_{s,1},q^v_{s,2}\le 1$, this implies
\[
    q^v_{s,t}\ge a-\Delta
\]
for both $t\in[2]$. Each dummy player puts probability at most $\eta$ on deleted actions, and $a-\Delta>\eta$.
Therefore, each local dummy player of $\BitB^v_s$ has positive probability on
surviving foreign actions, and hence $\BitB^v_s$ has an outgoing edge.

Take a sink strongly connected component $T$ of this graph.
The local dummy players of the battlefields in $T$ never play surviving actions outside $T$. Therefore, under the original profile $\bx$, these dummy players can obtain a payoff from battlefields outside $T$ only when they play deleted actions. Each dummy player puts probability at most $\eta$ on deleted actions, and the payoff from such an outside action is at most $1$. Hence
\[
    \sum_{\BitB^v_s\in T}
    \left(
        u_{\DumP^v_{s,1}}(\bx)+u_{\DumP^v_{s,2}}(\bx)
    \right)
    \le
    |T|+2\eta |T|.
\]
On the other hand, the same deviation argument as in \cref{lm:dummy-home-nonunif} gives, for every $\BitB^v_s\in T$,
\[
    u_{\DumP^v_{s,1}}(\bx)+u_{\DumP^v_{s,2}}(\bx)
    \ge
    2a+(A-a)(q^v_{s,1}+q^v_{s,2})-2\eps'.
\]
Since every battlefield in $T$ is heavy, we have
\[
    q^v_{s,1}q^v_{s,2}\ge a-\Delta,
\]
and hence
\[
    q^v_{s,1}+q^v_{s,2}\ge 2\sqrt{a-\Delta}.
\]
Therefore,
\[
    u_{\DumP^v_{s,1}}(\bx)+u_{\DumP^v_{s,2}}(\bx)
    \ge
    2a+2(A-a)\sqrt{a-\Delta}-2\eps'.
\]
Substituting
\[
    A=\frac7{10}, \qquad
    a=\frac7{20}, \qquad
    \Delta=\frac1{400}, \qquad
    \eps'=\frac1{4000000},
\]
the right-hand side is strictly larger than $1+2\eta$. Summing over all $\BitB^v_s\in T$ contradicts the upper bound $|T|+2\eta |T|$.
Therefore no foreign dummy action survives the deletion step. Hence every dummy player plays its own battlefield with probability $1$ in $\by$.

We will also use the following consequence for the original profile $\bx$.
Since all foreign dummy actions are deleted, and since every dummy player puts probability at most $\eta$ on deleted actions, every dummy player plays its own battlefield with probability at least $1-\eta$ in $\bx$.
Therefore, for every battlefield $\BitB^v_s$,
\[
    \mu^v_s \le \eta^2.
\]

We next show that no non-local action of a primary or auxiliary player survives the deletion step.
Consider a primary or auxiliary player $i$, and suppose $i$ chooses a battlefield to which it is non-local.
The payoff from this chosen battlefield can be positive only when both local dummy players of that battlefield are absent, which happens with probability at most $\eta^2$ under $\bx$.
In addition, while playing this non-local action, player $i$ may receive payoff from empty battlefields to which it is local.
Every primary or auxiliary player is local to at most two battlefields, and each such empty battlefield contribution also requires both local dummy players of that battlefield to be absent.
Therefore every non-local action of a primary or auxiliary player gives payoff at most $3\eta^2$ against $\bx_{-i}$.

On the other hand, every primary player has a local action with payoff at least $\rho - 2 \lambda_{\AND,0} = 7/50$,
and every auxiliary player has its output action with payoff $\beta_O = 1/20$.
Since $\frac1{20} - 3\eta^2 > \Delta$, every non-local action of a primary or auxiliary player is bad and is deleted.
Thus, in $\by$, every player only plays local actions: each dummy player plays its own battlefield, each primary player plays one of its two battlefields, and each auxiliary player plays either its input battlefield or its output battlefield.

It remains to show that $\by$ is an $\eps$-WSNE. For every player $i$, the total variation distance between $\bx_i$ and $\by_i$ is at most $\eta$, because $\by_i$ is obtained from $\bx_i$ by deleting a set of actions of total probability at most $\eta$ and renormalizing.

Fix a player $i$ and a local action $a_i$ of player $i$. The player $i$ is local to at most two battlefields.
For each such battlefield, once at least one local dummy player is present, the payoff contribution depends only on the local players of that battlefield: the two dummy players, the primary player, the input-side auxiliary player, and at most two output-side auxiliary players.
Thus, excluding $i$, at most five relevant opponents can directly affect the payoff contribution of one local battlefield. The product distribution of these relevant opponents changes by total variation distance at most $5\eta$.
The only remaining issue is the event that both local dummy players of this battlefield are absent under $\bx$, which has probability at most $\eta^2$. Therefore, the payoff contribution from one local battlefield changes by at most
\[
    5\eta + \eta^2.
\]
Since $i$ is local to at most two battlefields, we get
\[
    \left|
        u_i(a_i,\by_{-i})-u_i(a_i,\bx_{-i})
    \right|
    \le 2(5\eta+\eta^2)
    < 20\eta.
\]

Now fix a player $i$ and an action $a_i$ in the support of $\by_i$. Since $a_i$ survived the deletion
step,
\[
    u_i(a_i,\bx_{-i}) \ge M_i-\Delta.
\]
Let $a'_i$ be any deviation. If $a'_i$ is non-local, then it gives payoff zero against $\by_{-i}$, because every dummy player plays its own battlefield in $\by$. Hence such a deviation cannot violate the $\eps$-WSNE condition. If $a'_i$ is local, then
\[
    u_i(a'_i,\bx_{-i}) \le M_i.
\]
Using the perturbation bound for $a_i$ and $a'_i$, we get
\[
    u_i(a'_i, \by_{-i}) - u_i(a_i, \by_{-i})
    \le \Delta+40\eta.
\]
Finally,
\[
    \Delta+40\eta
    =
    \frac1{400}+\frac{40}{10000}
    =
    \frac{13}{2000}
    <
    \frac1{100}
    =
    \eps.
\]
Thus every action in the support of $\by_i$ is an $\eps$-best response to $\by_{-i}$. Therefore $\by$ is an $\eps$-WSNE.
\end{proof}

If there were a polynomial-time algorithm for computing an $\eps'$-NE of the constructed game,
then by \cref{lm:ne-to-wsne-nonunif}, we could compute an $\eps$-WSNE of the constructed game in polynomial time.
This contradicts the PPAD-hardness of computing an $\eps$-WSNE proved above.
Hence computing an $\eps'$-NE of the constructed game is PPAD-hard.
\end{proof}

\subsection{Analysis of \cref{ex:isolated-irrational}}
\label[appendix]{app:proof:isolated-irrational}

We first verify the isolated irrational equilibrium from \cref{ex:isolated-irrational}. At the end, we also give a rational equilibrium.

The game has three players and three battlefields. The budgets are
\[
    B_1 = 3,
    \qquad B_2 = B_3 = 1,
\]
and the value vectors are
\[
    v_1 = (4,7,6), \qquad
    v_2 = (3,0,6), \qquad
    v_3 = (0,7,4).
\]
Tie-breaking is uniform.

Let the players play the following strategies:
\begin{itemize}
    \item player $1$ mixes between $L=(0,1,2)$ and $R=(0,2,1)$, playing $L$ with probability $x$;
    \item player $2$ mixes between $C_2=(0,0,1)$ and $A=(1,0,0)$, playing $C_2$ with probability $y$;
    \item player $3$ mixes between $C_3=(0,0,1)$ and $B=(0,1,0)$, playing $C_3$ with probability $z$.
\end{itemize}

We first compute the indifference conditions of the players on the proposed support.
In each expression below, we condition on the pure actions of the other two players and sum the payoff obtained from the three battlefields in that case.
For player $1$,
\begin{align*}
    u_1(L)
    &= yz\Bigl(\frac43+7+6\Bigr)
     + y(1-z)\Bigl(\frac43+\frac72+6\Bigr) \\
    &\qquad
     + (1-y)z(0+7+6)
     + (1-y)(1-z)\Bigl(0+\frac72+6\Bigr) \\
    &= \Bigl(0+\frac72+6\Bigr)
     + y\Bigl(\frac43+\frac72+6-\bigl(0+\frac72+6\bigr)\Bigr) \\
    &\qquad
     + z\Bigl(0+7+6-\bigl(0+\frac72+6\bigr)\Bigr) \\
    &\qquad
     + yz\Bigl(\frac43+7+6
         -\bigl(\frac43+\frac72+6\bigr)
         -(0+7+6)
         +\bigl(0+\frac72+6\bigr)\Bigr) \\
    &= \frac{19}{2}+\frac{4}{3}y+\frac{7}{2}z
     = \frac{8y + 21z + 57}{6},
\end{align*}
\begin{align*}
    u_1(R)
    &= yz\Bigl(\frac43+7+2\Bigr)
     + y(1-z)\Bigl(\frac43+7+3\Bigr) \\
    &\qquad
     + (1-y)z(0+7+3)
     + (1-y)(1-z)(0+7+6) \\
    &= (0+7+6)
     + y\Bigl(\frac43+7+3-(0+7+6)\Bigr) \\
    &\qquad
     + z\Bigl(0+7+3-(0+7+6)\Bigr) \\
    &\qquad
     + yz\Bigl(\frac43+7+2
         -\bigl(\frac43+7+3\bigr)
         -(0+7+3)
         +(0+7+6)\Bigr) \\
    &= 13 - \frac{5}{3} y - 3z + 2yz
     = \frac{6yz - 5y - 9z + 39}{3}.
\end{align*}
Thus $u_1(L)=u_1(R)$ if and only if
\begin{equation} \label{eq:isolated-irrational:1}
    4yz - 6y - 13z + 7 = 0.
\end{equation}
For player $2$,
\begin{align*}
    u_2(A) &= \big( xz + x(1-z) + (1-x)z + (1-x)(1-z) \big) (3+0+0) = 3, \\
    u_2(C_2)
    &= xz(1+0+0)
     + x(1-z)(1+0+0) \\
    &\qquad
     + (1-x)z(1+0+2)
     + (1-x)(1-z)(1+0+3) \\
    &= 4 - 3x - z + xz.
\end{align*}
Thus $u_2(A) = u_2(C_2)$ if and only if
\begin{equation} \label{eq:isolated-irrational:2}
    xz - 3x - z + 1 = 0.
\end{equation}
The only off-support action of player $2$ is $(0,1,0)$, and
\begin{align*}
    u_2((0,1,0)) = 1.
\end{align*}
For player $3$,
\begin{align*}
    u_3(B)
    &= xy \Bigl( 0 + \frac72 + 0 \Bigr)
     + x(1-y) \Bigl( 0 + \frac72 + 0 \Bigr) \\
    &\qquad
     + (1-x)y(0+0+0)
     + (1-x)(1-y)(0+0+0) \\
    &= \frac{7}{2} x, \\
    u_3(C_3)
    &= xy(0+0+0)
     + x(1-y)(0+0+0) \\
    &\qquad
     + (1-x)y\Bigl(0+0+\frac43\Bigr)
     + (1-x)(1-y)(0+0+2) \\
    &= 2 - 2x - \frac{2}{3} y + \frac{2}{3} xy
     = \frac{2(1-x)(3-y)}{3}.
\end{align*}
Thus $u_3(B) = u_3(C_3)$ if and only if
\begin{equation} \label{eq:isolated-irrational:3}
    4xy - 33x - 4y + 12 = 0.
\end{equation}
The only off-support action of player $3$ is $(1,0,0)$, and
\begin{align*}
    u_3((1,0,0)) = 0.
\end{align*}

The three indifference equations
\eqref{eq:isolated-irrational:1}--\eqref{eq:isolated-irrational:3}
are satisfied by
\begin{align*}
    x^* = \frac{-73 + \sqrt{8689}}{70}, \qquad
    y^* = \frac{121 - \sqrt{8689}}{32}, \qquad
    z^* = \frac{109 - \sqrt{8689}}{84}.
\end{align*}
Since
\begin{align*}
    93 < \sqrt{8689} < 94,
\end{align*}
we have
\begin{align*}
    0 < x^* < 1,\qquad 0 < y^* < 1,\qquad 0 < z^* < 1.
\end{align*}
Thus all support actions are played with positive probability. Since $8689$ is not a square, the probabilities $x^*,y^*,z^*$ are irrational.

It remains to check that the off-support actions are strictly worse.
For player $2$, the only off-support action is $(0,1,0)$, and its payoff is $1$, whereas the two support actions have payoff $3$.
Hence player $2$ has no profitable off-support deviation.

For player $3$, the only off-support action is $(1,0,0)$, and its payoff is $0$, whereas the two support actions have payoff
\begin{align*}
    u_3(B) = u_3(C_3) = \frac{7x^*}{2}
    = \frac{\sqrt{8689}-73}{20} > 0.
\end{align*}
Hence player $3$ has no profitable off-support deviation.

For player $1$, the common payoff of the two support actions at $(x^*,y^*,z^*)$ is
\begin{align*}
    u_1(L) = u_1(R) = \frac{229 - \sqrt{8689}}{12}.
\end{align*}
Player $1$ has eight off-support allocations. The following table lists the payoff gap
$u_1(L) - u_1(a_1)$
for each off-support allocation $a_1$.
\[
    \begin{array}{c|c}
        a_1 & u_1(L)-u_1(a_1) \\ \hline
        (0,0,3) & \dfrac{361-\sqrt{8689}}{72} \\[0.7em]
        (0,3,0) & \dfrac{29(\sqrt{8689}-25)}{672} \\[0.7em]
        (1,0,2) & \dfrac{71+\sqrt{8689}}{144} \\[0.7em]
        (1,1,1) & \dfrac{\sqrt{8689}-89}{16}
    \end{array}
    \qquad
    \begin{array}{c|c}
        a_1 & u_1(L)-u_1(a_1) \\ \hline
        (1,2,0) & \dfrac{43\sqrt{8689}-3763}{672} \\[0.7em]
        (2,0,1) & \dfrac{361-\sqrt{8689}}{72} \\[0.7em]
        (2,1,0) & \dfrac{29(\sqrt{8689}-25)}{672} \\[0.7em]
        (3,0,0) & \dfrac{59\sqrt{8689}+7933}{2016}
    \end{array}
\]
Each expression is strictly positive using $93<\sqrt{8689}<94$. Hence every
off-support allocation of player $1$ is strictly worse. Therefore the stated
profile is an NE.

It remains to show that the equilibrium is isolated. Let
\begin{align*}
    F_1(x,y,z) &= 4yz - 6y - 13z + 7,\\
    F_2(x,y,z) &= xz - 3x - z + 1,\\
    F_3(x,y,z) &= 4xy - 33x - 4y + 12.
\end{align*}
The Jacobian matrix of $F = (F_1,F_2,F_3)$ with respect to $(x,y,z)$ is
\begin{align*}
    J(x,y,z) =
        \begin{pmatrix}
        0 & 4z-6 & 4y-13 \\
        z-3 & 0 & x-1 \\
        4y-33 & 4x-4 & 0
        \end{pmatrix}.
\end{align*}
At $(x^*,y^*,z^*)$,
\begin{align*}
    \det J(x^*,y^*,z^*) = -2 \sqrt{8689} \neq 0.
\end{align*}
Since all off-support actions are strictly suboptimal at the stated profile, they remain strictly suboptimal in a sufficiently small neighborhood by continuity.
Hence any nearby equilibrium must use the same support.
On this support, equilibrium requires the three indifference equations
\[
    F_1(x,y,z) = F_2(x,y,z) = F_3(x,y,z) = 0.
\]
The Jacobian of $F$ at $(x^*,y^*,z^*)$ is nonsingular, so by the inverse function theorem, this system has a unique solution in a sufficiently small neighborhood of $(x^*,y^*,z^*)$.
Therefore, no other equilibrium lies in that neighborhood, and the equilibrium is isolated.

\medskip
For completeness, we also show that the same instance has a rational equilibrium.
Let
\[
    x' = \frac{8}{29}, \qquad z' = \frac{1}{9}.
\]
Consider the following mixed strategy profile:
player $1$ plays $L = (0, 1, 2)$ with probability $x'$ and $R = (0, 2, 1)$ with probability $1 - x'$;
player $2$ plays $C_2 = (0, 0, 1)$ with probability $1$;
and player $3$ plays $C_3 = (0, 0, 1)$ with probability $z'$ and $B = (0, 1, 0)$ with probability $1 - z'$.

We verify that this is an NE.
Against player $2$ playing $C_2$ and player $3$ playing $C_3$ with probability $z$, player $1$ obtains
\[
    u_1(L) = \frac{65}{6} + \frac{7z}{2},
    \qquad
    u_1(R) = \frac{34}{3} - z.
\]
At $z = z' = 1/9$, these two payoffs are equal:
\[
    u_1(L) = u_1(R) = \frac{101}{9}.
\]
The payoff gaps between this common payoff and the off-support allocations of player $1$ are as follows:
\[
    \begingroup
    \renewcommand{\arraystretch}{1.35}
    \begin{array}{c|cccccccc}
        a_1
            & (0, 0, 3)
            & (0, 3, 0)
            & (1, 0, 2)
            & (1, 1, 1)
            & (1, 2, 0)
            & (2, 0, 1)
            & (2, 1, 0)
            & (3, 0, 0) \\
        \hline
        u_1(L) - u_1(a_1)
            & \frac{98}{27}
            & \frac{26}{9}
            & \frac{26}{27}
            & \frac{4}{9}
            & \frac{2}{9}
            & \frac{110}{27}
            & \frac{10}{3}
            & \frac{188}{27}
    \end{array}
    \endgroup
\]
Hence player $1$ has no profitable deviation.

For player $2$, at the proposed profile,
\[
    u_2(A) = 3,
    \qquad
    u_2((0, 1, 0)) = 1,
    \qquad
    u_2(C_2) = 4 - 3x' - z' + x' z'
        = \frac{269}{87} > 3.
\]
Thus $C_2$ is a strict best response for player $2$.

Finally, against player $2$ playing $C_2$ and player $1$ playing $L$ with probability $x$, player $3$ obtains
\[
    u_3(B) = \frac{7x}{2},
    \qquad
    u_3(C_3) = \frac{4(1 - x)}{3}.
\]
At $x = x' = 8/29$, these two payoffs are equal:
\[
    u_3(B) = u_3(C_3) = \frac{28}{29}.
\]
The remaining action of player $3$ is $(1, 0, 0)$, which gives payoff $0$.
Therefore player $3$ also has no profitable deviation.

Thus the same Blotto instance admits both the isolated irrational equilibrium described above and a rational equilibrium.

\subsection{Proof of \cref{cor:membership-wsne}}
\label[appendix]{app:proof:membership-wsne}

\begin{proof}
If $\eps \ge 1$, any pure profile suffices, so assume $0 < \eps < 1$.
Set
\[
    \eta = \frac{\eps^2}{8 n}.
\]
By \cref{thm:membership-ne}, computing an $\eta$-NE represented as a flow profile $f$ is in PPAD.
We show how to transform any such output into an $\eps$-WSNE in polynomial time.
For each player $i$, decompose $f_i$ into a convex combination of $s_i$--$t_i$ paths in $D_i$.
Since $D_i$ is acyclic and has polynomially many edges,
this decomposition has polynomial support and can be computed in polynomial time.
The resulting mixed strategy profile over pure Blotto allocations induces the same battlefield marginals as $f$.
Hence it is also an $\eta$-NE of the original Blotto game.

We now implement the pruning step from \cref{lm:wsne2ne}.
First, we compute the best-response value of each player.
Fix player $i$ and the opponents' flow profile $f_{-i}$.
The payoff for playing strategy $g_i$ is
\[
    U_i(g_i, f_{-i}) =  \sum_{e \in E(D_i)} g_i(e) C_{i, e}(f_{-i}).
\]
Therefore the best-response value is
\[
    M_i = \max_{g_i \in \cF_i} \sum_{e \in E(D_i)} g_i(e) C_{i, e}(f_{-i}).
\]
Since this is the maximization of a linear function over the unit-flow polytope,
an optimum is attained by an $s_i$--$t_i$ path.
Thus $M_i$ is the maximum weight of an $s_i$--$t_i$ path in $D_i$,
where edge $e$ has weight $C_{i, e}(f_{-i})$.
Because $D_i$ is acyclic, this value is computed in polynomial time by the standard longest-path dynamic program.

For each player $i$, delete from her support every path whose payoff against $f_{-i}$ is less than
\[
    M_i - \frac{\eps}{2},
\]
and renormalize the remaining distribution.
Since the starting profile is an $\eta$-NE,
the total probability assigned by player $i$ to deleted paths is at most $2 \eta / \eps$.
Otherwise, these paths alone would contribute more than $\eta$ to player $i$'s expected regret.
Since $2 \eta / \eps = \eps/(4 n) < 1$, at least one path remains in every player's support.

Thus the total variation distance between player $i$'s original distribution and her pruned distribution is at most $2 \eta / \eps$.
Fix a player $i$ and a pure allocation of player $i$.
Changing the opponents' strategies from the original profile to the pruned profile changes the payoff of this pure allocation by at most
\[
    \sum_{k \neq i} \frac{2 \eta}{\eps}
    \le \frac{2 n \eta}{\eps}
    = \frac{\eps}{4},
\]
where we use $\eta = \eps^2 / (8 n)$ and the normalization of utilities to lie in $[0, 1]$.

Every path that remains in player $i$'s support had payoff at least $M_i - \eps / 2$ before pruning.
After pruning the opponents, its payoff can decrease by at most $\eps / 4$.
Also, the payoff of any deviation can increase by at most $\eps / 4$.
Therefore every remaining support path is within $\eps$ of a best response to the pruned profile.

Hence the pruned mixed strategy profile is an $\eps$-WSNE.
The construction is polynomial time:
the flow decomposition has polynomial support,
the payoff of each supported path can be evaluated exactly using \cref{lm:compact-payoff},
and the best-response values are computed by the longest-path dynamic program.
All path weights, payoff comparisons, and renormalized probabilities have polynomial bit complexity.
\end{proof}

\subsection{Example of a non-uniform tie-breaking rule with hard payoff evaluation}
\label[appendix]{app:nonuniform-tiebreaking-hardness}

We give a simple example showing why the membership proof for \cref{thm:membership-ne} does not automatically extend to arbitrary succinct non-uniform tie-breaking rules.

Consider a graph $H$ whose vertices are the opponents of player $1$.
We define a non-uniform tie-breaking rule for one battlefield.
Let $\delta = 1/{4 n}$.
If the tied set is a singleton, the unique tied player receives the whole battlefield.
If the tied set $T$ does not contain player $1$, the battlefield is split uniformly among the players in $T$.
If $T$ contains player $1$ and at least one opponent, let $S = T \setminus \{1\}$.
Player $1$ receives $(1/2) + \delta$ if $S$ is an independent set in $H$,
and receives $1/2$ otherwise.
The remaining share is split uniformly among the players in $S$.

This tie-breaking rule is represented succinctly by the graph $H$.
Given a tied set, checking whether $S$ is an independent set can be done in polynomial time by inspecting the edges of $H$.

The rule is also monotone.
If a new player joins a tied set, then the share of every player who was already tied weakly decreases.
For player $1$, this follows because adding more opponents to $S$ cannot turn a non-independent set into an independent set.
Thus player $1$'s share can stay the same or decrease, but it cannot increase.
Now consider an opponent $k \in S$.
If another opponent joins the tied set, then the number of opponents in $S$ increases from $q$ to $q + 1$.
The only case in which the total share of players in $S$ can increase is when $S$ is independent before the
new opponent is added, but is not independent afterward.
In that case the share of each old opponent changes from $\frac{1 / 2 - \delta}{q}$
to $\frac{1 / 2}{q + 1}$.
Since $q \le n - 1$ and $\delta = 1 / (4 n)$, we have
$\frac{1 / 2}{q + 1} \le \frac{1 / 2 - \delta}{q}$.
In all other cases the remaining share does not increase and is split among
more players. Hence every old opponent's share weakly decreases. If player $1$
joins a tied set that previously did not contain her, then the old tied players
also weakly lose share, since before her arrival they split the whole
battlefield uniformly, while afterward they split at most one half of the
battlefield. Therefore the rule is monotone for every player.

We now show that exact payoff evaluation for this rule captures a hard counting
problem. Fix a battlefield and an amount $r$. Suppose that each opponent of
player $1$ independently allocates exactly $r$ with probability $1 / 2$, and
allocates less than $r$ with probability $1 / 2$.
These marginals are realized by legal Blotto strategies: take two battlefields and unit budgets, let player $1$ choose the first battlefield, and let every opponent choose each battlefield with probability $1/2$, so $r = 1$.
Then the random set $S$ of
opponents tied with player $1$ is uniformly distributed over all subsets of the
opponents. Let $\mathrm{IS}(H)$ denote the number of independent sets of $H$.
The conditional expected share of player $1$ is
\[
    h_{1jr}
    =
    2^{-(n - 1)}
    \left(
        1
        +
        \sum_{\emptyset \neq S \subseteq [n] \setminus \{1\}}
        \left(
            \frac{1}{2}
            +
            \delta \cdot \mathbf{1}\{S \text{ is independent in } H\}
        \right)
    \right).
\]
Equivalently,
\[
    h_{1jr}
    =
    2^{-(n - 1)}
    \left(
        1
        +
        \frac{2^{n - 1} - 1}{2}
        +
        \delta \bigl(\mathrm{IS}(H) - 1\bigr)
    \right),
\]
where the $-1$ subtracts the empty independent set, since the singleton tied set $\{1\}$ was handled separately.
Therefore exact evaluation of $h_{1jr}$ allows one to recover $\mathrm{IS}(H)$.

Counting independent sets is a standard $\#\mathrm{P}$-complete problem.
Hence a polynomial-time algorithm for exact expected-share evaluation for all succinct monotone tie-breaking rules would give a polynomial-time algorithm for counting independent sets.
Uniform tie-breaking avoids this obstruction because the share depends only on the number of tied opponents, not on their identities, and this symmetry is exactly what the dynamic program in \cref{lm:compact-payoff} uses.

\section{Uniform Tie-breaking and Single Resource}
\label[appendix]{app:singleton-congestion}

In this subsection, we justify the claim from \cref{sec:hardness} that the case $B_i=1$ for every player $i$ is tractable. Since every player has one resource, a pure strategy of player $i$ is simply the choice of a battlefield. For a pure profile $\ba$, let $n_j(\ba)$ denote the number of players who choose battlefield $j$, and let $n_j^{-i}(\ba_{-i})$ denote the number of players other than $i$ who choose battlefield $j$.

Fix the actions $\ba_{-i}$ of all players other than $i$. Let
\[
    E_i(\ba_{-i}) = \{ j : n_j^{-i}(\ba_{-i}) = 0 \}
\]
be the set of battlefields that receive no resource from the other players. Under standard uniform tie-breaking, player $i$ receives payoff $v_{ij}/n$ from every battlefield $j \in E_i(\ba_{-i})$ that she does not choose, because all players are tied at zero on that battlefield. Define
\[
    C_i(\ba_{-i}) = \sum_{j \in E_i(\ba_{-i})} \frac{v_{ij}}{n}.
\]
This term depends on the other players' actions, but not on the battlefield chosen by player $i$.

Now suppose player $i$ chooses battlefield $j$. If $j \notin E_i(\ba_{-i})$, then the battlefield already receives at least one resource from another player, and player $i$ receives payoff
\[
    \frac{v_{ij}}{n_j^{-i}(\ba_{-i}) + 1}
\]
from battlefield $j$. If $j \in E_i(\ba_{-i})$, then the term $v_{ij}/n$ has already been included in the baseline $C_i(\ba_{-i})$, but choosing $j$ makes player $i$ the unique winner of battlefield $j$ and raises her payoff from that battlefield to $v_{ij}$. Thus, the additional payoff from choosing such a battlefield is
\[
 v_{ij} - \frac{v_{ij}}{n} = \left( 1 - \frac{1}{n} \right) v_{ij}.
\]
Therefore, for every choice of battlefield $j$,
\[
    u_i(j,\ba_{-i}) = C_i(\ba_{-i}) + g_{ij}\bigl( n_j^{-i}(\ba_{-i} ) + 1 \bigr),
\]
where
\[
    g_{ij}(k) = \begin{cases}
        \left(1-\frac{1}{n}\right)v_{ij}, & k=1,\\
        \frac{v_{ij}}{k}, & k\ge 2.
    \end{cases}
\]
Since $C_i(\ba_{-i})$ is independent of $j$, it does not affect player $i$'s best responses. Hence, the best-response structure is exactly that of a singleton congestion game with player-specific payoff functions $g_{ij}$: each player chooses one battlefield, and the payoff from choosing battlefield $j$ depends only on the identity of the player, the battlefield, and the total load on that battlefield.

It follows that pure Nash equilibria of this singleton congestion game are precisely pure Nash equilibria of the one-resource Blotto game. Player-specific singleton congestion games are a special case of player-specific matroid congestion games, for which pure Nash equilibria exist and can be computed in polynomial time~\cite{milchtaich1996congestion,ackermann2009pure}. Thus, a pure Nash equilibrium of the one-resource Blotto game can be computed in polynomial time.

\clearpage
\bibliography{ref}

\end{document}